\documentclass[11pt]{article}
\usepackage[shortlabels]{enumitem}
\usepackage{titling}
\usepackage{stmaryrd}
\usepackage{palatino} 
\usepackage{mathpazo}
\usepackage{braket}
\usepackage{amsfonts}
\usepackage{amssymb}
\usepackage{amsmath}
\usepackage{mathtools}
\usepackage{mathrsfs} 
\usepackage{mathabx}
\usepackage{bbm}
\usepackage{latexsym}
\usepackage{amsthm}
\usepackage{complexity}
\usepackage{algorithm}
\usepackage[noend]{algpseudocode}
\usepackage[usenames]{color}
\usepackage{hyperref, etoolbox}
\usepackage{subcaption}
\usepackage{verbatim}
\usepackage{fontenc}
\usepackage{ragged2e, graphicx, pifont}
\usepackage{indentfirst}
\usepackage{thm-restate}	
\usepackage{caption}
\usepackage{scalerel}
\usepackage{geometry}
\usepackage{fullpage}
\hypersetup{
colorlinks = true,
citecolor= blue
}
\usepackage{makecell}
\usepackage{tikz}
\usepackage{fancybox}
\usepackage{tikz-cd}
\usetikzlibrary{decorations.pathmorphing}
\usetikzlibrary{arrows.meta}

\usepackage[framemethod=TikZ]{mdframed}
\usepackage{extarrows}
\usepackage{arydshln}
\usepackage{soul}
\usepackage{titlesec}

\usepackage{multirow}
\usepackage{moresize}
\usepackage{crossreftools}
\usepackage{relsize}
\usepackage{hyperref}
\usepackage{cleveref}

\newcommand{\I}{\mathrm{I}}

\newcommand{\cmark}{\ding{55}}

\newcommand{\defeq}{{\text{}:=}}

\newcommand{\eH}{\mathrm{H}}
\newcommand{\eps}{\varepsilon}

\newcommand{\brak}[1]{\{#1\}}

\newcommand{\clA}{\mathcal{A}}
\newcommand{\clB}{\mathcal{B}}

\newcommand{\emptystring}{{\scaleto{\emptyset}{8pt}}}

\renewcommand{\C}{\mathrm{C}}

\newcommand {\diver} [2] {{\D}(#1 \| #2)}
\newcommand{\vall}{\mathrm{val}}
\newcommand{\Id}{\mathbb{1}}

\newcommand{\bbF}{\mathbb{F}}

\newcommand{\cbra}[1]{[{#1}]}

\renewcommand{\C}{\mathrm{C}}
\renewcommand{\E}{\mathbb{E}}

\newcommand{\tab}{\hspace{2cm}}

\newcommand{\Tr}{\mathrm{Tr}}
\newcommand{\wt}{\mathrm{wt}}
\newcommand{\dist}{\mathrm{dist}}

\renewcommand{\poly}{\mathrm{poly}}
\renewcommand{\polylog}{\mathrm{polylog}}

\newcommand{\Sim}{\textsc{Sim}}
\newcommand{\ideal}{\textsc{ideal}}

\newcommand{\val}{\mathrm{val}}
\newcommand{\CNOT}{\mathsf{CNOT}}

\newcommand{\prt}[2]{\mathcal{P}_{#1 \vert #2}}

\newcommand{\boldcheckmark}{\checkmark}

\newcommand{\Pf}{\mathcal{P}_f}
\newcommand{\Pg}{\mathcal{P}_g}

\newcommand{\Pgf}{\prt{g}{f}}
\newcommand{\priorprt}{\mathcal{P}^{\textsc{pri}}}
\newcommand{\afterprt}{\mathcal{P}^{\textsc{aft}}}
\allowdisplaybreaks

\usepackage{crossreftools}
\usepackage{circuitikz}

\newcommand{\ketbra}[2]{\ket{#1}\!\!\bra{#2}}

\newcommand{\kb}[1]{\ketbra{#1}{#1}}

\newcommand{\norm}[1]{\ensuremath{\left\lVert #1 \right\rVert}}

\newcommand{\wtt}[1]{\widetilde{#1}}
\newcommand{\wht}[1]{\widehat{#1}}
\newtheorem{theorem}{Theorem}
\newtheorem{lemma}{Lemma}

\newtheorem{fact}{Fact}
\newtheorem{claim}{Claim}

\theoremstyle{definition}
\newtheorem{definition}{Definition}

\newcommand{\negl}{\mathsf{negl}}

\newcommand{\MS}{\textup{MS}}

\newcommand{\test}{\textsc{test}}

\newcommand{\hmin}{\mathrm{H}_{\min}}
\newcommand{\bob}{\mathsf{Bob}}
\newcommand{\alice}{\mathsf{Alice}}
\newcommand{\charlie}{\mathsf{Charlie}}
\newcommand{\ver}{\mathsf{Ver}}
\newcommand{\dave}{\mathsf{Dave}}
\newcommand{\good}{\mathsf{Good}}
\newcommand{\onenorm}[2]{\left\Vert #1 - #2\right\Vert_1}

\newcommand{\junkstate}{\ket{\texttt{junk}}}
\newcommand{\junkstatedagger}{\bra{\texttt{junk}}}
\newcommand{\ext}{\mathsf{Ext}}

\newcommand{\mcC}{\mathcal{C}}

\newcommand{\mcP}{\mathcal{P}}
\newcommand{\mcQ}{\mathcal{Q}}

\newcommand{\mcL}{\mathcal{L}}

\newcommand{\device}{\mathbf{D}}

\newcommand{\suppress}[1]{}

\usepackage{authblk}

\allowdisplaybreaks

\makeatletter
\newcommand{\optionaldesc}[2]{%
  \phantomsection
  #1\protected@edef\@currentlabel{#1}\label{#2}%
}
\makeatother

\title{\vspace{-1in}A robust and composable device-independent protocol for oblivious transfer using (fully) untrusted quantum devices in the bounded storage model}

\date{Appeared at QCrypt, 2025. ArXiv:2404.11283.}

\author[1]{Rishabh Batra}
\author[2]{Sayantan Chakraborty}
\author[1,3,4]{Rahul Jain}
\author[5]{Upendra Kapshikar}

\affil[1]{{\small Centre for Quantum Technologies, Singapore}}
\affil[2]{{\small Département d'informatique et de recherche opérationnelle (DIRO), Université de Montréal}}
\affil[3]{{\small Department of Computer Science, 
  National University of Singapore}}
\affil[4]{{\small MajuLab, UMI 3654, Singapore}}
\affil[5]{{\small QUASAR and Department of Mathematics and Statistics, University of Ottawa}}
  
\begin{document}
\maketitle

\thispagestyle{empty}

\vspace{-0.5in}

\begin{abstract}
We present a robust and composable device-independent (DI) quantum protocol between two parties for oblivious transfer (OT) using Magic Square devices in the bounded storage model \cite{DFR,DFSS_1} in which the (honest and cheating) devices and parties have no long-term quantum memory. After a fixed constant (real-world) time interval, referred to as \textbf{DELAY}, the quantum states decohere completely. The adversary (cheating party), with full control over the devices, is allowed joint (non-IID) quantum operations on the devices, and there are no time and space complexity bounds placed on its powers. The running time of the honest parties is $\polylog(\lambda)$ (where $\lambda$ is the security parameter). Our protocol has negligible (in $\lambda$) correctness and security errors and can be implemented in the NISQ (Noisy Intermediate Scale Quantum) era. By robustness, we mean that our protocol is correct even when devices are slightly off (by a small constant) from their ideal specification. This is an important property since small manufacturing errors in the real-world devices are inevitable.

Our protocol is sequentially composable and, hence, can be used as a building block to construct larger protocols (including DI bit-commitment and DI secure multi-party computation) while still preserving correctness and security guarantees. None of the known DI protocols for OT in the literature are secure against arbitrary (non-IID) devices and provide simulator-based (composable) security. This was a major open question in device-independent two-party distrustful cryptography, which we resolved.

We prove a parallel repetition theorem for a certain class of entangled games with a hybrid (quantum-classical) strategy. This parallel repetition allows us to show min-entropy guarantees on certain random variables, which helps in proving the security of our protocol. The hybrid strategy helps to incorporate \textbf{DELAY} in our protocol. This parallel repetition theorem is a main technical contribution of our work.  Since our games use hybrid strategies and the inputs to our games are not independent, we use a novel combination of ideas from previous works showing parallel repetition of classical games~\cite{Raz,holenstein}, quantum games~\cite{JPY,JMS_DIQKD,Jain_Kundu_DIQKD}, and anchored games~\cite{anchoring3,anchoring_2}.

Although we present security proof for protocols in the bounded storage model with no long-term quantum memory (after \textbf{DELAY}), we can extend our results, along the lines of \cite{DFR}, to incorporate linear (in the number of devices) long-term quantum memory.

\end{abstract}

\clearpage
\pagebreak
{
\thispagestyle{empty}
    \hypersetup{hidelinks}
    \tableofcontents
}
\pagebreak
\pagenumbering{arabic}

\section{Introduction}

The advent of quantum technologies often poses security risks for many existing classical cryptographic protocols. At the same time, it provides new ways of designing cryptographic protocols using quantum devices. However, in several situations, parties willing to engage in quantum protocols may not be able to manufacture quantum devices on their own and may have to source them from an external party. For example, consider the following real-world use case. 

Two (potentially distrustful) banks wish to conduct a cryptographic protocol between themselves (potentially for exchanging sensitive financial information). They have no quantum labs and cannot perform quantum operations on their own. They source quantum devices from a third-party vendor and interact with these (potentially distrustful) devices using classical inputs and classical outputs, the so-called {\em device-independent} (DI) setting. Since the banks are (potentially) distrustful, there could be collusion between an adversarial (cheating) bank and the vendor. Even when all the parties are honest, due to some small manufacturing errors expected in the real world, the vendor may provide slightly faulty devices to the banks, in which case, the protocol should still work. The protocol should be potentially implementable using the devices available in the NISQ (Noisy Intermediate Scale Quantum) era. The protocol should be composable so that it can be used as a building block to construct larger cryptographic protocols. The assumptions placed on the computational powers of the adversarial bank (potentially colluding with the vendor) should be as minimal as possible.

The key motivating question is:\vspace{-2mm}
\begin{quote}
\centering \emph{Can we design cryptographic protocols in this aforementioned setting?} \end{quote}
\subsection*{Oblivious transfer}
{\em Oblivious transfer} (OT) is one of the most important two-party distrustful cryptographic primitives. It can be used to securely compute any multiparty functionality~\cite{Killian_OT_implies_everything}, including bit-commitment (BC). Informally, in a protocol for OT, $\alice$ gets as input two bits $(s_0, s_1)$ and $\bob$ gets one bit $b$. At the end of the protocol, $\bob$ should get to know $s_b$ and should get no information about $s_{1-b}$ while $\alice$ should get no information about $b$.

It is well known that information-theoretically secure protocols for OT are not possible even if $\alice$ and $\bob$ have access to quantum devices~\cite{Mayers_BC_Impossible, Lo_imposibility_stuff, LO_Chahu_BC_Impossible}. Therefore, protocols for OT are considered under different bounds on the computational powers of the adversaries. A popular and extensively studied setting is that of time-bounded adversaries, for example~\cite{OT_time1, Morimae_2022, time3, brakerski2022computational}. Damg\r{a}rd, Fehr, Salvail, and Schaffner~\cite{DFSS_1} and Damg\r{a}rd, Fehr, Renner, Salvail, and Schaffner~\cite{DFR} considered the \textit{bounded quantum-storage model (BQSM) }for the adversaries and exhibited secure protocols for OT and BC in this setting. In this model, initially, the adversary is allowed arbitrary quantum memory and quantum computation power. After a certain point in time (referred to as the \textbf{DELAY}), the memory bound applies, which means that only a fixed number of qubits can be stored, and all the remaining qubits are measured, although there is no bound on the computing power of the adversary. Another popular model is the \textit{noisy-storage model} used in Konig, Wehner, and Wullschleger~\cite{Konig_Wehner_Noisy_Storage} (see also \cite{Noisy_storage}). Similar to the bounded storage model, after the \textbf{DELAY}, there is only a fixed amount of quantum memory that the adversary can retain, but this memory is \textit{noisy} and decays with time. \cite{Konig_Wehner_Noisy_Storage} also shows secure protocols for OT and BC in this model. We note that the adversary in the bounded storage model is stronger than the noisy-storage model, as in the bounded storage model, there is no decoherence in the quantum memory of the adversary after the \textbf{DELAY}. We note that the protocols proposed in \cite{DFR},~\cite{DFSS_1}, and \cite{Konig_Wehner_Noisy_Storage} were not in the DI setting. 

A common assumption that is made in the DI setting, for example, in the works of Broadbent and Yuen \cite{Broadbent} and Kaniewski and Wehner \cite{KW16}, is that the devices are independent and identically distributed (referred to as the IID assumption). This assumption often simplifies security proofs. However, in cryptographic settings in the real world, the adversary could very well be colluding with the vendor to source the devices as per the adversary's specifications, in which case the IID assumption on the devices will fail to hold. Designing DI protocols for OT, which are secure without the IID assumption on the devices, has been a major open problem~\cite{KW16}. 

\subsection*{Our result}
We present a protocol for OT between $\alice$ and $\bob$, with the following properties and assumptions. Let $\lambda$ be the security parameter. 

{\bf Properties:}
\begin{enumerate}

\item {\bf Device independence:} Honest $\alice$ and $\bob$ are classical parties which use {\em Magic Square} (MS) quantum devices to execute the protocol. They provide classical inputs to these quantum devices and obtain classical outputs. Apart from this, they perform other classical computations. 

\item {\bf Efficiency:} Honest $\alice$ and $\bob$ run in time $\polylog(\lambda)$ and use $\polylog(\lambda)$ MS devices. 

\item {\bf Implementable in NISQ:} Since MS devices can be assumed to be implementable in the NISQ era and the other operations of honest $\alice$ and $\bob$ are classical, our protocol can potentially be implemented in the NISQ era. Indeed, our protocol requires honest parties to store several EPR pairs for a short time. This is consistent with our assumption of BQSM, where there is no ``long-term" quantum memory, but quantum states can be stored for a short time. 

\textbf{Real-world scenario:} Suppose the native storage time of an EPR pair (with high fidelity) is $1$ microsecond, and it is (highly) decohered after $1$ second (these numbers may depend on specific experimental implementation). It is possible that several hundred EPR pairs may be generated, and honest Alice and Bob can perform their operations in $1$ microsecond (this would suffice for several hundred bits of security). We would use DELAY = $1$ second in the protocol. This is how the competing requirements of keeping several EPR pairs for a short time and having no long-term quantum memory can be satisfied in NISQ.

\item {\bf Negligible correctness and security errors:} Our protocol has negligible (in $\lambda$) correctness and security errors.

\item {\bf Robustness:} The protocol between honest $\alice$ and $\bob$ has negligible correctness error even when each of the $\polylog(\lambda)$ MS devices that they use is a constant away from an ideal MS device (see Definition \ref{def:device}). This is a very useful feature for real-world protocols where small manufacturing errors on quantum devices are almost inevitable. 

\item {\bf Composable security:} We present a framework for composing (two-party) quantum cryptographic protocols and show a (sequential) composition theorem for quantum protocols using (augmented) {\em simulator}-based security. In particular, we show that if we have two protocols with augmented security, then their (sequential) combination also has augmented security. We show that our OT protocol is augmented-secure and hence can be composed with any other augmented-secure protocol to build larger cryptographic protocols. Since OT can be used to obtain any secure multi-party functionality, our protocol implies protocols for all secure multi-party functionalities (e.g, bit-commitment) in the bounded storage model.  \textbf{None of the earlier DI two-party cryptographic protocols for OT are composable or prove simulator-based security.}

{\bf Remark:} Due to the $\textbf{DELAY}$ assumption, when our protocol is used as a subroutine in an outer protocol, the parties source fresh devices at the beginning of our protocol. For our protocol to be composable, we require that the quantum state of the outer protocol (even with adversarial parties) completely decoheres before the inner protocol for OT starts. This is a reasonable assumption in the bounded storage model and can be enforced by an additional \textbf{DELAY} at the beginning of the protocol for OT, which decoheres the end state of the outer protocol. Due to this decoherence, we don't face any quantum rewinding issues while proving composable security.

\item {\bf General (non-IID) attacks:} An adversarial party can design all the devices used in the protocol together (the combined states and measurements) before the protocol starts.   
\textbf{No other known DI protocols for OT allow for adversaries with general (non-IID) attacks.}

{Note that the no-signalling assumption (see below) implies that the adversary cannot change the behaviour of the share of devices held by the honest party, once the protocol starts.   }
\end{enumerate}
{\bf Assumptions:} 
\begin{enumerate}

\item {\bf Bounded quantum storage:} The honest and adversarial quantum devices and the adversarial parties ($\alice^*$ and $\bob^*$) have no long-term quantum memory. The quantum memory completely decoheres after a fixed constant (real-world) time interval, referred to as \textbf{DELAY}\footnote{If no input is provided to a device before \textbf{DELAY}, it is assumed to be 0.}. This is a justified assumption in the NISQ era, where the devices are unlikely to have a long-term quantum memory. This corresponds to the bounded quantum-storage model, with no long-term quantum memory \cite{DFR,DFSS_1}.

Although we present our proofs in this model, we state that we can extend our results along the lines of~\cite{DFR}, to incorporate linear (in the number of devices) quantum memory after \textbf{DELAY} for the (inner) OT protocol with the following brief justification. In our proof, we show a linear lower bound on the min-entropy of certain random variables conditioned on the adversary's side information. Our proof currently allows for classical side information by the adversary (after \textbf{DELAY}).  When allowing for linear quantum information, this reduces the conditional min-entropy by a linear term, which can be absorbed since our original lower bound is linear as well.

\item {\bf No-signalling devices:} There is no communication between $\alice$ and $\bob$'s parts of (a shared) device after inputs are provided to it. 

\textbf{Remark:} A desirable property of any quantum cryptographic protocol is leakage-resilience, that is, the protocol should be secure even when some communication happens between the devices after inputs are provided to them and before outputs are produced by them. We note that almost all DI protocols in the literature for well-studied primitives like Quantum Key Distribution also place the no-signalling assumption~\cite{DIQKD_5, Acin_2006, VV_DIQKD}. Only very recently, some leakage-resilient protocols for DI QKD have been presented \cite{Jain_Kundu_DIQKD}. We leave investigations into leakage-resilient DI protocols for OT as future work. Note that this leakage is related to the communication that happens between the quantum devices after inputs are provided to them and before outputs are produced by them, i.e., by the quantum components in the protocol. No leakage by classical components is a standard (and necessary) assumption in distrustful cryptography.

\item {\bf Complexity bounds on adversaries:} The adversarial parties can perform quantum operations. They have no bounds on quantum time complexity and classical space complexity. They have no bounds on quantum space usage before \textbf{DELAY}.

\item {\bf Trusted private randomenss:} Honest $\alice$ and $\bob$ have access to trusted private randomness.

\item {\bf Access to clocks:} Honest $\alice$ and $\bob$ have access to  synchronized clocks. In each round of the protocol, if a message is not received from one of the parties, the other party outputs the {\em abort} symbol $\bot$.

\item {\bf Classical communication:} $\alice$ and $\bob$ have access to an ideal (non-noisy) classical communication channel between them. 
\end{enumerate}

Our main result is:
\begin{theorem} [Informal]\label{cor:iid_OT_intro}
     There exists a protocol for OT with the aforementioned properties and assumptions.
\end{theorem}

\subsubsection*{Previous works}
\cite{DFSS_1,DFR, Konig_Wehner_Noisy_Storage} present secure protocols for OT and BC in the {bounded/noisy quantum memory} model. These protocols are in the non-DI setting where $\alice$  and $\bob$ must trust their preparation and measurements of the quantum state to guarantee the security of the protocol.

\cite{KW16} presents DI protocols for the task of {\em Weak String Erasure} in the bounded quantum memory and noisy-storage models with the IID assumption.
The work of \cite{KW16} also implies DI security in the non-IID setting under an additional assumption that the adversary makes only sequential attacks. The techniques used in \cite{KW16} also imply the security of a DI protocol for OT under the IID assumption. However, the model for device independence used in that paper is quite different from the one we consider. 
In \cite{KW16}, the sender can test her devices, but the receiver is unable to do so. \cite{KW16} doesn't claim robustness. \cite{KW16} uses a \emph{trusted} device in the protocol execution, called a `switch'. They argue that such a device can be manufactured easily by the honest party and hence can be trusted. In our protocol, we do not require any such trusted device. Our protocol is sequentially composable, whereas \cite{KW16} makes no such claims.
%Note that our techniques for composability rely on the messages between Alice and Bob being classical, hence, they would not help in proving security for the protocol in \cite{KW16}.

Kundu, Sikora, and Tan~\cite{srijita_xot} present a DI protocol for OT (also using MS devices)  where an adversary can cheat with a constant (bounded away from $1$) probability, as opposed to negligibly small in the security parameter. This presents a problem with multiple uses of their protocol as a subroutine in a larger cryptographic protocol. Their protocol does not impose any additional computational restrictions on the adversary. They don't claim that their protocol is composable. The design and analysis of the protocol presented in that paper differ significantly from ours.

\cite{Broadbent} presents a DI protocol for OT under the IID assumption on devices based on the hardness of solving the {\em Learning With Errors} (LWE) problem~\cite{Regev_LWE}.
They allow some leakage between the devices during the protocol. However, this protocol is not shown to be composable.

Some works have proposed frameworks for building composable quantum protocols~{\cite{unruh, Unruh_1, Unruh_2, Müller-Quade_2009, Wehner_composibile, kundu_tan}}. Kundu and Tan~\cite{kundu_tan} present composable DI protocols for {\em certified deletion}.
 \cite{Unruh_2, Unruh_1} provide a composable non-DI protocol for OT.
 
Table~\ref{table:main} presents the comparison of our work to prior works on different parameters (here $\mathfrak{D}^*$ denotes the set of all devices with arbitrary (non-IID) behavior). In the table, we have different rows comparing whether quantum communication in the protocol is needed or not, correctness and security when using arbitrary (non-IID) devices, self-testability of devices, robustness of the protocol, etc.

\begin{table}[!h]
\centering
\begin{tabular}{|>{\centering\arraybackslash}m{2.8cm}
                |>{\centering\arraybackslash}m{2.8cm}
                |>{\centering\arraybackslash}m{2cm}
                |>{\centering\arraybackslash}m{2cm}
                |>{\centering\arraybackslash}m{2cm}|}
\hline

 & \textbf{\cite{DFR}, \cite{DFSS_1}, \cite{Konig_Wehner_Noisy_Storage}}
 & \textbf{\cite{srijita_xot}}
 & \textbf{\cite{KW16}}
 & \textbf{This work} \\
\hline

\textbf{Primitive}
& OT and BC
& XOR OT
& WSE and BC
& OT \\
\hline

\textbf{Composability}
& \cmark
& \cmark
& \cmark
& \boldcheckmark \\
\hline

\textbf{Device independence}
& \cmark
& \boldcheckmark
& \boldcheckmark
& \boldcheckmark \\
\hline

\textbf{Quantum memory}
& BQSM or \newline Noisy-storage
& No assumptions
& BQSM or \newline Noisy-storage
& BQSM or \newline Noisy-storage \\
\hline

\textbf{Leakage between devices}
& Not Applicable
& Not Allowed
& Not Allowed
& Not Allowed \\
\hline

\textbf{Completeness and soundness error}
& $\negl$
& constant
& $\negl$
& $\negl$ \\
\hline

\textbf{Robustness}
& \boldcheckmark
& \cmark
& \cmark
& \boldcheckmark \\
\hline

\textbf{Quantum Communication needed between parties} 
& \boldcheckmark
& \cmark
& \boldcheckmark
& \cmark \\
\hline

\textbf{Self-testability of devices}
& Not Applicable
& \boldcheckmark
& \cmark
& \boldcheckmark \\
\hline

\textbf{Security against $\mathfrak{D}^*$}
& Not Applicable
& Not Applicable
& \cmark
& \boldcheckmark \\
\hline

\end{tabular}
\caption{Comparison with other works}
\label{table:main}
\end{table}

\subsubsection*{Novelty in techniques}

The first non-IID security proofs for DI QKD (after a considerable literature with the IID assumption) were very well received \cite{VV_DIQKD}. Since then, removing the IID assumption for any two-party distrustful primitives was a major open question in the DI literature, which we settle for DI OT. This required novel ideas of using parallel repetition (of anchored games with \textbf{DELAY}) to prove the security of our protocol. For DI QKD, a key tool used in the literature is the Entropy Accumulation Theorem (EAT) (see \cite{Arnon_Friedman_DIQKD_2}). Since for OT, one of the parties may not cooperate in playing non-local games, it is not a priori clear how EAT can be used in such a setting (unlike QKD), further emphasising the utility of security proofs via parallel repetition. 

We prove a parallel repetition theorem for a certain class of entangled games with a hybrid (quantum-classical) strategy. This parallel repetition allows us to show min-entropy guarantees on certain random variables, which helps in proving the security of our protocol. Parallel repetition results for games have been used to argue the security of quantum cryptographic primitives before, for example, DI QKD \cite{JMS_DIQKD, vidick_paralleldiqkd,  Jain_Kundu_DIQKD}. In DI QKD, $\alice$ and $\bob$  are two trusted parties that want to generate a shared secret string and can trust each other's input distributions when trying to test the devices (say, using non-local games). However, in the setting of OT, the cheating party may not be using the correct distribution on the inputs to the (untrusted) devices, and hence, self-testing in this scenario is not as straightforward. 

To show parallel repetition of our game, we use arguments similar to  Jain, Pereszlényi, and Yao~\cite{JPY} (which in turn uses techniques from Jain, Radhakrishnan, and Sen~\cite{JRS}). However, we cannot directly apply the ideas in \cite{JPY} because of the following reasons.
\begin{itemize}
    \vspace{-1mm} \item $\cite{JPY}$ deals with games in which input distributions are independent. In our case, the inputs to the two parties are not independent. We use the idea of anchored distributions to deal with this non-product nature of the input distribution, similar to \cite{anchoring3, anchoring_2}.
    \vspace{-1mm} \item Our protocol (and security game) is in the bounded storage model where, after  \textbf{DELAY}, the quantum state decoheres completely. This is different from the usual setting of non-local games, where quantum parties can do any arbitrary quantum operations to win the games.   We define a restricted (hybrid quantum-classical) strategy in which \textbf{DELAY} corresponds to a mandatory $\CNOT$ gate between the quantum state of a party (say $\varphi^N$) and environment registers, $\ket{0}^{N'}$, after which the parties do not have access to $N'$. Before the $\CNOT$, the parties can do any arbitrary quantum operation on $\varphi^N$.  The $\CNOT$ corresponds to a computational basis measurement that simulates the quantum state's decoherence after \textbf{DELAY}. After \textbf{DELAY}, the parties perform classical operations on their registers. 
    
 To show the parallel repetition theorem for such hybrid games, we use a novel combination of ideas from previous works showing parallel repetition of classical games~\cite{Raz,holenstein}, quantum games~\cite{JPY,JMS_DIQKD,Jain_Kundu_DIQKD}, and anchored games~\cite{anchoring3,anchoring_2}.

 % \vspace{-1mm} \item We require additional quantum information tools to deal with entangled games that are inspired by the work of Jain, Radhakrishnan, and Sen~\cite{JRS}. 
\end{itemize}
Using these novel techniques, we are able to show the security of our protocol, even when the adversary is allowed full control over the manufacturing of the devices (states and measurements) before the protocol starts (i.e., non-IID quantum attacks). Note that although our security proof is via a parallel repetition theorem for a class of games, it also works for sequential strategies, which are a special case of arbitrary parallel cheating strategies.

\subsection*{Proof overview}

 In an ideal Magic Square (MS) device, $\alice$ and $\bob$ provide inputs $x\in\lbrace 0,1,2\rbrace$ and $y\in\lbrace 0,1,2\rbrace$ respectively.
$\alice$ receives output $a\in\mathcal{A} = \lbrace 000,011,101,110\rbrace$ (parity $0$ strings) and $\bob$ receives output $b\in\mathcal{B}= \lbrace 001,010,100,111\rbrace$ (parity $1$ strings) such that $a_y=b_x$. In other words, both $\alice$ and $\bob$ receive $3$-bit strings with the guarantee that the {\em common-bit} is the same (see Definition \ref{def:MS_Game}). Our protocol for device-independent OT uses $n$ (= $\polylog(\lambda)$) MS devices as given in Table \ref{table:OT}. We denote the inputs to these devices as $X=X_1, X_2\dots X_n$ (and similarly for $Y$) and outputs as $A=A_1, A_2\dots A_n$ (and similarly for $B$). We use $A(X)$ to denote $A_1(X_1),A_2(X_2) \dots A_n(X_n)$. For a bit $b\in \brak{0,1}$, we use $A(b)$ to denote $A(b^n)$, i.e., we select the $b$-th column (or bit) for each of $\alice$'s outputs (and similarly for $\bob$).

\begin{table}[!h]
        \centering
            \begin{tabular}{  l c r }
            \hline \\
            $\alice$ (input $S_0,S_1$) &  & $\bob$ (input $D$) \\
            \hline \\ 
             
 \hdashline \\

& \large{\textsc{Setup for Test Phase}} &  \\

\hdashline  \\\\
$X\leftarrow \brak{0,1,2}^{n}$
& &   $Y=D^n$  \\ 
\resizebox{0.12\textwidth}{!}{%
\begin{circuitikz}
\tikzstyle{every node}=[font=\normalsize]
\draw (6.25,12.75) rectangle (7.5,11.75);
\node [font=\large] at (6.8,12.25) {$\MS^A$};
\draw [->, >=Stealth, dashed] (5.5,12.5) -- (6.25,12.5);
\draw [->, >=Stealth, dashed] (6.25,12) -- (5.5,12);
\node [font=\large] at (5.2,12.5) {$X$};
\node [font=\large] at (5.2,12) {$A$};
\end{circuitikz}
} && 
\resizebox{0.12\textwidth}{!}{%
\begin{circuitikz}
\tikzstyle{every node}=[font=\normalsize]
\draw (6.25,12.75) rectangle (7.5,11.75);
\node [font=\large] at (6.8,12.25) {$\MS^B$};
\draw [->, >=Stealth, dashed] (5.5,12.5) -- (6.25,12.5);
\draw [->, >=Stealth, dashed] (6.25,12) -- (5.5,12);
\node [font=\large] at (5.2,12.5) {$Y$};
\node [font=\large] at (5.2,12) {$B$};
\end{circuitikz}
} \\ 
&\textbf{DELAY} \\  \\
\hdashline \\

& \large{\textsc{Test Phase Checks}} &  \\

\hdashline  \\
$T_0,T_1\leftarrow_U\lbrace0,1\rbrace^{O(\polylog(n))}$ & & \\$R_{0} ={A}\left(0\right)$, $R_{1} ={A}\left(1\right)$  & & \\
    $W_{0} = \eH\left(R_{0}\right)$, $W_{1} = \eH\left(R_{1}\right)$ \\
        $F_0 = S_0 \oplus \ext\left(R_0,T_0\right)$ & & \\
          $F_1 = S_1 \oplus \ext\left(R_{1},T_1\right)$  
     & $\xrightarrow{{F_0,F_1,T_0, T_1, {X}}, W_0,W_1}$ & \\
     &&  $R^\prime ={B}\left({X}\right)$ \tab \hspace{1cm} \\
&& $E^\prime = \textsc{SyndDec}\left(\eH(R^\prime) - W_{D}\right)$  \\
  && $L^\prime = R^\prime+E^\prime$ \tab  \quad \quad \tab  \\
       && Output $F_D \oplus \ext(L^\prime,T_D)$\hspace{0.9cm} \\
 \hdashline \\      \end{tabular} 
        \caption{Simplified protocol for OT}
        \label{prot:OT-simple} 
    \end{table}
   
The proof idea for the protocol is as follows: $\alice$ receives as input two bits $S_0$ and $S_1$, and $\bob$ receives the choice bit $D$. At the beginning of the protocol, there is a test phase where $\alice$ selects a random (and large enough) subset of devices $S_{test}$ to test the MS predicate. $[n]$ in \Cref{prot:OT-simple} represents the non-test devices. She generates both her and $\bob$'s inputs for the test devices and sends $\bob$ the location of the test indices $S_{test}$ and his inputs. $\bob$ needs to send to $\alice$ the outputs of the test devices, which allows $\alice$ to test the MS predicate for the test subset. 
Note that the test phase is not pre-appended to the protocol but is interwoven with the protocol. In particular, the test checks are carried out after the \textbf{DELAY}, at which point the outputs used in the actual protocol have already been generated. 

After the test subset input and outputs are generated, during an honest execution of the protocol, $\alice$ generates $X$ uniformly from $\brak{0,1,2}^n$ and inputs them into her non-test MS devices. $\bob$, on the other hand, inputs $D$ into all of his non-test MS devices, i.e., he sets $Y=D^n$. The main idea is that $\alice$ will use random bits extracted from the first two columns of the output of her MS devices, $A$, to obfuscate $S_0$ and $S_1$ if the test checks pass after the \textbf{DELAY}. For this, she samples independently and uniformly at random two strings $T_0$ and $T_1$ which act as seeds for a strong seeded-extractor $\ext(\cdot,\cdot)$. Her messages to $\bob$ then include the bits $F_b= S_b \oplus\ext(A(b), T_b)$ for $b\in \brak{0,1}$, along with the strings $X, T_0, T_1$. 

\noindent {\bf Correctness:} Suppose the devices are ideal. Consider the case when $D=0$ and honest Alice and Bob. In that case, $B(X)=A(0)$ (by the property of the MS device). Then $\bob$ can compute $\ext(A(0), T_0)=\ext(B(X), T_0)$, and recover $S_0$ by computing $\ext(B(X), T_0)\oplus F_0$. However, if the devices are not ideal but slightly faulty, $B(X)$ may be close to $A(0)$ but not exactly equal. To remedy this, $\alice$ includes $W_0=\eH(A(0))$ and $W_1=\eH(A(1))$ with her messages, where $\eH$ is the parity check matrix of an appropriately chosen classical \textit{error correcting code}. If the test-phase passed, then the syndrome $W_0$ and the string $B(X)$ allow $\bob$ to do syndrome decoding to recover $A(0)$ (and hence, $S_0$) correctly with high probability. This guarantees the \textbf{robustness} of our protocol. 

\noindent {\bf Receiver Security:} The communication between (cheating) $\alice$ and $\bob$ for the test phase checks is independent of $\bob$'s input $D$. After the test phase, there is no communication from $\bob$ to $\alice$. This guarantees (informally) that $\alice$ does not get any information about $\bob$'s choice bit $D$.

    \noindent {\bf Sender Security:} For sender security, we require that a cheating $\bob$ cannot guess both $S_0$ and $S_1$. Since we use $\ext(A(b), T_b)$ to hide $S_b$, it suffices to show that there is min-entropy in either $A(0)$ or $A(1)$ given all bits of $\bob$.  
   \begin{itemize}
    \item  We construct a single-copy (hybrid quantum-classical) security game (see Table \ref{table:ot-single-copy}). The hybrid strategy takes care of the \textbf{DELAY} in the protocol. The idea behind the game (say $G_1$) is that with probability $\delta>0$, testing of the devices happens where the MS predicate is checked. Otherwise, $\bob$ needs to guess two of $\alice$'s output bits. 
    \item We bound the maximum probability of winning $G_1$. \textbf{Note that for this step, the bounded quantum storage is crucial.} We then prove a parallel repetition theorem for $G_1$ (see below), which shows that the probability of winning the $n$-copy game in parallel is exponentially small in $n$. 
\item Using this parallel repetition, we show that if the test checks passed, $\bob$ cannot guess both $A(0)$ and $A(1)$ (see Claim \ref{claim:case_analysis}), i.e.,  there exists a bit $\widetilde{D}$  for which min-entropy in $A(\wtt{D})$ given (all bits of) $\bob$ is large. 
\item The extractor ensures that $\ext(A(\widetilde{D}), T_{\widetilde{D}})$ is uniform and independent of $\bob$ and can be used to hide $S_{\widetilde{D}}$ as a one-time pad, hence implying sender security.
    \end{itemize}
     This parallel repetition theorem is a key technical contribution of this work.

\subsubsection*{Parallel repetition}
A parallel repetition theorem is concerned with the success probability 
of multiple copies of a game $G$ played in parallel. Let $\val(G)$ denote the maximum probability (among all strategies) of winning the game $G$. 
Given a game $G$ with $\val(G)$ bounded away from 1, a parallel-repetition theorem says that for an $n$-fold repetition of the game $G$ played in parallel (denoted as $G^n$), the $\val(G^n)$ is exponentially small in $n$. Parallel repetition of quantum games when the input distribution is independent across the two parties is known \cite{JPY}. Since the input distribution to our security game is not independent, we use the idea of anchored games~\cite{anchoring3,anchoring_2} to define a game $G_1$ (see Table \ref{table:ot-single-copy}) and show that $\val(G_1)$ is bounded away from $1$. In anchoring, the idea is to modify the input distribution and introduce a special symbol $*$ with some probability such that, conditioned on $*$, the input distribution for both parties becomes independent. 

\textbf{Bounding the value of $G_1$:} In the game $G_1$, with probability $\delta$, we test the MS predicate, and with probability $1-\delta$, we want (cheating) $\bob$ to guess two of $\alice$'s output bits. If the probability that the devices satisfy the MS predicate is small (say less than $1-\eps_r$), then we are done, as the value of the game is automatically bounded away from $1$. If the probability that the devices win the MS game is large (more than $1-\eps_r$), then using the rigidity theorem (see Fact~\ref{fact:rigidity}),  we get that $\alice$'s state and measurements (since $\alice$ is honest) are close to that of ideal MS measurements. Since $\bob$ is dishonest, we cannot get any guarantee for the measurements of $\bob$. Since $\bob$ does not know $X$ (which is uniformly random) before \textbf{DELAY}, and $\alice$'s measurements (for different values of  $X$) are roughly "mutually unbiased", $\bob$ cannot guess 
 two of $\alice$'s output bits; although $\bob$ can guess a single bit with certainty by playing the MS game honestly. After \textbf{DELAY}, $\bob$ gets $X$, but measurements have already taken place on the quantum state and hence $\vall(G_1)$ is bounded away from $1$. Note that for this step, the \textbf{bounded quantum storage is crucial}.
 
 %We use a hybrid (quantum-classical) strategy to show the security of our protocol, which allows us to deal with the  \textbf{DELAY} in the game.

\textbf{Multiple-copy game}: Consider the $n$-copy game $G_n$. Let us condition on success on a subset $\mcC \subseteq [n]$. We are done if the probability of success above is already (exponentially) small. If not, we show that for a randomly selected coordinate $j\notin \mcC$, the probability of winning in the coordinate $j$, conditioned on success in $\mcC$, is close to $\val(G_1)$ and hence is bounded away from $1$. 
Then, the overall success probability becomes exponentially small in $n$, after we have identified $\Omega(n)$ such coordinates.

We suppose towards contradiction that for some coordinate $j\notin \mcC$, the probability of winning in the coordinate $j$, conditioned on success in $\mcC$, is close to $1$.  In our protocol, the inputs to $\alice$ and $\bob$ are $X=X_1\dots X_n$ and $Y=\wht{Y}C\wht{X}=\wht{Y}_1C_1\wht{X}_1\dots \wht{Y}_nC_n\wht{X}_n$ respectively, which are not independent. Using anchoring ideas, if we condition on $\wht{X}_j=*$, the input distribution for the $j$-th coordinate becomes independent. Let the starting state for the protocol be $\ket{\theta}$ and the global state conditioned on success in $\mcC$ be $\ket{\varphi}$ (see eq. \eqref{eq:def_varphi}). Define the state $\ket{\varphi_*}$ as  $\ket{\varphi}$ conditioned on $\wht{X}_j=*$. Let $\wtt{A}$ and $\wtt{B}$ be all registers in possession of $\alice$ and $\bob$ other than $\wtt{X}X$ and $\wtt{Y}Y$ respectively (here $\wtt{X}$ is a copy of $X$ and shows that $X$ is classical, similarly for $Y$).  

\textbf{Strategy for winning game $G_1$:} $\alice$ and $\bob$ share the state $\ket{\varphi_*}$ at the beginning. 
\begin{itemize}
    \item  Using chain rules for relative entropy and mutual information, we can show that $\I(X_j:\wtt{Y}Y\wtt{B})_{\ket{\varphi_*}}$ and $\I(Y_j:\wtt{X}X\wtt{A})_{\ket{\varphi_*}}$ are close to 0; we actually prove that the mutual information conditioned on $R_j$ defined in eq. \eqref{eq:define_R_j} is small. For this to hold, $\wht{X}_j=*$ is crucial. 
    \item We also show that for the $j$-th coordinate, the distribution of questions ${\varphi_*}^{X_j\wht{Y}_jC_j}$ is close to the anchored distribution $\mu(X_j\wht{Y}_jC_j)$ (see eq. \eqref{eq:anchored_def}). 
    \item Suppose $\alice$ and $\bob$ receive inputs $x$ and $\wht{y}c$ according to $\mu(X_j\wht{Y}_jC_j)$. Since $\I(X_j:\wtt{Y}Y\wtt{B})_{\ket{\varphi_*}}\approx 0$, $\bob$'s state is roughly independent of $X_j$ (similarly for $\alice$). Using ideas similar to \cite{JRS}, we show that there exist a unitary transformations $U_x$ and $V_{\wht{y}c}$ that $\alice$ and $\bob$ can apply on their sides of the state $\ket{\varphi_*}$ to correctly embed the inputs $x$ and $\wht{y}c$ in the $j$-th coordinate of $\ket{\varphi_*}$.
\item If $c_j=0$ (corresponding to the test-phase), $\bob$ measures his register $B_j$ to generate the answer. If $c_j=1$, a $\CNOT$ acts between all registers of $\bob$ and $\ket{0}^{N''}$ after which $N''$ cannot be accessed by $\bob$ (corresponding to after the \textbf{DELAY} in the protocol). After this, $\bob$ gets access to $\wht{x}_j$. We show that using classical post-processing, $\bob$ can now generate the register $B_j$ properly, even though the quantum states on his side have decohered. \end{itemize}
Thus the state in the relevant registers generated by $\alice$ and $\bob$ is close to the state of those registers in $\ket{\varphi}$ (conditioned on the correct inputs), in which we have the guarantee that the probability of winning in the $j$-th coordinate is high (close to $1$). This gives a strategy for winning the game $G_1$ with a probability close to $1$, which is a contradiction as $\val(G_1)$ is bounded away from $1$.  This completes the parallel repetition argument.

For our security result, we need a strengthening of the above parallel repetition result in the form of a threshold theorem where we show that for the $n$-copy game $G_n$ (using ideas similar to \cite{anup_rao, concentration}) for some $\delta>0$, the probability of winning more than $(\val(G_1)+\delta)n$ fraction of games in $G_n$ is also exponentially small.

\subsection*{Organization}
In Section~\ref{sec:Preliminaries_DI_BC} we introduce the preliminaries and notations required for this work.
The section also contains several useful claims and lemmas that will be useful for the analysis of the device-independent protocols that we propose. 
In particular, this section contains results regarding distributions obtained from Magic Square devices.
In Section~\ref{sec:framework_for_composibility}, we establish our framework of composability.
Section~\ref{section:OT} contains our device-independent protocol for OT, and its correctness and security analysis. In Section~\ref{threshold}, we give proof of our threshold theorem using parallel repetition. We also defer the proof of the composition theorem and several claims to the Appendix.

\section{Preliminaries}\label{sec:Preliminaries_DI_BC}
In this section, we present some notation, definitions, and facts that we will need later for our proofs. 

\begin{definition}[$\negl$ function]\label{def:negligible}
    We call a function $\mu : \mathbb{N} \rightarrow \mathbb{R}^{+}$ negligible if for every positive polynomial $p(\cdot)$, there exists an $N \in \mathbb{N}$ such that for all $\lambda > N$, $\mu(\lambda)< \frac{1}{p(\lambda)}$. We let $\mu(\lambda) = \negl(\lambda)$ represent that $\mu(\lambda)$ is a negligible function in $\lambda$. 
\end{definition}
\subsection*{Notation}
\begin{enumerate}
    \item For any $l \in \mathbb{R}$ we use $[l]$ to denote the set $\lbrace 1, 2, \ldots, \lfloor l \rfloor \rbrace$.
    \item  $x \leftarrow_P X$ denotes $x$ drawn from $X$ according to distribution $P$. We use $U$ to denote the uniform distribution when the underlying domain is clear from context. 
    
    \item We use $U_t$ to represent uniform distribution on $t$-bits.

    \item $\negl(n)$ is the class of negligible  functions in $n$.
    \item $A\approx_{\negl(n)}B$ denotes $\onenorm{A}{B}\leq \negl(n)$.
    \item Capital letters denote random variables, and small letters denote their realizations; for example, $X=x$.
    %\item parameters: $p(n)=n$ and $q(n)=3$. 
   
    \item For a vector or a string $v=v_1,v_2,\ldots,v_n$ we use $v(i):=v_i$. 
    For a restriction to subset $S \subseteq [n]$, we use $v_{S}$; for example, $v_{\lbrace 1,4,6\rbrace}= v_1 v_4 v_6$.  For an index $j\in [n]$, $v_{-j}\defeq v_1 \dots v_{j-1}v_{j+1}\dots v_n$.
    \item For a random variable $X$, we use a shorthand $X(c)$ to denote $\Pr(X=c)$.
    \item $d_H(x,y)$ is the Hamming distance between binary strings $x$ and $y$ and $\wt_H(x)$ denotes the Hamming weight of $x$. 
     \item For a distribution $P^{AB}$ we use $\left(P\big\vert b\right)^A$ or $P_b^A$ to denote conditional distributions. 
     When the underlying space $A$ is clear from the context, we drop it. 
                  
                     \item For distributions $P_1$ and $P_2$, we use $P_1 \otimes P_2$ to denote $P_1$ and $P_2$ are independent.
                \item     Let $\rho^{XYR}=P^{XYR}$ be a classical state. We denote \begin{equation}
        \label{eq:notation}
\rho^{YR}_x=P^{YR}_x=P^Y_{x}\cdot P^R_{x,Y}=\rho^{Y}_x \cdot \rho^{R}_{x,Y}.
 \end{equation}
        
    \end{enumerate}

    \subsection*{Error correcting codes }
A linear error-correcting code $\C$ is a $k$-dimensional linear subspace of an $n$-dimensional space (for $k \leq n$). Elements of the code $\mathrm{C}$ are called \emph{codewords}.
In this paper, we shall only consider vector spaces over $\bbF_2$, which is the alphabet of the code.
Thus, a codeword $c \in \C$ is a bit string of length $n$ and will be represented as $c_1,c_2,\ldots, c_n$.

The \emph{minimum distance}, or simply \emph{distance}, of a code is the minimum Hamming distance between two distinct codewords $u$ and $v$, i.e., the number of components in the vector $u\oplus v$ that have non-zero entries. For a linear code $\C$, the minimum distance is the same as the minimum Hamming weight of a non-zero codeword (note that the zero vector is always a codeword), and it will be denoted by $\dist(\C)$. We denote the distance of a string $w$ from the code as $d_h(w,\C)=\min \{d_h(w,c)|c\in \C\}$. A $k$-dimensional linear code with distance $d$, sitting inside an $n$-dimensional subspace is denoted as $\cbra{n,k,d}$ code. The quantity $\frac{k}{n}$ is called the rate of the code and $\frac{d}{n}$ is called the relative distance of the code. 

Since a code $\C$ is a $k$-dimensional subspace of an $n$-dimensional space, it can be thought of as the kernel of a matrix $H \in \mathbb{F}_2^{(n-k)\times n}$, which is called its \emph{parity check matrix}.
The distance of a code is then the minimum Hamming weight of a non-zero vector $u$ such that $Hu = 0^{n-k}$.
For our purposes, we will need codes as given in Fact \ref{fact:existance_of_codes}. 
Well-known code constructions such as expander codes meet this requirement.

\begin{fact}[\cite{guruswami2012essential}] \label{fact:existance_of_codes} For any $0<r<1$, there exists an explicit $[n,r\cdot n,\gamma \cdot n]$ code family with rate $r$ for small enough relative distance $\gamma>0$ depending on $r$. 
\end{fact}

 {\begin{definition}[Syndrome decoding] Given a linear $[n,k,d]$ code $\C$ having parity check matrix $H$ and a syndrome  $s$  with the promise that $s=Hw$ for some $w$ having $d_H(w,\C)<d/2$, syndrome decoding finds out the unique vector $y$ satisfying $wt_H(y)<d/2$ and $Hy=s$. \\ Note: If the promise on $s$ is not satisfied, the syndrome decoding algorithm may not be able to find a string $y$ for which $wt_H(y)<d/2$ and $Hy=s$. In this case, we assume that the output of the syndrome decoding algorithm is $\bot$.
We can assume this without loss of generality because both conditions can be checked efficiently for any output of the syndrome decoding algorithm.    
\end{definition}
\begin{fact}[\cite{Overbeck2009}] \label{fact:syndrome_decoding}
For linear codes with efficient decoding, syndrome decoding is also efficient. 
    \end{fact}}

\subsection*{Extractors}
\begin{definition}[Min-entropy]
    For a random variable $X$, its min-entropy $\hmin(X)$ is defined as:
    \[\hmin(X) \defeq \min_{x} \left( - \log\left({\Pr(X=x)}\right)\right).\]
\end{definition}A \emph{source} having min-entropy at least $k$ {and supported on strings of length $n$ is said to be an $(n,k)$-min-entropy source. When the length of the strings produced by the source is clear from context, we omit $n$ and simply call it a $k$-source.} For the definition of conditional min-entropy, $\hmin(\cdot\vert\cdot)$, refer to \Cref{def:Hmin}.
        {\begin{definition}[Strong extractor]
A $(k,\eps)$-strong extractor is a function 
\[
\ext:\{0,1\}^n \times \{0,1\}^{l} \rightarrow \{0,1\}^{m},
\] 
such that for every  $k$-source $X$, we have
\[
\norm{\ext(X,U_{l})\otimes U_{l}-U_{m+l}}_1\leq \eps.
\]
\end{definition}
\begin{fact}[\cite{DeVidick}]\label{fact:strong_extractor}
For any $\eps>0$ and $m\in \mathbb{N}$, there exists an explicit $(k, \eps)$-strong extractor \[{{\ext}}: \lbrace{0,1\rbrace}^m\times \brak{0,1}^t\to \brak{0,1}\] with uniform seed $t= O(\log^2\frac{m}{\eps})$ and $k= O(\log \frac{1}{\eps})$.
\end{fact}

    \begin{fact}[\cite{DORS}] \label{fact:average_to_worst_min_entropy}
    For any $\eps >0$,
    \[\Pr_w\left( \hmin\left( X \vert w\right) \geq \hmin\left( X \vert W\right) - \log\left(\frac{1}{\eps} \right)\right) \geq 1- \eps.\]
\end{fact}
\begin{fact}[\cite{operational_meaning}]\label{fact:p_guess}
    Let $\rho^{XB} = \sum_x p_x \ketbra{x}{x} \otimes \rho_x^B$  where $X$ is classical. Then
\[\hmin(X|B)_\rho = -\log (p_{guess}(X|B)_\rho),\]
where $p_{guess}(X|B)_\rho$ is the maximal probability of decoding $X$
from $B$ with a POVM $\{E_x^B\}_x$ on $B$, i.e.,
\begin{equation}
    \label{eq:guessing_prof_def}
    p_{guess}(X|B)_\rho \defeq \max_{\{E_x^B\}_x}
p_x\Tr(E_x^B\rho_x^B) .\end{equation}
If $B$ is trivial, then \[
p_{guess}(X)=\max_{x}p_x.
\]
 \end{fact}
   
\subsection*{Quantum Device}

\begin{definition}[Quantum Device] \label{def:device} 
    A quantum device $\device$ consists of input and output ports, allowing for classical inputs $x\in \mathcal{X}$ and classical outputs $o\in \mathcal{O}$. The behavior of the device $\device$ is specified by a tuple $(\rho,\brak{\Lambda^x_o}_{x,o})$ where  $\forall x \in \mathcal{X}: \sum\limits_{o}\Lambda^x_o=\mathbb{I}$, and for all $x,o,~\Lambda^{x}_o \geq 0$, i.e., $\forall x\in \mathcal{X}$, $\brak{\Lambda^x_o}_{o}$ is  a POVM.  The probability of output $o$ on input $x$ is given by $\Pr^{\device}(o|x) \defeq \Tr\left[\Lambda^x_o\rho\right]$.

    For a device $\device$, an $\eps$-near device $\device_\eps$ satisfies  $\forall (x,o): \left|\Pr^{\device}(o|x)  -  \Pr^{\device_\eps}(o|x) \right| \leq \eps$. We denote the set of all $\eps$-near devices to $\device$ as $\mathcal{D}_\eps\left( \device \right)$.  
\end{definition} 

The above device is used by parties participating in a communication protocol as follows: 
\begin{enumerate}
    \item An honest party can only interact with the device via classical input and classical output. 
    When input $x$ is chosen, the resultant output is obtained by measuring the state $\rho$ with POVM $\lbrace \Lambda^x_o \rbrace_{o}$. 
     \item A device may be shared between two parties (say $\alice$ and $\bob$). 
    In that case,  the device contains a joint state $\rho^{AB}$ (where $A$ is $\alice$'s share and $B$ is $\bob$'s share). 
    \item An adversarial party can design all the devices used in the protocol together (the combined states and measurements) before the protocol starts. 
%    \item A dishonest party may ``open" the box and choose to do an arbitrary measurement on their part of the shared device.  
\end{enumerate}

\begin{definition}
    [Ideal Magic Square device]\label{def:MS_Game}
   A Magic Square device has classical inputs $\mathcal{X} \times \mathcal{Y}$, where $\mathcal{X} = \mathcal{Y} = \lbrace 0,1,2\rbrace$.
   The output of a Magic Square device is $(a,b) \in \mathcal{A} \times \mathcal{B}$ where $\mathcal{A} = \lbrace 000,011,101,110\rbrace$ and $\mathcal{B}= \lbrace 001,010,100,111\rbrace$.
   On input $(x,y)$ an ideal Magic Square device produces outputs $(a,b)$ such that $a_y=b_x$. 
   The ideal state $\rho$ is a double-EPR state;
   \[ \rho =  \ket{\Psi^{\MS}} := \frac{1}{\sqrt{2}} \left( \ket{00}+ \ket{11}\right)^{M_1 N_1} \ \frac{1}{\sqrt{2}} \left( \ket{00}+ \ket{11}\right)^{M_2 N_2}.\]
   Note that the third bit of $a$ (and analogously $b$) is deterministically fixed by the first two bits.
   Hence, the first two bits suffice to uniquely describe the output.
   The measurements $\Lambda^{x,y}_{a,b}$ are defined via the following equation (see~\cite{srijita_xot} for more details):
   \[ \Lambda^{x,y}_{a,b} = \left(\Pi^{x0}_{a_0} \Pi^{x1}_{a_1} \right)^{M_1M_2} \otimes  \left(\Pi^{0y}_{b_0} \Pi^{1y}_{b_1} \right)^{N_1N_2} = M^{\MS}_{a|x} \otimes M^{\MS}_{b|y}, \text {where }\]
\begin{table}[H]
\centering
\begin{tabular}{|c|c|c|c|}
\hline
\text{$x, y$} & 0 & 1 & 2 \\\hline
\rule{0pt}{4.5ex}0 & $\begin{aligned} \Pi_0^{00} = \ketbra{0}{0}\otimes\Id \\ \quad \end{aligned}$ & $\begin{aligned} \Pi_0^{01} = \Id\otimes\ketbra{0}{0} \\ \quad \end{aligned}$ &  $\begin{aligned} \Pi_0^{02} & = \ketbra{0}{0}\otimes\ketbra{0}{0} \\ & \quad + \ketbra{1}{1}\otimes\ketbra{1}{1} \end{aligned}$ \\[0.5cm]
  & $\begin{aligned} \Pi_1^{00} = \kb{1}\otimes\Id \\ \quad \end{aligned}$ & $\begin{aligned} \Pi_1^{01} = \Id\otimes\kb{1} \\ \quad \end{aligned}$ & $\begin{aligned} \Pi_1^{02} & = \kb{0}\otimes\kb{1} \\ & \quad + \kb{1}\otimes\kb{0} \end{aligned}$ \\[0.5cm] \hline
\rule{0pt}{5ex}1 & {$\begin{aligned} \Pi_0^{10} = \Id\otimes \kb{+} \\ \quad \end{aligned}$} & {$\begin{aligned} \Pi_0^{11} = \kb{+}\otimes\Id \\ \quad \end{aligned}$} & {$\begin{aligned} \Pi_0^{12} & = \kb{+}\otimes\kb{+} \\ & \quad + \kb{-}\otimes\kb{-}\end{aligned}$} \\[0.5cm]
  &  $\begin{aligned} \Pi_1^{10} = \Id\otimes\kb{-} \\ \quad \end{aligned}$ & $\begin{aligned} \Pi_1^{11} =  \kb{-}\otimes\Id \\ \quad \end{aligned}$ & {$\begin{aligned}\Pi_1^{12} & = \kb{+}\otimes\kb{-} \\ & \quad + \kb{-}\otimes\kb{+}\end{aligned}$} \\[0.5cm]\hline
\rule{0pt}{5ex}2 & $\begin{aligned} \Pi_0^{20} & = \kb{1}\otimes\kb{+} \\ & \quad + \kb{0}\otimes\kb{-}\end{aligned}$ & $\begin{aligned} \Pi_0^{21} & = \kb{+}\otimes\kb{1} \\ & \quad + \kb{-}\otimes\kb{0}\end{aligned}$ & $\begin{aligned} \Pi_0^{22} & = \kb{+i}\otimes\kb{+i} \\ & \quad + \kb{-i}\otimes\kb{-i}\end{aligned}$ \\[0.5cm]
  & $\begin{aligned} \Pi_1^{20} & = \kb{0}\otimes\kb{+} \\ & \quad + \kb{1}\otimes\kb{-}\end{aligned}$ & $\begin{aligned} \Pi_1^{21} & = \kb{+}\otimes\kb{0} \\ & \quad + \kb{-}\otimes\kb{1}\end{aligned}$ & $\begin{aligned} \Pi_1^{22} & = \kb{+i}\otimes\kb{-i} \\ & \quad + \kb{-i}\otimes\kb{+i}\end{aligned}$ \\[0.5cm] \hline
\end{tabular}
\caption{$\Lambda^{x,y}_{a,b}$ description for a Magic Square device}
\label{tab:ms-meas}
\end{table} \label{example:idea_MS}
\end{definition}

\begin{definition}[\MS~Predicate] Let $X,Y$ be random variables on $\lbrace 0,1,2 \rbrace$ and let $\rho^{AB}$ be a state. 
Let $\Lambda^A$ and $\Lambda^B$ be POVMs on $A$ and $B$ respectively, that is,
$\Lambda^A := \lbrace \Lambda^{A}_{o_A}\rbrace_{o_A}$ and $\Lambda^B := \lbrace \Lambda^B_{o_B} \rbrace_{o_B}$ be such that
\begin{enumerate}
    \item For all $o_A , o_B \in \lbrace 0,1 \rbrace^3$, we have,  $ 0 \leq \Lambda^{A}_{o_A} \leq \Id^A$ and $ 0 \leq \Lambda^{B}_{o_B} \leq \Id^B$.
    \item  $\sum\limits_{o_A \in \lbrace 0,1\rbrace^3} \Lambda^{A}_{o_A} = \Id^A$ and $\sum\limits_{o_B \in \lbrace 0,1\rbrace^3} \Lambda^{B}_{o_B} = \Id^B$.
\end{enumerate}
Let $\mu^A \mu^B$ be the distribution obtained by measuring $\rho$ with POVM $\Lambda^A \otimes \Lambda^B$. 
We define a random variable 
$\MS(X,Y,\Lambda^A, \Lambda^B)_{\rho} := \MS (X,Y,\mu^A,\mu^B) = \Id_{\mu^{A}(Y) = \mu^B(X)}$.    
\end{definition}
The following fact can be easily verified (also stated in \cite{srijita_xot}).
\begin{fact}\label{fact:ms_distribution_both_sides} For all $a,b,x,y$, the probability distribution satisfies: 
     \begin{align*}
  \Pr\left( a,b \vert x,y\right) & = \frac{1}{8}  & \mbox{if $a_y=b_x, \oplus_i a_i=0, \oplus_i b_i=1$}
  \\
  & = 0 & \mbox{otherwise}.
\end{align*}
\end{fact}

\subsection*{Rigidity manipulations}

Rigidity theorems are a class of theorems that say that any robust device must produce distributions close to the ideal one in order to win games with a probability close to 1.
\begin{fact}[Lemma 5, \cite{srijita_xot}]\label{fact:rigidity}
    Consider any state $\ket{\rho}$ on the registers $AB$ and projective measurements $M_{a|x}$ and $N_{b|y}$ such that $M_{a|x}$ acts only on $A$ and $N_{b|y}$ acts only on $B$ (see Definition \ref{def:MS_Game}).
    If  this state and measurements win the Magic Square game with probability $1-\eps_r$, then there exist local isometries $V_A:A\to A_0 A_1J^A$ and $V_B:B \to B_0B_1J^B$ and a state $\ket{\texttt{junk}}$ on $J^AJ^B$ such that for all $x, y, a, b$, we have:    
    \begin{align*}
        &\norm{\left(V_A\otimes V_B\right)\ket{\rho}-\ket{\Psi^{\MS}} \otimes \junkstate}_2 \leq O(\eps_r^{1/4}), \\
                &\norm{\left(V_A\otimes V_B\right)\left(M_{a|x}\otimes\Id\right)\ket{\rho}-\left(M^{\MS}_{a|x}\otimes \Id \right)\ket{\Psi^{\MS}}\otimes {\junkstate}}_2 \leq O(\eps_r^{1/4}), \\
                &\norm{\left(V_A\otimes V_B\right)\left(\Id\otimes N_{b|y}\right)\ket{\rho}-\left(\Id\otimes N^{\MS}_{b|y}  \right)\ket{\Psi^{\MS}}\otimes {\junkstate}}_2 \leq O(\eps_r^{1/4}).
    \end{align*}
    \end{fact}
    \begin{claim}\label{claim:length}
For all $a,x$ such that $\oplus_i a_i=0$,
    \[\norm{\left(M^{\MS}_{a|x}\otimes \Id \right)\ket{\Psi^{\MS}}\otimes {\junkstate}}_2=\frac{1}{2}.\]
\end{claim}
\begin{proof}
Since $M^\MS_{a|x}$ is a projective measurement, we have
\begin{align*} 
\norm{\left(M^{\MS}_{a|x}\otimes \Id \right)\ket{\Psi^{\MS}}\otimes {\junkstate}}_2^2 &=\junkstatedagger\otimes\bra{\Psi^{\MS}}\left(M^{\MS}_{a|x}\otimes \Id \right)\ket{\Psi^{\MS}}\otimes {\junkstate}.
\\&=\bra{\Psi^{\MS}}\left(M^{\MS}_{a|x}\otimes \Id \right)\ket{\Psi^{\MS}}
\\&=\frac{1}{4}. &\mbox{(from Fact \ref{fact:ms_distribution_both_sides})}
\end{align*}
Taking square root now completes the proof.
\end{proof}
\subsection*{Quantum information tools}
\begin{definition}[$\ell_1$ distance] For an operator $A$, the $\ell_1$ norm is defined as $\|A\|_1 \defeq \Tr \sqrt{A^\dagger A}$. For operators $A,B$, their $\ell_1$ distance is defined as $\onenorm{A}{B}$. We use shorthand $A\approx_{\eps}B$ to denote $\onenorm{A}{B}\leq \eps$.
    
\end{definition}

\begin{definition}[Fidelity]
    For (quantum) states $\rho,\sigma$, \[
    F(\rho,\sigma)\defeq \|\sqrt{\rho}\sqrt{\sigma}\|_1.
    \]
\end{definition}
\begin{definition}[Bures metric]
    For states $\rho,\sigma$, \[
    \Delta_B(\rho,\sigma)\defeq \sqrt{1-F(\rho,\sigma)}.
    \]
   
\end{definition}

\begin{fact}[\cite{fuchs}]\label{fact:fuchs}
    For states $\rho, \sigma$,
\[ \Delta_B^2(\rho, \sigma) \leq \frac{\onenorm{\rho}{ \sigma}}{2} \leq
\sqrt{2}\Delta_B(\rho, \sigma).\]
\end{fact}
\begin{definition}\label{def:entropy}For a  state $\rho^A$, its von Neumann entropy is defined as \[
{\S}(A)_\rho\defeq -\Tr(\rho\log\rho).
\]
\end{definition}
\begin{fact}[\cite{wilde}]\label{fact:entropic_inequalities}
    For a state $\rho^A$, \[{\S}(A)_\rho \leq |A|.\]
\end{fact}
    \begin{fact}[Gentle Measurement Lemma \cite{wilde}] \label{fact:gentle_measurement}
     Let $\rho$ be a state. Let $M \in \mcL(\mathcal{H}_A)$   such
that $M^\dagger M \leq I_A$ and $\Tr(M \rho M^\dagger ) \geq 1 - \eps$. Let $\rho' =
\frac{M \rho M^\dagger }
{\Tr M \rho M^\dagger} $. Then, \[\Delta_B (\rho, \rho') \leq \sqrt{\eps}.\]
    \end{fact}
\begin{claim}\label{claim:L_2_to_L_1}
    Let $\ket{u}$ and $\ket{v}$ be such that
    \begin{enumerate}
    \item $\Vert \ket{u} \Vert_2 \leq 1$ and $\Vert \ket{v} \Vert_2 \leq 1$.
        \item $\Vert \ket{u}-\ket{v} \Vert_2 \leq \eps$ for some small $\eps>0$. 
        \item $\Vert \ket{u} \Vert_2 \geq c$ for some $c>0$ ($\eps<c/10$).
    \end{enumerate}
    Then, \[\onenorm{\ketbra{u}{u}}{\ketbra{v}{v}} \leq O\left(\frac{\eps}{c^4}\right).\]
\end{claim}
\begin{proof}
     We defer the proof of this claim to the Appendix.  
\end{proof}

 \begin{definition}[Relative entropy \cite{wilde}] The quantum relative entropy between two states $\rho$ and $\sigma$ is defined as:
    \[\D(\rho\|\sigma)\defeq \Tr\left(\rho(\log(\rho)-\log(\sigma))\right).\]  
    
    \end{definition}
    \begin{definition}[Max-divergence \cite{marco_book}]\label{def:H_infinity}
    \[\D_\infty(\rho\|\sigma)\defeq \min\{\lambda : \rho
   \leq 2^\lambda\sigma\}.\]  
    \end{definition}
    \begin{definition}\label{def:Hmin}
    The conditional min-entropy is defined as:\[
        \hmin(X|B)_\rho \defeq -\min_{\sigma^B} \mathrm{D}_\infty(\rho^{XB}\|\Id^X\otimes\sigma^B).
        \]
\end{definition}
    \begin{definition}[Mutual information \cite{wilde}]\label{def:mutual_info} The quantum mutual information of a bipartite state $\rho^{XY}$ is defined as:
    \[\I(X:Y)_\rho\defeq \S(X)_\rho+\S(Y)_\rho-\S(XY)_\rho = \D(\rho^{XY}\|\rho^X\otimes\rho^Y).\]  Let $\rho^{XYZ}$ be a quantum state with $Y$ being a classical register.
The mutual information between $X$ and $Z$, conditioned on $Y$,
is defined as
\[
\I(X : Z |Y )_\rho
\defeq \E_{y\leftarrow Y}
\I(X : Z)_{\rho_y}.\]    \end{definition}
\begin{fact}\label{fact:d_infty_calculation}
     Let $\ket{\phi}^{AB}$ be a bipartite pure state with the
marginal state on register $B$ being $\rho$. Let a 0/1 outcome
measurement be performed on register $A$ with outcome $O$.
Let $\Pr[O = 1] = q$. Let the marginal states on register $B$
conditioned on $O = 0$ and $O = 1$ be $\rho_0$ and $\rho_1$ respectively.
Then, \[\D_\infty(\rho_1\|\rho) \leq -\log(q).\]
\end{fact}
\begin{proof}
    \[\rho=q\rho_1+(1-q)\rho_0.\]
    From Definition \ref{def:H_infinity}, we get \[\D_\infty(\rho_1\|\rho) \leq -\log(q).\]
\end{proof}
\begin{fact}\label{fact:chain_rule_dmax}
    Let $\rho_1,\rho_2$ and $\rho_3$ be states. Then,
    \[\D_\infty(\rho_1\|\rho_2)\leq\D_\infty(\rho_1\|\rho_3)+\D_\infty(\rho_3\|\rho_2).\]
\end{fact}
\begin{proof}
Let $l_1=\D_\infty(\rho_1\|\rho_3)$ and $l_2=\D_\infty(\rho_3\|\rho_2)$. Then, we have \[
    \rho_1 \leq 2^{l_1}\rho_3, \] \[ \rho_3 \leq 2^{l_2}\rho_2.\] Combining these, we get \[ \rho_1 \leq 2^{l_1+l_2}\rho_2.\] This along with Definition \ref{def:H_infinity} completes the proof.
    
\end{proof}
\begin{fact} [Monotonicity~\cite{marco_book}]\label{fact:monotonicity}
  For states $\rho_1, \rho_2$, we have \[
  \D_\infty(\rho_1\|\rho_2) \geq \D(\rho_1\|\rho_2).
  \]
\end{fact}
\begin{fact}[Chain rule for relative entropy \cite{khatri_book}]\label{fact:chain_rule_D}
    Let \[\rho=\sum_x\mu(x)\ketbra{x}{x}\otimes\rho_x \] and
    \[\rho^1=\sum_x\mu^1(x)\ketbra{x}{x}\otimes\rho^1_x. \]
    Then\[
\D(\rho^1\|\rho)=\D(\mu^1\|\mu)+\E_{x\leftarrow\mu^1}[\D(\rho^1_x\|\rho_x)].
    \]
\end{fact}
\begin{fact}[Data processing~\cite{marco_book}] \label{fact:data} Let $\rho, \sigma$ be states and $\Phi$ be a CPTP map. Then, 
  \[\Delta_B(\Phi(\rho), \Phi(\sigma)) \leq  \Delta_B(\rho, \sigma).\]
  The above also holds for the $\ell_1$ distance. We also have,
  \[\D(\Phi(\rho)\|\Phi(\sigma)) \leq \D(\rho\|\sigma).\]
\end{fact}
We also have the following fact which is similar to data processing for $\rho,\rho'\geq 0$ (here these are not necessarily states).\begin{fact}[\cite{renner_thesis}]\label{fact:data-PSD}
Let $\rho,\rho' \geq 0$ and let $\Phi$ be a CPTP map. Then
\[\| \Phi(\rho) - \Phi(\rho')\|_1\leq \|\rho - \rho'\|_1 .\]
    
\end{fact}
\begin{fact}[Joint-convexity of  relative entropy \cite{wilde}]\label{fact:joint_convexity} Let $\rho_0, \rho_1, \sigma_0, \sigma_1 $ be states and $\lambda\in[0,1]$. Then \[
\diver{\lambda\rho_0+(1-\lambda)\rho_1}{\lambda\sigma_0+(1-\lambda)\sigma_1} \leq \lambda\diver{\rho_0}{\sigma_0}+(1-\lambda)\diver{\rho_1}{\sigma_1}.
\]
\end{fact}

\begin{fact}[\cite{khatri_book}] \label{fact:minimisation_for_mutual_info}
   
   For quantum states $\rho^{XY}, \sigma^X$ and $\tau^Y$, the following holds:\[
    \D(\rho^{XY}\|\sigma^X\otimes\tau^Y)\geq \D(\rho^{XY}\|\rho^X\otimes\rho^Y)=\I(X:Y)_\rho \geq 0.
    \]
\end{fact}

\begin{fact}[Chain rule for mutual information \cite{khatri_book}]\label{fact:chain_rule_mutual_info} For a state $\varphi^{XYZ}$,
    \[{\I(X:Y,Z)_{\varphi}=\I(X:Z)_{\varphi}+\I(X:Y|Z)}_{\varphi}.\]
\end{fact}
\begin{fact}[\cite{JRS}]\label{fact:unitary_existence}
    Let $\mu$ be a probability distribution over
$X \times Y.$ Let $\mu_X$ and $\mu_Y$ be the marginals of $\mu$ on $X$ and $Y$.
Let \[\ket{\phi} \defeq \sum_{
x\in X ,y\in Y}
\sqrt{\mu(x, y)} \ket{xxyy}^{\widetilde{X}X \widetilde{Y}Y}  \otimes \ket{\psi_{x,y} }^{AB}\]
be a joint pure state of Alice and Bob, where registers $\widetilde{X}XA$
belong to Alice and registers $\widetilde{Y} Y B$ belong to Bob. Let
\[\I(X : BY \widetilde{Y})_\phi\leq\eps_1 \text{ and } \I(Y : AX \widetilde{X})_\phi \leq\eps_2.\]
Let $\ket{\phi_{x,y}} \defeq \ket{xxyy} \otimes \ket{\psi_{x,y}}$. There exist unitary operators,
$\{U_x\}_{x\in X}$ on $\widetilde{X}XA$ and $\{V_y \}_{y\in Y}$ on $\widetilde{Y} Y B $ such that
\[\E_{(x,y)\leftarrow\mu}
[\|\ketbra{\phi_{x,y}}{\phi_{x,y}} - (U_x \otimes V_y ) \ketbra{\phi}{\phi} (U^*_x \otimes V^*_y)\|_1
]
\leq 4\sqrt{\eps_1}+ 4\sqrt{\eps_2} + 2 \|\mu - \mu_X \otimes \mu_Y \|_1.\]
\end{fact}
\begin{fact}[Non-negativity \cite{marco_book}] \label{fact:non_negative} For states $\rho$ and $\sigma$, \[
 \D(\rho\|\sigma)\geq 0.
\]
\end{fact}
\begin{fact}[Pinsker’s Inequality \cite{wilde}]\label{fact:pinsker} For states $\rho$ and $\sigma$, \[\onenorm{\rho}{\sigma} \leq \sqrt{\D(\rho \|\sigma)}.\] 
\end{fact}
\begin{fact}\label{fact:approx_markov_sampling}
 Let $\rho^{XYR}=P_{XYR}$ be a classical state.  Let $\I(R:Y|X)_\rho \leq \eps$ and $\I(R:X|Y)_\rho \leq \eps$. Then, 
 \begin{equation*}
      \onenorm{ P_{XYR}}{P_{XY}P_{R|{X}}} \leq \sqrt{\eps},
  \end{equation*}
\begin{equation*}
       \onenorm{ P_{XYR}}{P_{XY}P_{R|{Y}}} \leq \sqrt{\eps}.
  \end{equation*}
\end{fact}
\begin{proof}We defer the proof of this Fact to the Appendix. \end{proof}

\begin{definition}[\cite{holenstein}] For two distributions $P_{XY}$ and $P_{X'Y 'ST}$ , we say $(X, Y )$ is $(1-\epsilon)$-embeddable
in $(X'S, Y 'T )$ if there exists a random variable $R$ on a set $\mathcal{R}$ independent of $XY$ and functions
$f_A : \mathcal{X} \times \mathcal{R} \rightarrow \mathcal{S}$ and $f_B : \mathcal{Y} \times \mathcal{R} \rightarrow \mathcal{T}$ , such that
\[\|P_{XY} f_A(X,R)f_B (Y,R) - P_{X'Y 'ST} \|_1 \leq\epsilon. \]
\end{definition}
\begin{fact}[\cite{holenstein}]\label{fact:hollenstein_sampling} If two distributions $P_{XY}$ and $P_{X'Y 'R'}$ satisfy \[
\|P_{X'Y 'R'} - P_{XY} P_{R'|X'} \|_1 \leq\epsilon \]\[\|P_{X'Y 'R'} - P_{XY} P_{R'|Y '} \|_1 \leq\epsilon,\]
then $(X, Y )$ is $(1 - 5\epsilon)$-embeddable in $(X'R', Y 'R')$.\end{fact}

\begin{fact}[Chernoff bound]\label{fact:Chernoff}
    Suppose $X_1, X_2,\ldots, X_n$ are independent random variables taking values in $\lbrace 0,1\rbrace$.
    Let $X = \sum\limits_{i=1}^{n} X_i$ and $\mu= \E(X)$. 
    Then, for any $\epsilon > 0$,
    \[ \Pr\left( \vert X - \mu \vert \geq \epsilon \mu  \right) \leq 2 e^{-\frac{\epsilon^2 \mu}{3}.}\]
    
\end{fact}
\begin{definition}[Markov-chain] \label{def:markov}
    We call a state (here $X$ is classical, $A$ and $B$ can be classical or quantum) of the form $\rho^{AXB}= \sum_x p_x \cdot \rho^{A}_x \otimes \ketbra{x}{x}^X \otimes \rho^{B}_x$, a cq-Markov-chain (with $X$ classical), denoted  $(A-X-B)_\rho$.  
\end{definition}
\begin{fact} \label{fact:recovery}
Let $\rho^{AXB}$ be a cq-Markov-chain $(A-X-B)_\rho$ with $X$ classical. There exists a CPTP recovery map $\Phi: X \rightarrow XB$ such that $\Phi(\rho^{AX}) = \rho^{AXB}$.
\end{fact}

\section{Framework for composability \label{sec:framework_for_composibility}}
In this section, we will describe the framework of composability that we will work in, for protocols that include only one phase. 
Some typical examples of such protocols are OT and coin tossing. Extension to multi-phase protocols (such as bit-commitment) follows similar ideas. 
We begin by formally defining the notion of functionality and then define the security of their corresponding protocols in the composable framework. 
\begin{definition}[Ideal (single phase) functionality \cite{Goldreich_2001}]
    For any two-party cryptographic task $f$ with parties $\alice$ and $\bob$, the ideal functionality of that task, referred to as $\textsc{ideal}(f)$, is described by the following:
\begin{enumerate}
    \item Input sets $\mathcal{A}$ and $\mathcal{B}$ for $\alice$ and $\bob$, respectively.
    \item Output sets $\mathcal{O}_{\alice}$ and $\mathcal{O}_{\bob}$
    for $\alice$ and $\bob$, respectively.
    \item The sets $\mathcal{A}, \mathcal{B}, \mathcal{O}_\alice$ and $\mathcal{O}_\bob$ may contain special symbols: empty string $\emptystring$ and abort $\perp$.
%    \item For $A \in \mathcal{A}$ and $B \in \mathcal{B}$ (joint) output distribution $(O_1 O_2)_{AB}$ supported on $\mathcal{O}_\alice \times \mathcal{O}_\bob$.
\end{enumerate}
$\textsc{ideal}(f)$ is defined as follows.
    \begin{enumerate}
        \item $\textsc{ideal}(f)$ takes as input $A\in \mathcal{A}$ and $B\in \mathcal{B}$ from $\alice$ and $\bob$, respectively.    
        \item $\textsc{ideal}(f)$ sends $O_A^\ideal$ to  $\alice$\footnote{Without loss of generality, we assume that $\alice$ receives her output before $\bob$ in the actual protocol.}. 
        \item $\alice$ sends $A_{\perp} \in \lbrace \emptystring, \perp \rbrace$ to $\textsc{ideal}(f)$.
        \item $\textsc{ideal}(f)$ sends $O_{B}^\ideal$ to $\bob$, with the property that $O_{B}^\ideal = \perp$ if $A_\perp=\perp$. 
    \end{enumerate}
    The above process can be expressed diagrammatically as below:
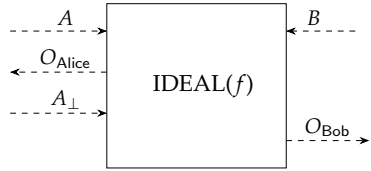
\begin{center}
         {\centering
\resizebox{0.3\textwidth}{!}{%
\begin{circuitikz}
\tikzstyle{every node}=[font=\normalsize]
\draw (5.25,14.75) rectangle (8.5,11.75);
\node [font=\large] at (7,13.25) {IDEAL($f$)};
\draw [->, >=Stealth, dashed] (3.5,14.25) -- (5.25,14.25);
\draw [->, >=Stealth, dashed] (9.75,14.25) -- (8.5,14.25);
\node [font=\normalsize] at (4.5,14.5) {$A$};
\node [font=\normalsize] at (9,14.5) {$B$};
\draw [->, >=Stealth, dashed] (5.25,13.5) -- (3.5,13.5);
\node [font=\normalsize] at (4.5,13.75) {$O_A^\ideal$};
\draw [->, >=Stealth, dashed] (3.5,12.75) -- (5.25,12.75);
\node [font=\normalsize] at (4.5,13) {$A_\perp$};
\draw [->, >=Stealth, dashed] (8.5,12.25) -- (10,12.25);
\node [font=\normalsize] at (9.25,12.5) {$O_B^\ideal$};
\end{circuitikz}
}}%
  \end{center}

\end{definition}

  %  Note that Alice is allowed to further input either $\emptyset$ or $\bot$ before ideal functionality provides an output to Bob at the end of the reveal phase. Also, if at any point any input from either party is $\bot$, then both subsequent outputs are $\bot$.

For the purposes of security, we will need to analyze various protocols.
So, let us first set up our notation for protocols.

\subsubsection*{Notation for a protocol}
Suppose we are interested in some two-party cryptographic task $f$. 
An (honest) protocol $\mathcal{P}_f$ for the task $f$ will be described as a tuple $\mathcal{P}_f$$(\alice, \bob, \device,\varphi)$. Here $\alice$ and $\bob$ start with a device $\device$ (with state $\rho$) and $\varphi$ contains inputs for both $\alice$ and $\bob$ as well as other work registers. In the honest protocol, the state of the device is independent of initial shared state, i.e.,  $\rho \otimes \varphi$.
                 \begin{table}[!h]
   
        \centering
            \begin{tabular}{  l c r }
            \hline 
             \\           $\alice^*$ (input $A$) &  & $\bob$ (input $B$)\\ \hline 
           \\
$A\widehat{A}$ & $\eta$ & $B\widehat{B}$
\\ \hdashline\\  
  {Internal computation (using $\device^*$)} & & {Internal computation (using $\device^*$)}\\
  & $\xrightarrow{communication}$ & \\
   & $\xleftarrow{communication}$ & \\
   & \vdots & \\
  % \\ & \large{input:$A_f, B_f$}

\\ \\ % \\ & \large{
{Output $O_A$}&& {Output $O_B$}\\
  \\
$\widehat{A}O_A$ & $\tau_{f}$ & $\widehat{B}BO_B$

\\ \hdashline\\  
    \end{tabular} 
        \caption{ {Protocol $\mathcal{P}_f$$\left( \alice^*, \bob, \device^*, \varphi^*\right)$}}
        \label{prot:DI-protocol} 
       
    \end{table}

Note that the above description is for the case when $\alice$ and $\bob$ are honest and the device is also ideal.
Any of the parties could also be deviating from their honest description. We will denote an arbitrary (cheating) behavior of any party with the asterisk mark.
For example, $\mathcal{P}_f$$\left( \alice^*, \bob, \device^*, \varphi^*\right)$ is a protocol in which $\alice^*$ is arbitrary (or cheating) and device and shared state are also different from the ideal description but $\bob$ is honest. 
Note that honest $\bob$ uses the shared state $\varphi^*$ only for getting its input, although (a cheating) $\alice^*$ can use all registers of the shared state $\varphi^*$ (that $\alice^*$ has access to) in the protocol.

%We call these \emph{secure simulators} and define them below:

     Consider an $\alice^*$ protocol $\mathcal{P}_f(\alice^*, \bob, {\device^*},\varphi^*)$ (given in Table~\ref{prot:DI-protocol}).
     The initial starting state (including the state of the device) is $\eta^{A\widehat{A}B\widehat{B}} $.
     Here $A$ and $B$ are the input registers and $\widehat{A}$ and $\widehat{B}$ represent all other registers with $\alice^*$ and $\bob$ respectively.
     Throughout this paper, for any state, we let $\widehat{A}$ and $\widehat{B}$ (or $\overline{A}$ and $\overline{B}$) denote all registers with $\alice$ and $\bob$ that we don't write explicitly (so they can denote different sets of registers in different states). 
     Let the state at the end of $\mathcal{P}_f(\alice^*, \bob, {\device^*}^{},\varphi^*)$ be $\tau_{f}^{\widehat{A}O_A \widehat{B}BO_B}$. 
     Here $O_A$ and $O_B$ are output registers of respective parties.
     In the notation for protocols (e.g. $\mathcal{P}_f(\alice^*, \bob, {\device^*}^{},\varphi^*)$), when $\device^*$ and $\varphi^*$ are clear from the context, we may drop them.
     For the purposes of readability, we will then refer to this simply as (a cheating) $\alice^*$ protocol for $f$. 

\subsubsection*{Notation for a simulator}
 When both  $\alice^*$ and $\device^*$ are dishonest, they can be clubbed together into a single cheating party $\alice^*$ (note there is still a state $\varphi^*$ from outside that contains the inputs for the protocol).  { We consider a simulator, \textsc{Sim}($\alice^*$) for cheating $\alice^*$ that has complete knowledge of the behavior (circuit) of $\alice^*$}. Since $\alice^*$ has full control over the devices in the device-independent setting, $\alice^*$'s circuit also has the description of the circuit used to create the devices. Note that the initial shared state between $\Sim$($\alice^*$) and $\bob$ in the simulator-based protocol is the same as that between $\alice^*$ and $\bob$ in the actual protocol.
   Similarly, in the case of cheating $\bob^*$, we have the simulator \textsc{Sim}($\bob^*$).
     We will use \textsc{Sim} to denote both \textsc{Sim}($\alice^*$) and \textsc{Sim}($\bob^*$) wherever the cheating party is clear from context.

         \begin{table}[!h]
   
        \centering
            \begin{tabular}{  l c r }
            \hline 
             \\            \textsc{Sim} (input $A$) &  & $\bob$ (input $B$)\\ \hline 
           \\
$A\widehat{A}$ & $\eta$ & $B\widehat{B}$
\\ \hdashline\\ 
           
{Internal computation} &  % \\ & \large{input:$A_f, B_f$}

\\  
& \multirow{4}{*}
{\centering
\resizebox{0.3\textwidth}{!}{%
\begin{circuitikz}
\tikzstyle{every node}=[font=\normalsize]
\draw (5.25,14.75) rectangle (8.5,11.75);
\node [font=\large] at (7,13.25) {IDEAL($f$)};
\draw [->, >=Stealth, dashed] (3.5,14.25) -- (5.25,14.25);
\draw [->, >=Stealth, dashed] (9.75,14.25) -- (8.5,14.25);
\node [font=\normalsize] at (4.5,14.5) {$A^\Sim$};
\node [font=\normalsize] at (9,14.75) {$B$};
\draw [->, >=Stealth, dashed] (5.25,13.5) -- (3.5,13.5);
\node [font=\normalsize] at (4.5,13.75) {$O_A^\ideal$};
\draw [->, >=Stealth, dashed] (3.5,12.75) -- (5.25,12.75);
\node [font=\normalsize] at (4.5,13) {$A_\perp$};
\draw [->, >=Stealth, dashed] (8.5,12.25) -- (10,12.25);
\node [font=\normalsize] at (9.25,12.5) {$O^\ideal_B$};
\end{circuitikz}
}%

 } &  \\

 & &  \\
  & &  \\
 Internal computation   & &  \\
   & &  \\
     & &  \\

{Output $O_A^\Sim$}&& {Output $O_B^\Sim = O^\ideal_B$}\\
  \\
$\widehat{A}O_A^\Sim$ & $\tau_{\textsc{Sim}}$ & $\widehat{B}BO_B^\Sim$

\\ \hdashline\\  
    \end{tabular} 
        \caption{ $\mathcal{P}_f({\textsc{Sim}, \textsc{ideal}(f), \bob )}$ with input $(A,B)$}
        \label{prot:Sim_sec_alice} 
       
    \end{table}

     Below, we describe correctness and security for honest $\bob$ (against any $\alice^*$); the other case, when $\alice$ is honest, can be written by just switching parties.  Consider the communication protocol between \textsc{Sim}, \textsc{ideal}($f$) and ${\bob}$ (Table~\ref{prot:Sim_sec_alice}).
We want that the state (on the relevant registers) in the real protocol is negligibly close to the state produced in the `ideal' situation. 
This intuitively means that the cheating $\alice^*$ cannot get any more information in the protocol than it could have gotten with a single call to the ideal functionality.
    
   \begin{definition}[Simulator-based security against any $\alice^*$]   \label{def:simulator_cheating_alice} We say that the protocol $\mathcal{P}_f$ is  secure (with security parameter $\lambda$)  if for all  $\alice^*$, there exists a simulator $\textsc{Sim}(\alice^*)$ (denoted as \textsc{Sim})  such that the following holds:

\[ \onenorm {\tau^{\widehat{A} O_A^\Sim O^\Sim_BB}_{\textsc{Sim}}} { \tau_{f}^{\widehat{A}O_AO_{B}B}} \leq {\negl({\lambda}}).\]

    \end{definition}
For the purposes of composability of protocols (see the composition theorem in the Appendix), we need a stronger notion of security than the usual simulator-based definition.

 \begin{definition}[Augmented-security] \label{def:augmented-secure}
   Consider the setting of Definition~\ref{def:simulator_cheating_alice} (Table~\ref{prot:DI-protocol} and~\ref{prot:Sim_sec_alice}).
  Let the initial state be $\eta^{A\wht{A}B\widehat{B}}$ between $\alice^*$ and $\bob$ where $A\wht{A}$  is with $\alice^*$ and $B$ is $\bob$'s input and $\widehat{B}$ is everything else with $\bob$.
  Let $ {{\eta}^{A\wht{A}B\widehat{B}R}}$ be an extension of  $\eta^{A\wht{A}B\widehat{B}}$ where {$R$ is an additional relevant register from the outer protocol}. 
  We say that protocol $\mathcal{P}_f$ is augmented-secure (with security parameter  {$\lambda$})  if for all  $\alice^*$, there exists a simulator $\textsc{Sim}(\alice^*)$ (denoted as \textsc{Sim})  such that the following holds:
  \[ \onenorm {\tau^{\widehat{A} O^\Sim_AO^\Sim_BBR}_{\textsc{Sim}}} { \tau_{f}^{\widehat{A}O_AO_{B}BR}} \leq {\negl(\lambda)}.\]

    \end{definition}

\begin{definition}[Correctness/robustness of the protocol] \label{def:simulator_correctness}
 Consider a protocol between two (honest) parties $\alice$ and $\bob$ using a device $\device = \otimes_i \device^i$. Define $$\text{$\mathfrak{D}_\eps$} \defeq \Big\{ \otimes_i \device^i_\eps ~|~  \device^i_\eps \in \mathcal{D}_\eps\left( \device^i\right) \Big\}.$$ For the purpose of correctness, given a security parameter $\lambda$, we want that if $\alice^*=\alice$ and $\device^* \in \mathfrak{D}_\eps$, then $A^\Sim=A$ and  $O_A^\Sim=O^\ideal_A$ with probability at least $1-{\negl(\lambda)}$ . 
  % Note that when both parties are honest {and devices are close to honest}, the condition $1$ above forces the inputs and outputs of the simulated protocol to be negligibly close to those of the ideal functionality. Also, Definition \ref{def:augmented-secure} forces the inputs and outputs of the simulated protocol to be negligibly close to those of the real protocol. Combining these ensures that for honest $\alice$ and $\bob$, the input-output behavior of the real protocol is negligibly close to the ideal functionality which gives us correctness. This also means that giving correctness for either $\textsc{Sim}(\alice^*, \device^*,\varphi^*)$ or $\textsc{Sim}(\bob^*, \device^*,\varphi^*)$ is sufficient to show correctness. 

\end{definition}

 The intuition behind the above definitions is as follows: Consider an `ideal' situation when the inputs $A$ and $B$ are fed directly into the ideal functionality $\textsc{ideal}(f)$  to produce outputs $O_A$ and $O_B$. Informally, we would like the `correctness' of a protocol to mean that when both parties are honest and the devices are close to their ideal specification, the input-output behaviour produced by the real protocol is negligibly close to the one produced in the `ideal' situation.
This informal idea of `correctness' is made rigorous by Definitions \ref{def:augmented-secure} and \ref{def:simulator_correctness}. To see this, note that Definition \ref{def:simulator_correctness} ensures that when $\device^* \in \mathfrak{D}_\eps$, inputs and outputs of the simulator-based protocol (with simulator $\textsc{Sim}(\alice)$) are negligibly close to those using the ideal functionality. Further note that Definition \ref{def:augmented-secure} forces the overall state produced by the simulator-based protocol (including the register $R$) to be negligibly close to the state produced by the real protocol. Thus, the overall inputs and outputs produced in the real protocol are negligibly close to the `ideal' state.
Note that for correctness, since we require
closeness of real and ideal protocols, it suffices to show Definition \ref{def:simulator_correctness} with 
either $\textsc{Sim}(\alice)$ or $\textsc{Sim}(\bob)$.

\begin{definition}[{DI Protocols}] \label{def:DIprot} Let $\eps>0$. 
    An $\eps$-robust device-independent (DI) protocol $\mcP_f$ (Table \ref{prot:DI-protocol}) for a cryptographic primitive $f$ between two (honest) parties, say $\alice$ and $\bob$ using a device $\device = \otimes_i \device^i$ and shared state $\varphi$ is specified as follows:
    \begin{enumerate}
        \item $\alice$ receives input $A\in \mathcal{A}$ and $\bob$ receives input $B\in \mathcal{B}$ from the shared state $\varphi$.
              {Both $\alice$ and $\bob$ also receive $1^\lambda$ where $\lambda$ denotes the security parameter.}
        \item They begin with the device $\device$ (with state $\rho$) with input and output ports possibly shared between them and a shared state $\varphi$. The overall starting state is represented as $\eta^{A\widehat{A}B\widehat{B}}=\rho\otimes\varphi$.
        \item During the protocol, $\alice$ and $\bob$ provide inputs to these devices, obtain outputs, and perform classical communication between themselves.
        \item At the end of the protocol they output $(O_A, O_B)\in \mathcal{O}_A \times \mathcal{O}_B$, where $O_A$ (and $O_B$) is the random variable corresponding to output of $\alice$ (and $\bob$).
      
        \end{enumerate}
        The above description is that of the honest protocol. In DI protocols, parties may also behave arbitrarily and thus we require the following efficiency, correctness and security guarantees:
        \begin{enumerate}
            \item (Efficiency):  The running time for all $\mathcal{P}_f\left(\alice, \bob, \device^* \in \mathfrak{D}_\eps,\varphi^*\right)$ is in ${\poly (\lambda)}$.
            
            \item (Security): The  simulators   $\textsc{Sim}(\alice^*)$ and $\textsc{Sim}(\bob^*)$ are secure according to Definition \ref{def:augmented-secure}.
            \item (Correctness): Protocol $\mathcal{P}_f$ is correct according to Definition \ref{def:simulator_correctness}.
             \end{enumerate}

             \end{definition}

\begin{definition}[Ideal functionality for oblivious transfer \cite{Goldreich_2001}]\label{def:idealOT}
    We refer to the ideal functionality for oblivious transfer as $\textsc{ideal}(\mathsf{OT})$.  $\alice$ inputs $(S_0,S_1) \in \brak{0,1, \perp}^2 $ and $\bob$ inputs $D \in \brak{0,1,\perp}$ to $\textsc{ideal}(\mathsf{OT})$. 
     
    \begin{enumerate}
        \item $\alice$ gets $O_{A}^\ideal=\emptystring$  from $\textsc{ideal}(\mathsf{OT})$ {except if $D=\perp$ in which case $O_{A}^\ideal=A_\perp=\perp$}.

        \item $\alice$ sends $A_\perp$ to $\textsc{ideal}(\mathsf{OT})$.
        
       \item If $A_\perp=\perp$ then $O_B^\ideal=\perp$. 
       Otherwise, $O_B^\ideal=S_D$.
       \end{enumerate} 

    \textbf{Remark:} According to the definition of $\textsc{ideal}(\mathsf{OT})$ functionality, $\alice$ is allowed to send $\bot$ in one (or both) of their input bits instead of just $\{0,1\}$. The provision to allow $\bot$ as one of the inputs is just to accommodate a stronger cheating behaviour by $\alice$ (similarly for $\bob$).

Note that honest $\alice$ aborts if one of the inputs is missing or $\bot$. A cheating $\alice$ can interact with honest $\bob$ as follows: $\alice$ creates devices that produce the correct output when $\bob$ chooses the input bit $D=0$ but produce garbage when $D=1$.
In this case, we expect a secure $\bob$ to execute a correct protocol for $D=0$ but abort the protocol for $D=1$.
We capture such behaviour by allowing inputs of the form $(0, \perp)$.
To summarise, to accommodate stronger cheating behaviour in DI settings, we also need to allow additional inputs of the form $(0,\perp)$ or $(\perp,0)$.   
 
\end{definition}

\section{Device-independent oblivious transfer}\label{section:OT}

In this section, we provide a device-independent OT protocol (Table \ref{table:OT}). 
For the purpose of protocol $\mathcal{Q}$, we consider $\alice$'s input $(S_0,S_1)\in \{0,1\}^2$  and $\bob$'s input $D \in \{0,1\}$. Note if (honest) $\alice$ receives only 1 input from above (say $S_0$) while the other input is missing or $\perp$, she will not even invoke the OT protocol $\mathcal{Q}$ (similarly for $\bob$). Let $\lambda$ be the security parameter.  Alice and Bob share $n=O(\log^{10}(\lambda))$ devices. Throughout this paper, both $\alice$ and $\bob$ also receive $1^\lambda$ as input, but we don't write it explicitly for ease of notation. Honest $\alice$ and $\bob$ are efficient (in input-size) classical parties, although cheating $\alice$ and $\bob$ may have unbounded quantum powers. 

The idea behind the protocol is as follows: At the beginning of the protocol, there is a test phase where $\alice$ selects a random subset of devices $S_{test}\subseteq [n]$ to test the Magic Square predicate and selects both her and $\bob$'s inputs for $S_{test}$. After the test phase, during an honest execution of the protocol, $\alice$ generates $X_{\overline{S_{test}}}$ uniformly from $\brak{0,1,2}^{n-S_{test}}$ and inputs them into her Magic Square devices. $\bob$, on the other hand, inputs $D$ into all of his remaining Magic Square devices. $\alice$ will use random bits extracted from the first two columns of the output of her Magic Square devices, $A_{\overline{S_{test}}}$ (see Table \ref{table:OT}), to obfuscate $S_0$ and $S_1$ using a strong seeded-extractor $\ext(\cdot,\cdot)$. $\bob$ can now decode $S_D$ (with high probability) using syndrome decoding and the Magic Square winning guarantee (from the test phase). 

Since the inputs to our security game (Table \ref{table:ot-single-copy}) are not independent, we use anchoring ideas. For this, we introduce a new symbol $*$ and set $\wht{X}_i=*$ with some small probability $\delta_*>0$. This ensures that conditioning on $\wht{X}_i=*$, the input distribution for the $i$-th coordinate becomes independent. 

\subsection*{Error-correcting code parameters} 

 We refer to the checks that determine the value of $O_A$ in Table \ref{algo:ot_cheating_bob} as the "test phase checks". Let $\test=1$ denote the event where $O_A\neq \bot$ (we also refer to this event as "{test phase passed}").

In the protocol $\mcQ$, we use a linear code whose parity check matrix is $\eH$.
Let the rate and the relative distance of the code be $r$ and $\gamma$ (the input to the code is of size $m= |S_{prot}|$). Let $\delta_*>0$ be a small constant as defined in Claim \ref{claim:main}.  
 Let $\delta_t=\delta+(1-\delta)\delta_*$. Let $\delta_1>0$ be a small constant. We know that the test passes only if \begin{align}
(1-\delta_1)\delta_t n\leq |S_{test}|+|S_*| \leq (1+\delta_1)\delta_t n. \label{eq:size_of_S_test}
\end{align} This happens with $1-\negl(\lambda)$ probability using Fact \ref{fact:Chernoff}. Whenever the test phase passes, we want our code to satisfy the following (let $\delta'>0$ be some small constant):
\begin{itemize}
    \item We want the rate of the code to satisfy the following for proving $\alice$'s security. This is needed so that when $\alice$ sends across $W=W_0W_1$ to $\bob$, the min-entropy loss (which is equal to $|W|$) in Claim \ref{claim:ot_entropy} is small.
\begin{align}
    &\frac{\delta'n}{40}\geq 2(1-r)\cdot |S_{prot}|. \label{eq:ecc_rate}
     \\ \implies& 1-r \leq \frac{\delta'n}{80|S_{prot}|}. \nonumber
   \\ \implies& 1-r \leq \frac{\delta'n}{80(n-|S_{test}|-|S_*|)}. &\mbox{($n=|S_{test}|+|S_{*}|+|S_{prot}|$)}  \nonumber
    \\ \implies& 1-r \leq \frac{\delta'n}{80n(1-(1+\delta_1)\delta_t)}. &\mbox{(eq. \eqref{eq:size_of_S_test})}\nonumber
     \\ \implies& r \geq 1-\frac{\delta'}{80(1-(1+\delta_1)\delta_t)}.\nonumber
\end{align} 

\item  To ensure the robustness of our protocol, we want 
$(1+c_r)\eps<\gamma/2$ (see Lemma \ref{lem:sim_alice_correctness}). 
\end{itemize}
Such a code is known to exist from Fact~\ref{fact:existance_of_codes} for any $\delta'>0$ on setting a small enough $\delta,\delta_*, \eps >0$ (which gives a small enough $\delta_t >0$). We use this setting of $\delta$ and $\delta_*$ in Table \ref{table:OT} and set $\kappa=2\eps$.
From Fact~\ref{fact:syndrome_decoding}, we also know that there exists an efficient syndrome decoding algorithm for this code, which we denote as $\textsc{SyndDec}$. 
We will also use the strong extractor, $\ext(\cdot,\cdot)$ from Fact~\ref{fact:strong_extractor} with $m=|S_{prot}|=O(\log^{10}(\lambda)),$ $ t=O(\log^6(\lambda))$ and error atmost $2^{-\log^3(\lambda)}$.

\begin{table}     \centering \resizebox{15cm}{10.5cm}{
            \begin{tabular}{  l c r }
            \hline \\
            \textbf{$\alice$ (with input $(S_0,S_1)$)} & & \textbf{$\bob$ (with input $D$)} \\
            \hline  \\ $S_0S_1\widehat{A}$ & {$\rho \bigotimes \varphi$} & $D\widehat{B}$\\
 \hdashline 
           \\
           $n=O(\log^{10}(\lambda))$\\
            ${X} \leftarrow_U \lbrace 0,1,2\rbrace^{n}$ \\   for $i$ in [n]: 
            \\ \hspace{1cm} $C_i \leftarrow \mathrm{Bern}(1-\delta)$ \\
            \hspace{1cm} $Y^{\alice}_i \leftarrow_U \lbrace 0,1,2\rbrace$ & 
            \\ $S_{test}=\lbrace i\in [n] : C_i=0 \rbrace$  &$\xrightarrow{S_{test}, Y^{\alice}}$\\ for $i$ in $[n]$: 
             & & for $i$ in $S_{test}$: \hspace{3.6cm}
            \\ \tab$\wtt{C}_i \leftarrow \mathrm{Bern}(1-\delta_*)$
            && $Y_i=Y^{\alice}_i$ \tab  \tab 
            \\\tab if {$\wtt{C}_i=1, \wht{X}_i=X_i$}&& Obtain $B_{S_{test}} =\MS^B(Y_{S_{test}})$ \tab\\\tab else $\wht{X}_i=*$
            && for $i$ not in $S_{test}$: \hspace{3cm}
            \\$S_{*}=\lbrace i\in [n]\setminus S_{test} : \wtt{C}_i=0 \rbrace$  &&  $Y_i =D$ \tab \tab
            \\ Obtain $A_{[n]}=\MS^A(X_{[n]})$&&  Obtain $B_{[n]\setminus S_{test}}=\MS^B(Y_{[n]\setminus S_{test}})$ \tab 
\\ $S_{prot}=[n]\setminus (S_*\cup S_{test})$\\

& \textbf{DELAY} & \\\\   
 & $\xleftarrow{B_{S_{test}}}$ &   \\ {$\delta_t=\delta+(1-\delta)\delta_*.$}\\
    {if $ \vert S_{test} \vert+\vert S_{*}\vert \leq (1-\delta_1)\delta_t n$, then $O_A=\perp$ }
      \\      
      {if $ \vert S_{test} \vert+\vert S_{*} \vert \geq (1+\delta_1)\delta_t n$, then $O_A=\perp$ } \\\\
      for $i\in S_{test}$: \tab \tab\\
\hspace{1cm} $Z_i=\Id\left(\MS^i(X_i,Y^{\alice}_i,A_i,B_i)=1 \lor \wht{X}_i=*\right)$\\
if $\sum\limits_{i\in S_{test}}Z_i<(1-{\kappa })|S_{test}|)$,
\\
\hspace{1cm} $O_A=\perp$
\\ \hspace{1cm}{ $F_0=F_1=T_0=T_1=\perp$} \\ \hspace{1cm}{ $X=W_0=W_1=\perp$}
\\ else  \\\hspace{1cm} $O_A=\emptystring$

\\ \hspace{1cm}$T_0,T_1\leftarrow_U\lbrace0,1\rbrace^{O(\log^6(\lambda))}$ \\\hspace{1cm}
 $R_{0} ={A}_i\left(0\right)\Big\rvert_{i \in S_{prot}}$ & & \\\hspace{1cm}
    $W_{0} = H\left(R_{0}\right)$ & & \\
 \hspace{1cm} $R_{1} ={A}_i\left(1\right)\Big\rvert_{i \in S_{prot}}$ & & \\\hspace{1cm}
    $W_{1} = H\left(R_{1}\right)$ & & \\\hspace{1cm}
         $F_0 = S_0 \oplus \ext\left(R_0,T_0\right)$ & & \\\hspace{1cm}
          $F_1 = S_1 \oplus \ext\left(R_{1},T_1\right)$ & & \\
     & $\xrightarrow{{F_0,F_1,T_0, T_1, \wht{X}}}$ & \\
   & $\xrightarrow{{W_0,W_{1}}}$ & \\
    
&&  $R^\prime ={B}_i\left(\wht{X}_i\right)\Big\rvert_{i \in  S_{prot}}$ \tab \hspace{0.0cm}  \\
&&   \text{    If }$\textsc{SyndDec}\left(\eH(R^\prime)-W_{D}\right)=\bot$,   \quad \quad  \\
  && \tab $O_B=\bot$  \quad \quad \quad \quad \tab   \quad \tab  \\

   && Otherwise, $E^\prime = \textsc{SyndDec}\left(\eH(R^\prime) - W_{D}\right)$  \\
  && \tab $L^\prime = R^\prime+E^\prime$ \tab  \quad \quad \tab  \\
       Output $O_{A}$  && Output $O_B= F_D \oplus \ext(L^\prime,T_D)$.\,\,\hspace{1cm} 
      \end{tabular} }
        \caption{Protocol $\mathcal{Q}$-
OT with robust MS devices}
        \label{table:OT} 
    \end{table}

   Our main theorem can then be stated as follows:
\begin{theorem}\label{thm:main_theorem_rob_ot}
         There exists a constant $\eps >0$ such that Protocol $\mathcal{Q}$ (given in Table~\ref{table:OT}) is an $\eps$-robust DI OT protocol as per Definition \ref{def:DIprot}.        
    \end{theorem}

\begin{proof}
    Let $\lambda$ be the security parameter. Our protocol requires only $n=\polylog(\lambda)$  number of devices to obtain negligible (in $\lambda$) correctness and security error, and runs in time $\poly(n)=\polylog(\lambda)$. Thus, the efficiency of the protocol is clear from the construction. We now need to show the following properties of Protocol $\mathcal{Q}$ for it to be an $\eps$-robust DI-protocol as per Definition \ref{def:DIprot}:
   \begin{enumerate}
   \item Correctness: In Lemma~\ref{lem:sim_alice_correctness},
        we show that the protocol is correct.
       \item Security: In Lemma~\ref{lem:sim_alice_security} and Lemma~\ref{lem:sim_bob_security}, we show that the protocol is secure according to Definition \ref{def:simulator_cheating_alice}. This, along with Claim \ref{claim:non-aug-to-aug} gives us composable security according to Definition \ref{def:augmented-secure}.
   \end{enumerate}
\end{proof}

\subsection{Security (Cheating $\alice$)}
\begin{table}[H]     \centering
            \begin{tabular}{  l c r }
            \hline \\
            \textbf{$\alice^*$ (with input $(S_0,S_1)$)} & & \textbf{$\bob$ (with input $D$)} \\
            \hline \\
    \\ $S_0S_1\widehat{A}$ & $\eta^*_\mathcal{Q}$ & $D\widehat{B}$\\
 \hdashline  \\ Generate $S_{test}, Y^{\alice}$  &$\xrightarrow{S_{test}, Y^{\alice}}$\\ & & for $i$ in $S_{test}$: \hspace{3.6cm}
            \\
            && $Y_i=Y^{\alice}_i$ \tab  \tab 
            \\&& Obtain $B_{S_{test}} =\MS^B(Y_{S_{test}})$ \tab\\
            && for $i$ not in $S_{test}$: \hspace{3cm}
            \\ &&  $Y_i =D$ \tab \tab
            \\Generate $A_{[n]}$&&  Obtain $B_{[n]\setminus S_{test}}=\MS^B(Y_{[n]\setminus S_{test}})$ \tab 
\\
& \textbf{DELAY} & \\\\
& $\xleftarrow{B_{S_{test}}}$ &   \\ 
$\widehat{\widehat{A}}$ & $\kappa_\mathcal{Q}$ &  {$D \overline{B}\widehat{B}$}\\
 \hdashline \\

     Generate $F_0,F_1,T_0,T_1,\wht{X},W_0,W_1$   & & \\
         & & \\
     & $\xrightarrow{{F_0,F_1,T_0, T_1, \wht{X}}}$ & \\
   & $\xrightarrow{{W_0,W_{1}}}$ & \\
    
&&  $R^\prime ={B}_i\left(\wht{X}_i\right)\Big\rvert_{i \in S_{prot}}$ \tab \hspace{0.0cm}  \\
&&   \text{    If }$\textsc{SyndDec}\left(\eH(R^\prime)-W_{D}\right)=\bot$,   \quad \quad \\
  && \tab $O_B=\bot$  \quad \quad \quad \quad \tab   \quad \tab  \\

   && Otherwise $E^\prime = \textsc{SyndDec}\left(\eH(R^\prime)-W_{D}\right)$  \\
  && \tab $L^\prime = R^\prime+E^\prime$ \tab  \quad \quad \tab  \\
     Generate $O_A$   && \text{Output $O_B= F_D \oplus \ext(L^\prime,T_D)$.}\hspace{0cm}\\\\    \\ $S_0S_1O_A\widehat{A}$ & $\tau_\mathcal{Q}$ & $DO_B\overline{B}\widehat{B}$\\
 \hdashline \\
\textbf{}

      \end{tabular} 
        \caption{Cheating $\alice^*$ protocol}
        \label{algo:ot_cheating_alice} 
    \end{table}
Since $\bob$ enters its input $D$ in the devices after generating the output for the test devices $B_{S_{test}}$ and the only communication from $\bob$ to $\alice^*$ is $B_{S_{test}}$ (which is independent of $D$ as it is generated before $D$ is input into the devices), security against cheating $\alice^*$ is relatively straightforward.

    \subsection*{Simulator for  
cheating $\alice^*$}
\noindent We denote the simulator   \textsc{Sim}($\alice^*$) as \textsc{Sim} for readability. From our no-memory assumption on the devices, after the \textbf{DELAY} in Protocol $\mathcal{Q}$, the joint state at the end of protocol $\mathcal{Q}$ between $\alice^*$ and $\bob$, $\tau_\mathcal{Q}^{ S_0S_1 \widehat{A}O_AD\widehat{B}\overline{B}O_{B}}$ is classical which we denote as $P_{\mathcal{Q}}(S_0S_1 \widehat{A}O_AD\widehat{B}\overline{B}O_{B})$  or $(S_0S_1 \widehat{A}O_AD\widehat{B}\overline{B}O_{B})^{\tau_\mathcal{Q}}$. To denote the corresponding state in case of simulator, we use the notation $P_{\textsc{Sim}}(S_0S_1 \widehat{A}O^\Sim_AD\widehat{B}O^\Sim_{B})$ or $(S_0S_1 \widehat{A}O^\Sim_AD\widehat{B}O^\Sim_{B})^{\tau_\textsc{Sim}}$.
We will also not explicitly include workspace registers for the simulator, as they are not relevant for security. 

 \begin{table}[!h]      \centering
            \begin{tabular}{  l c r }
            \hline \\
            \textbf{\textsc{Sim} (with input $(S_0,S_1)$)} & & \textbf{$\bob$ (with input $D$)} \\
            \hline \\
              $S_0S_1\widehat{A}$ & $\eta^*_\Sim$ & $D{\widehat{B}}$\\
 \hdashline \\

   $\Sim$ follows $\alice^*$ and $\bob$'s \\
   behavior from Table \ref{algo:ot_cheating_alice} to  set\\ $(S_0S_1O_A\widehat{A})^{\tau_\Sim}=(S_0S_1O_A\widehat{A})^{\tau_\mathcal{Q}}$\\ \\
 
$\forall (\widehat{a}o_As_0s_1)$,
 \\
 $(\widehat{S}_0|\widehat{a}o_As_0s_1)^{}=(O_{B}|\widehat{a}o_As_0s_1,D=0)^{\tau_\mathcal{Q}}$ \\
 $(\widehat{S}_1|\widehat{a}o_As_0s_1)=(O_{B}|\widehat{a}o_As_0s_1,D=1)^{\tau_\mathcal{Q}}$ \\ \\

& 
\centering
\resizebox{0.3\textwidth}{!}{%
\begin{circuitikz}
\tikzstyle{every node}=[font=\normalsize]
\draw (5.25,14.75) rectangle (8.5,11.75);
\node [font=\large] at (7,13.25) {IDEAL($OT$)};
\draw [->, >=Stealth, dashed] (3.5,14.25) -- (5.25,14.25);
\draw [->, >=Stealth, dashed] (9.75,14.25) -- (8.5,14.25);
\node [font=\normalsize] at (4.5,14.5) {$(\widehat{S}_0,\widehat{S}_1)$};
\node [font=\normalsize] at (9,14.75) {$D$};
\draw [->, >=Stealth, dashed] (5.25,13.5) -- (3.5,13.5);
\node [font=\normalsize] at (4.5,13.75) {$O^\ideal_A$};
\draw [->, >=Stealth, dashed] (3.5,12.75) -- (5.25,12.75);
\node [font=\normalsize] at (4.2,13) {{$A_\perp=\emptystring$} };
\draw [->, >=Stealth, dashed] (8.5,12.25) -- (10,12.25);
\node [font=\normalsize] at (9.25,12.5) {$O^\ideal_B$};
\end{circuitikz}
}%

 & \\

 Output {$O_A^\Sim=O_A$} && Output $O_B^\Sim=O^\ideal_B$ \\  \\
  $S_0S_1O_A^\Sim\widehat{A}$ & $\tau_{\textsc{Sim}}$ & $DO_B^\Sim\widehat{B}$\\
 \hdashline \\
      \end{tabular} 
        \caption{Simulator for cheating $\alice^*$: $\mathcal{Q} \left(\textsc{Sim}, \bob\right)$.}
        \label{algo:alice_sim} 
    \end{table}

\begin{lemma}\label{lem:sim_alice_security}
   Protocol $\mathcal{Q}$ (Table \ref{table:OT})  with cheating $\alice^*$ is secure according to Definition~\ref{def:simulator_cheating_alice}. In other words, 
    \[ \onenorm {\tau^{\widehat{A} O_A^\Sim O^\Sim_BS_0S_1D}_{\textsc{Sim}}} { \tau_\mathcal{Q}^{\widehat{A}O_AO_{B}S_0S_1D}} \leq {\negl(\lambda)}.\]
\end{lemma}

\begin{proof}

\suppress{
Since $S_0,S_1$ and $D$ are classical, we can write 
\begin{equation}
    \label{eq:sim_alice_cq1}
\tau^{\widehat{A} O_AO_BS_0S_1D}_{\mathcal{Q}}=\sum_{d}P_\mathcal{Q}(d)\ketbra{d}{d}\otimes{\tau_\mathcal{Q}^{\widehat{A} O_AS_0S_1O_Bd}}.
\end{equation}
Since the distribution of inputs remains same when we replace $\alice^*$ with \textsc{Sim}, we get
\begin{equation}
    \label{eq:sim_alice_cq2}
\tau^{\widehat{A} O_AO_BS_0S_1D}_{\textsc{Sim}}=\sum_{d}P_\mathcal{Q}(d)\ketbra{d}{d}\otimes{\tau_\textsc{Sim}^{\widehat{A} O_AS_0S_1O_Bd}}.
\end{equation}}
In both Protocol $\mathcal{Q}$ and the simulated protocol, the initial starting state is the same. \[
{\eta^*_\Sim}^{S_0S_1\widehat{A}D{\widehat{B}}}={\eta^*_\mathcal{Q}}^{S_0S_1\widehat{A}D\widehat{B}}.
\]
Since $\bob$ is honest and $\Sim$ has complete knowledge of the circuit of $\alice^*$ (and the devices), $\Sim$ can generate $B_{test}$ on its own. $\Sim$ then checks the test condition to generate $A_\perp=O_A$. Thus, the simulator can perform the test phase checks on its own (without any communication from $\bob$). Other than the test phase, since there is no communication from $\bob$ to $\alice^*$ in protocol $\mathcal{Q}$  from Table \ref{algo:alice_sim}, we see that $\Sim$ does the same operations as $\alice^*$  on the starting state $\eta^*$ to give $\tau$. Thus,  

\begin{equation}\label{eq:sim_alice_ot_main}
    (\widehat{A}O_AS_0S_1D)^{\tau_\Sim}=(\widehat{A}O_AS_0S_1D)^{\tau_\mathcal{Q}}.
\end{equation}
Since  \textsc{Sim} does not know $D$, it generates $\forall (\widehat{a}o_As_0s_1)$:
 \begin{equation}\label{eq_alicesim_conditional_1} (\widehat{S}_0|\widehat{a}o_As_0s_1)^{\tau_\Sim}=(O_{B}|\widehat{a}o_As_0s_1,D=0)^{\tau_\mathcal{Q} }.\end{equation}
 \begin{equation}\label{eq_alicesim_conditional_2}
(\widehat{S}_1|\widehat{a}o_As_0s_1)^{\tau_\Sim}=(O_{B}|\widehat{a}o_As_0s_1,D=1)^{\tau_\mathcal{Q} }.\end{equation}
This is possible as $\Sim$ knows the description of the circuit used to create the devices and the \textbf{DELAY} has already taken place.

Since $A_\perp=\emptystring$, we have $(O_B^\Sim)^{\tau_\Sim}=(\widehat{S}_D)^{\tau_\Sim}$. Hence, from eqs. \eqref{eq_alicesim_conditional_1} and  \eqref{eq_alicesim_conditional_2},  we get $\forall do_B\widehat{a}o_As_0s_1$
\begin{equation}\label{eq:sim_alice_ot_main_2}
    P_{\textsc{Sim}}(O_B^\Sim=o_b|\widehat{a}o_As_0s_1d) =P_{\mathcal{Q}}(O_B=o_b|\widehat{a}o_As_0s_1d).\end{equation} 
    Now,
{\small\begin{align} 
 &\onenorm {{\tau^{\widehat{A} O^\Sim_AO^\Sim_BS_0S_1D}_{\textsc{Sim}}}} { {\tau_\mathcal{Q}^{\widehat{A}O_AO_{B}S_0S_1D}}}\nonumber\\
       &=\onenorm {{\tau^{\widehat{A} O_AO^\Sim_BS_0S_1D}_{\textsc{Sim}}}} { {\tau_\mathcal{Q}^{\widehat{A}O_AO_{B}S_0S_1D}}}\nonumber &\mbox{(since $\Sim$ sets $O^\Sim_A=O_A$)}\\&=   \onenorm {(\widehat{A}O_A O^\Sim_{B}S_0S_1D)^{\tau_\Sim}}{ (\widehat{A}O_AO_{B}S_0S_1D)^{\tau_\mathcal{Q}}} \nonumber\\  &= \sum_{\widehat{a}o_Ao_Bs_0s_1d}\Big| \Big(P_{\textsc{Sim}}(\widehat{A}O_AS_0S_1D=\widehat{a}o_As_0s_1d)\cdot P_{\textsc{Sim}}(O^\Sim_B=o_b|\widehat{a}o_As_0s_1d)\nonumber\\& \tab\tab  -P_{\mathcal{Q}}(\widehat{A}O_AS_0S_1D=\widehat{a}o_As_0s_1d)\cdot P_{\mathcal{Q}}(O_B=o_b|\widehat{a}o_As_0s_1d)\Big) \Big|\nonumber\\ &=0.  &\mbox{(from eqs. \eqref{eq:sim_alice_ot_main} and \eqref{eq:sim_alice_ot_main_2})}\nonumber
      \end{align}}

\end{proof}

\subsection{Correctness}
Consider the case when  $\alice^*=\alice$, $\bob^*=\bob$ and devices are in $\mathfrak{D}_\eps$. 

\begin{claim} \label{claim:OT-correctness}
   Let $\eps>0$. 
     For all $\mcQ\left(\alice,\bob, \device^* \in \mathfrak{D}_\eps\right)$, 
            \[
            \forall D, S_0,S_1 \in \{0,1\}: \Pr\left( O_B = S_D \right) \geq 1-{\negl(\lambda)}.
            \]
    % with respect to the class $\mathfrak{D}_{\eps_r}$.
\end{claim}
\begin{proof}

First suppose $d_H(R_{D},R^\prime)<\frac{\gamma}{2}|S_{prot}|$. Consider
\begin{align}
  L^\prime&=R^\prime+E^\prime \nonumber\\
&=R^\prime+\textsc{SyndDec}\left(\eH(R^\prime)-W_{D}\right) \nonumber&\mbox{(from definition of $E^\prime$)}\\
&=R^\prime+\textsc{SyndDec}\left(\eH(R^\prime)-\eH(R_D)\right) \nonumber&\mbox{(from definition of $W_D$)}\\
&=R^\prime+\textsc{SyndDec}\left(\eH(R^\prime-R_D)\right) \nonumber&\mbox{(since $\eH$ is a linear map)}\\
&=R^\prime+(R^\prime-R_D)&\mbox{(since $d_H(R_{D},R^\prime)<\frac{\gamma}{2}|S_{prot}|$ )}\nonumber\\ &=R_D.\label{eq:Syndec_correctness_1}
\end{align}
  This means that the bit that $\bob$ receives is given by,
\begin{align*} O_B&=F_D\oplus\ext(L^\prime,T_D)\\
&=F_D\oplus\ext(R_D,T_D) &\mbox{(from eq.~\eqref{eq:Syndec_correctness_1})}
\\&=S_D\oplus\ext(R_D,T_D)\oplus\ext(R_D,T_D) &\mbox{(from definition of $F_D$)}\\&=S_D.
\end{align*}  
Thus, the probability that $\bob$ does not receive the correct bit is given by,  
    \begin{align}
        \Pr\left( \text{$O_B\neq S_D$ }\right) & \leq  \Pr \left( ~   d_{H}(R_{D},R^\prime) \geq \frac{\gamma}{2} |S_{prot}| \right) .\label{eq:ot_correctness_1}
    \end{align}
 Note that
 $R_{D} ={A}_i\left(D\right)\Big\rvert_{S_{prot}}$  and $R^\prime_{} ={B}_i\left(X_i\right)\Big\rvert_{S_{prot}}$ where $\alice$ and $\bob$ input $X_{S_{prot}}$ and $Y_{S_{prot}}=D^{|S_{prot}|}$ to the Magic Square devices (outside $S_{test}$ and $S_*$) respectively. 
\begin{align}
  \nonumber  \Pr\left(   d_{H}\left(R_{D},R^\prime\right) \geq \frac{\gamma}{2} |S_{prot}| \right) 
   \nonumber = &  \Pr\left( \sum_{i\in S_{prot}} \Id_{A_{i}(Y_i) \neq B_{i}(X_i)}   \geq \frac{\gamma}{2} |S_{prot}|    \right) &\mbox{(since $Y_i=D$) } \\
    \leq & \Pr\left( \sum_{i\in S_{prot}} \Id_{A_{i}(Y_i) \neq B_{i}(X_i)}   \geq (1+c_r) \eps |S_{prot}| \right) \label{OT_correctness1}\\
     \leq &~  {{2^{-c\cdot |S_{prot}|}} }\label{OT_correctness2}\\
     \leq &~{\negl(\lambda)}, \label{eq:ot_correctness_3}
\end{align}
{where $c$ is some absolute constant depending on $\eps$ and $c_r$. For eq.~\eqref{OT_correctness1} to hold,  we choose a small enough $c_r$ and appropriate $\eps$ such that $(1+c_r)\eps<\gamma/2$.}
 Eq.~\eqref{OT_correctness2} follows due to Fact \ref{fact:Chernoff}, the independence on the devices coming from $\mathfrak{D}_\eps$, and the fact that the devices win with probability at least $1-\eps$. Eq.~\eqref{eq:ot_correctness_3} follows as $(1-\delta_1)\delta_t n\leq|S_{prot}|\leq (1+\delta_1)\delta_t n$ with $1-\negl(\lambda)$ probability (using Fact \ref{fact:Chernoff}). 
The desired now follows from eq.~\eqref{eq:ot_correctness_1} and \eqref{eq:ot_correctness_3}.
\end{proof}

\begin{lemma}\label{lem:sim_alice_correctness}
 Protocol $\mathcal{Q}$ (Table \ref{table:OT})  is correct  according to Definition~\ref{def:simulator_correctness}.
\end{lemma}

\begin{proof}
Consider the simulator for  $\alice^*$ from Table \ref{algo:alice_sim}. Let $\alice^*=\alice$ and $\bob^*=\bob$. 
\begin{enumerate}
\item We see from Table \ref{algo:alice_sim} that $O_A^\ideal=\emptystring$. Also when $\device^*\in \mathfrak{D}_\eps$, with probability at least $1-\negl(\lambda)$, $O_A^\textsc{Sim}=O_A=\emptystring$ for a small enough $\eps>0$ using Fact \ref{fact:Chernoff} (since $\kappa=2\eps$). Thus with probability at least $1-\negl(\lambda)$, $O_A^\textsc{Sim}=O_A^\ideal$.

\item  From Claim \ref{claim:OT-correctness}, we see that with probability at least $1-{\negl(\lambda)}$,  $\widehat{S}_0=S_0$ (and  $\widehat{S}_1=S_1$) whenever $\device^*\in \mathfrak{D}_\eps$.

\end{enumerate}
Combining these gives us 
correctness according to Definition~\ref{def:simulator_correctness}.

\end{proof}

\subsection{Security (Cheating $\bob$)}
\noindent We first describe the cheating $\bob^*$ protocol in Table \ref{algo:ot_cheating_bob}. We define events $\mathcal{L}_0(\wht{x}\widehat{b}\overline{b})$ and $\mathcal{L}_1(\wht{x}\widehat{b}\overline{b})$ as in eqs. \eqref{eq:def_L0} and \eqref{eq:def_L1}.
  We will use a shorthand $W$ to denote the combined vector $W_{0}W_{1}$.

\begin{table}[H]     \centering \resizebox{12.5cm}{9.3cm}{
            \begin{tabular}{  l c r }
            \hline \\
            \textbf{$\alice$ (with input $(S_0,S_1)$)} & & \textbf{$\bob^*$ (with input $D$)} \\
            \hline  \\ $S_0S_1\overline{A} $& $\kappa_\mcQ$ & $D\overline{B}$\\
 \hdashline 
           \\
           $n=O(\log^{10}(\lambda))$\\
            ${X} \leftarrow_U \lbrace 0,1,2\rbrace^{n}$ \\   for $i$ in [n]: 
            \\ \hspace{1cm} $C_i \leftarrow \mathrm{Bern}(1-\delta)$ \\
            \hspace{1cm} $Y^{\alice}_i \leftarrow_U \lbrace 0,1,2\rbrace$ & 
            \\ $S_{test}=\lbrace i\in [n] : C_i=0 \rbrace$  &$\xrightarrow{S_{test}, Y^{\alice}}$\\ for $i$ in $[n]$: 
             & &
            \\ \tab$\wtt{C}_i \leftarrow \mathrm{Bern}(1-\delta_*)$
            &&  \tab  \tab 
            \\\tab if $\wtt{C}_i=1, \wht{X}_i=X_i$&& Generate $B_{S_{test}}$ \tab\\\tab else $\wht{X}_i=*$
            &&  Generate $B_{[n]\setminus S_{test}}$
            \\$S_{*}=\lbrace i\in [n]\setminus S_{test} : \wtt{C}_i=0 \rbrace$  &&  
            \\ Obtain $A_{[n]}=\MS^A(X_{[n]})$&&  \tab 
\\ $S_{prot}=[n]\setminus (S_*\cup S_{test})$\\

& \textbf{DELAY} & \\\\   
 & $\xleftarrow{B_{S_{test}}}$ &   \\ {$\delta_t=\delta+(1-\delta)\delta_*.$}\\
    {if $ \vert S_{test} \vert+\vert S_{*}\vert \leq (1-\delta_1)\delta_t n$, then $O_A=\perp$ }
      \\      
      {if $ \vert S_{test} \vert+\vert S_{*} \vert \geq (1+\delta_1)\delta_t n$, then $O_A=\perp$ } \\\\
      for $i\in S_{test}$: \tab \tab\\
\hspace{1cm} $Z_i=\Id\left(\MS^i(X_i,Y^{\alice}_i,A_i,B_i)=1 \lor \wht{X}_i=*\right)$\\
if $\sum\limits_{i\in S_{test}}Z_i<(1-{\kappa })|S_{test}|$,
\\
\hspace{1cm} $O_A=\perp$
\\ \hspace{1cm}{ $F_0=F_1=T_0=T_1=\perp$} \\ \hspace{1cm}{ $X=W_0=W_1=\perp$}
\\ else  \\\hspace{1cm} $O_A=\emptystring$

\\ \hspace{1cm}$T_0,T_1\leftarrow_U\lbrace0,1\rbrace^{O(\log^6(\lambda))}$ &&Processing\\\hspace{1cm}
 $R_{0} ={A}_i\left(0\right)\Big\rvert_{i \in S_{prot}}$ & & \\\hspace{1cm}
    $W_{0} = H\left(R_{0}\right)$ & & \\
 \hspace{1cm} $R_{1} ={A}_i\left(1\right)\Big\rvert_{i \in S_{prot}}$ & & \\\hspace{1cm}
    $W_{1} = H\left(R_{1}\right)$ & & \\\hspace{1cm}
         $F_0 = S_0 \oplus \ext\left(R_0,T_0\right)$ & & \\\hspace{1cm}
          $F_1 = S_1 \oplus \ext\left(R_{1},T_1\right)$ & & \\
     & $\xrightarrow{{F_0,F_1,T_0, T_1, \wht{X}}}$ & \\
   & $\xrightarrow{{W_0,W_{1}}}$ \\
   
 \\ $\widehat{A}$ & $\eta_\mathcal{Q}$ & $F_0F_1T_0T_1\wht{X}W\widehat{B}\overline{B}$
 \\ \hdashline \\
       Output $O_{A}$  && Generate output $O_B$ \\\\ $O_A\widehat{A}$ & $\tau_\mathcal{Q}$ & $O_B\widehat{B}$\\ \hdashline \\
      \end{tabular} }
        \caption{Cheating $\bob^*$ protocol }
        \label{algo:ot_cheating_bob} 
    \end{table}
\subsection*{Simulator for  
cheating $\bob^*$}

The simulator for cheating $\bob^*$,  \textsc{Sim}($\bob^*$) (denoted as \textsc{Sim}), takes as input the description of the circuit of $\bob^*$. 

 \begin{table} [H]   \resizebox{15.5cm}{7.2cm}{   \centering
            \begin{tabular}{  l c r }
            \hline \\
            \textbf{$\alice$ (with input $(S_0,S_1)$)} &  & \textbf{\textsc{Sim} (with input $D$)} \\
            \hline \\
            \\ $S_0S_1\overline{A} $& $\kappa_{\Sim}$ & $D\overline{B}$\\
 \hdashline \\
 & &  Perform test phase checks locally \\
            & &    If test fails, $\widehat{D}=\perp$ \\
            \text{
            } & &  Perform $\bob^*$'s steps from Table $\ref{algo:ot_cheating_bob}$  
            \\ && before $\alice$ sends $W_0W_1\wht{X}T_0T_1F_0F_1$\\&& This generates $\widehat{B}\overline{B}$
            \\&& Generate $W\wht{X}T_0T_1$ from the conditional\\
 &&  distribution $(W\wht{X}T_0T_1|\widehat{B}\overline{B})^{\eta_\mathcal{Q}}$
\\ && If $\mathcal{L}_0(\wht{X}\widehat{B}\overline{B})$ holds, $\widehat{D}=1$
\\&&
Otherwise $\widehat{D}=0$ 
            
    \\

    & 
\centering
\resizebox{0.3\textwidth}{!}{%
\begin{circuitikz}
\tikzstyle{every node}=[font=\normalsize]
\draw (5.25,14.75) rectangle (8.5,11.75);
\node [font=\large] at (7,13.25) {IDEAL($OT$)};
\draw [->, >=Stealth, dashed] (3.5,14.25) -- (5.25,14.25);
\draw [->, >=Stealth, dashed] (9.75,14.25) -- (8.5,14.25);
\node [font=\normalsize] at (4.5,14.5) {$(S_0,S_1)$};
\node [font=\normalsize] at (9,14.75) {$\widehat{D}$};
\draw [->, >=Stealth, dashed] (5.25,13.5) -- (3.5,13.5);
\node [font=\normalsize] at (4.5,13.75) {$O^\ideal_A$};
\draw [->, >=Stealth, dashed] (3.5,12.75) -- (5.25,12.75);
\node [font=\normalsize] at (4.1,13.05) {$A_\perp=O^\ideal_A$};
\draw [->, >=Stealth, dashed] (8.5,12.25) -- (10,12.25);
\node [font=\normalsize] at (9.7,12.7) {$O_B^\ideal = S_{\widehat{D}}$};
\end{circuitikz}
}%

 & \\

    && Generate $(F_{\wht{D}}|S_{\wht{D}}\widehat{B}\overline{B}W\wht{X}T_{0}T_{1})$ \\
    &&Set $F_{1-\wht{D}}=U_1$\\
 
  \\ $\widehat{A}$ & $\eta_\textsc{Sim}$ & $F_0F_1T_0T_1\wht{X}W\widehat{B}\overline{B}$
 \\ \hdashline

  \\&& Perform $\bob^*$'s steps from Table~\ref{algo:ot_cheating_bob}
  \\&& after communication from $\alice$\\ Output $O^\Sim_A=O^\ideal_A$&& Output $O^\Sim_B=O_B$ \\  \\
  $O^\Sim_A\widehat{A}$ & $\tau_{\textsc{Sim}}$ & $O^\Sim_B\widehat{B}$\\
 \hdashline \\
      \end{tabular} }
        \caption{Simulator for $\bob^*$ : $\mathcal{Q} \left(\alice,\textsc{Sim}\right)$.}
        \label{algo:bob_sim} 
    \end{table}
        
        \begin{lemma}\label{lem:sim_bob_security}
 Protocol $\mathcal{Q}$  is secure  according to Definition~\ref{def:simulator_cheating_alice} . In other words,
    \[ \onenorm {{(\widehat{B} O^\Sim_AO^\Sim_BS_0S_1)}^{\tau_\textsc{Sim}}} { {(\widehat{B}O_AO_{B}S_0S_1)}^{\tau_\mathcal{Q}}} \leq {\negl(\lambda)}.\]
\end{lemma}
\begin{proof}

Since $\alice$ is honest and $\Sim$ has knowledge of the circuit used to create the devices, the test phase checks can be performed by the simulator without any communication from $\alice$, and hence, $\Sim$ can generate $O_A$. If $O_A=\perp,~\Sim$ inputs $\widehat{D}=\perp$ which gives $O_A^\ideal=\perp$. Otherwise if $O_A=\emptystring$, $\Sim$ inputs $\wht{D}=1$ if $\mcL_0$ holds and $\wht{D}=0$ if $\mcL_1$ holds. In this case, $O_A^\ideal=\emptystring$. 
 Thus, we get (in the state $\tau_\Sim$), $\Pr[O_A^\ideal=O_A]=1$. Since $\textsc{Sim}$ sets $O^\Sim_A=O_A^\ideal$, we get

\begin{equation}{\label{eq:sim_bob_main_eq}
  \onenorm {{(\widehat{B} O^\Sim_AO^\Sim_BS_0S_1)}^{\tau_\textsc{Sim}}} { {(\widehat{B}O_AO_{B}S_0S_1)}^{\tau_\mathcal{Q}}}=  \onenorm {{(\widehat{B} O^\Sim_BS_0S_1)}^{\tau_\textsc{Sim}}} { {(\widehat{B}O_{B}S_0S_1)}^{\tau_\mathcal{Q}}}}.
 \end{equation}
Since the simulator follows exactly the same steps as $\bob^*$ from Table \ref{algo:ot_cheating_bob}, if the simulator can generate $\alice$'s message $W\wht{X}T_0T_1F_0F_1$ from Table \ref{algo:ot_cheating_bob}, it can generate the final state $\tau_\mathcal{Q}$.
 In case the test phase does not pass ($O_A=\perp$), $\Sim$ can easily generate this message as all of these are set to $\perp$ in Table \ref{algo:ot_cheating_bob}.

Suppose that the test phase passes. Let the joint state of $\alice$ and $\bob^*$  after $\alice$ communicates  $F_0,F_1,T_0,T_1,\wht{X},W$ to $\bob^*$ in Protocol $\mathcal{Q}$ be  $\eta_\mathcal{Q}$. The corresponding joint state between $\alice$ and \textsc{Sim} is denoted as $\eta_{\textsc{Sim}}$.

 From the simulator construction, we see $(S_0S_1\widehat{B}\overline{B})^{\eta_\mathcal{Q}}=(S_0S_1\widehat{B}\overline{B})^{\eta_\Sim}$ since $\textsc{Sim}$ generates these exactly as in Protocol $\mathcal{Q}$.
The simulator now generates  $W\wht{X}T_0T_1$  from the conditional distribution 
$(W\wht{X}T_0T_1|\widehat{B}\overline{B})^{\eta_\mathcal{Q}}$. This is possible as $\Sim$ knows the description of the circuit used to create the devices and the following Markov Chain in the state $\eta_Q$ (since $S_0,S_1$ are only used for computing $F_0,F_1$):
\begin{equation}
    \label{eq:sim-markov1}
S_0S_1-\widehat{B}\overline{B}-W\wht{X}T_0T_1.
\end{equation}
We suppose, without loss of generality, that $\mathcal{L}_0(\wht{X}\widehat{B}\overline{B})$ holds, which gives $\wht{D}=1$. $\Sim$ then generates $(F_{1}|S_{1}\widehat{B}\overline{B}W\wht{X}T_{0}T_{1})$ and sets $F_0=U_1$.
\begin{align*}
  & ( S_{0}S_{1}\widehat{B}\overline{B}W\wht{X}T_{0}T_{1}F_{0}F_{1} )^{\eta_\Sim} \\&= \left(S_{0}S_{1}\widehat{B}\overline{B}~(W\wht{X}T_{0}T_{1}|\widehat{B}\overline{B})~(F_{1}|S_{1}\widehat{B}\overline{B}W\wht{X}T_{0}T_{1})\otimes U_1\right)^{\eta_\mathcal{Q}} \\&= \left(S_{0}S_{1}\widehat{B}\overline{B}W\wht{X}T_{0}T_{1}(F_{1}|S_{1}\widehat{B}\overline{B}W\wht{X}T_{0}T_{1})\otimes U_1\right)^{\eta_\mathcal{Q}} &\mbox{(eq. \eqref{eq:sim-markov1})}\\&\approx_{{\negl(\lambda)}} ( S_{0}S_{1}\widehat{B}\overline{B}W\wht{X}F_{1}T_{1}T_{0}F_{0} )^{\eta_\mathcal{Q}}. &\mbox{(from Claim \ref{claim:OT_markov_actual_protocol})}
\end{align*}
Since \textsc{Sim} and $\bob^*$  perform exactly the same operations on $\widehat{B}\overline{B}W\wht{X}F_{1}T_{1}T_{0}F_{0}$  (from Table \ref{algo:ot_cheating_bob}) after communication from $\alice$ to generate $\widehat{B}O_B$, Fact \ref{fact:data} gives us the following:
\begin{equation}\label{eq:sim_bob_closeness}
   \onenorm {{(\widehat{B} O_BS_0S_1)}^{{\tau_\Sim}}} { {(\widehat{B}O_{B}S_0S_1)}^{\tau_\mathcal{Q}}} \leq {\negl(\lambda)}. \end{equation} 
   Since $\Sim$ sets $O^\Sim_B=O_B$, using eqs.~\eqref{eq:sim_bob_main_eq} and \eqref{eq:sim_bob_closeness}, we get
    \[ \onenorm {{(\widehat{B} O^\Sim_AO^\Sim_BS_0S_1)}^{{\tau_\Sim}}} { {(\widehat{B}O_AO_{B}S_0S_1)}^{\tau_\mathcal{Q}}} \leq {\negl(\lambda)}.\]
   \end{proof}

We also show how to go from non-augmented security to augmented security for Protocol~$\mathcal{Q}$.
We show this only for the cheating $\alice^*$ simulator (Table \ref{algo:ot_cheating_alice}).
The proof for cheating $\bob^*$ follows similarly.
\begin{claim}\label{claim:non-aug-to-aug}
   Protocol $\mathcal{Q}$ is augmented secure and hence, composable.
\end{claim}
\begin{proof}We defer the proof of this claim to the Appendix.
\end{proof}

 \begin{claim}\label{claim:OT_markov_actual_protocol}  In state $\eta_\mathcal{Q}$ (Table \ref{algo:ot_cheating_bob}), conditioned on $\test=1$ and $\mathcal{L}_0(\wht{X}\widehat{B}\overline{B})$
 \begin{equation*}
    S_{0}S_{1}\widehat{B}\overline{B}W\wht{X}T_{0}T_{1}F_{0}F_{1} \approx_{{\negl(\lambda)}}  S_{0}S_{1}\widehat{B}\overline{B}W\wht{X}T_{0}T_{1}(F_{1} |S_{1}\widehat{B}\overline{B} {W\wht{X}T_0T_1}) \otimes U_{1}.
\end{equation*}

Similarly, conditioned on $\test=1$ and $\mathcal{L}_1(\wht{X}\widehat{B}\overline{B})$,
 \begin{equation*}
    S_{0}S_{1}\widehat{B}\overline{B}W\wht{X}T_{0}T_{1}F_{0}F_{1} \approx_{{\negl(\lambda)}}  S_{0}S_{1}\widehat{B}\overline{B}W\wht{X}T_{0}T_{1}(F_{0}|S_{0}\widehat{B}\overline{B} {W\wht{X}T_0T_1}) \otimes U_{1}.
\end{equation*}
\end{claim}
\begin{proof}
We consider the case when $\test=1$ and $\mathcal{L}_0(\wht{X}\widehat{B}\overline{B})$. The other case follows similarly.  From Claim~\ref{claim:OT_security},  in state $\eta_\mathcal{Q}$ we have,
\begin{equation} \label{eq:closeness_with_uniform}
S_{0}S_{1}\widehat{B}\overline{B}W\wht{X}T_{0}T_{1}F_{0}F_{1} \approx_{{\negl(\lambda)}}  S_{0}S_{1}\widehat{B}\overline{B}W\wht{X}T_{0}T_{1}F_{1} \otimes U_{1}. \end{equation}  
Also, since $F_1$ can be generated from $S_{1}\widehat{B}\overline{B}W\wht{X}T_{1}T_{0}$ without involving $S_0$, in state $\eta_\mathcal{Q}$ we have,
{\begin{equation} \label{eq:markov_closeness_with_uniform}S_{0}S_{1}\widehat{B}\overline{B}W\wht{X}T_{0}T_{1}F_{1} = S_{0}-S_{1}\widehat{B}\overline{B}W\wht{X}T_{1}T_{0}-F_1.\end{equation}}
Combining eqs. \eqref{eq:closeness_with_uniform} and \eqref{eq:markov_closeness_with_uniform}, we get,  
\begin{align*}
  S_{0}S_{1}\widehat{B}\overline{B}W\wht{X}T_{0}T_{1}F_{0}F_{1} &\approx_{{\negl(\lambda)}} S_{0}S_{1}\widehat{B}\overline{B}W\wht{X}T_{0}T_{1}(F_1|S_{1}\widehat{B}\overline{B}W\wht{X}T_{0}T_{1}) \otimes U_1.
  \end{align*}
  
\end{proof}
Let us define the set $\textsc{good}_{\wht{x}\widehat{b}\overline{b}}(0)$  as follows (the case $\textsc{good}_{\wht{x}\widehat{b}\overline{b}}(1)$ is similar):
 \begin{equation} \label{eq:def_of_good_0}
    \textsc{good}_{\wht{x}\widehat{b}\overline{b}}(0) = \Big\lbrace wt_1f_1\ : \hmin(A(0)_{S_{prot}}|\wht{x}\widehat{b}\overline{b}wt_1f_1)_{\eta_{\mathcal{Q}}}\geq {\Omega(\log^{10}(\lambda))}\Big\rbrace .
\end{equation} 
\begin{equation*}
    \textsc{good}_{}(0) = \Big\lbrace \wht{x}\widehat{b}\overline{b}wt_1f_1\ : wt_1f_1 \in \textsc{good}_{\wht{x}\widehat{b}\overline{b}}(0)\Big\rbrace .
\end{equation*}
Before going to security, we will prove some useful claims.
We prove relevant claims only for the case when $\mathcal{L}_0(\wht{X}\widehat{B}\overline{B})$ holds.  
The case when $\mathcal{L}_1(\wht{X}\widehat{B}\overline{B})$ holds follows analogously.
 
In Claim \ref{claim:OT_security} and Claim \ref{claim:wtc_good}, we use the following shorthand for ease of notation (where $\eta_\mathcal{Q}$ is from Table \ref{algo:ot_cheating_bob}): \[{\Pr_{\wht{x}\widehat{b}\overline{b}wt_1f_1\leftarrow \eta_\mcQ^{\hat{X}\widehat{B}\overline{B}WT_1F_1}} \equiv\Pr_{\wht{x}\widehat{b}\overline{b}wt_1f_1}}\]

\begin{claim} \label{claim:OT_security}  Conditioned on $\test=1$ and $\mathcal{L}_0(\wht{X}\widehat{B}\overline{B})$, $\forall s_0,s_1\in \{0,1\}$ (in state $\eta_\mathcal{Q}$),
\[
\onenorm{F_0T_0F_1T_1W\wht{X}\widehat{B}\overline{B}}{U_{1}\otimes T_0F_1T_1W\wht{X}\widehat{B}\overline{B}} \leq {\negl(\lambda)}.
\] 

\end{claim}

\begin{proof} 
From our \textbf{DELAY} assumption, after the \textbf{DELAY}, the entire state in the protocol is classical. We represent the state $\eta_\mathcal{Q}$ (see Table \ref{algo:ot_cheating_bob}) conditioned on $\test=1$ and $\mathcal{L}_0(\wht{X}\widehat{B}\overline{B})$ as a probability distribution $P_0$.
 \small{
\begin{flalign*}
     & \onenorm{F_0T_0F_1T_1W\wht{X}\widehat{B}\overline{B}_{\eta_\mathcal{Q}}}{U_{1}\otimes T_0F_1T_1W\wht{X}\widehat{B}\overline{B}_{\eta_\mathcal{Q}}} 
\\=&\onenorm{P_0^{F_0T_0F_1T_1W\wht{X}\widehat{B}\overline{B}}}{U_{1+\polylog(n)}\otimes P_0^{F_1T_1W\wht{X}\widehat{B}\overline{B}}}&\mbox{{(since $\alice$ is honest)}}
    \\ \leq  &  
    \sum_{\substack{f_0t_0f_1t_1w\wht{x}\widehat{b}\overline{b},\\ \wht{x}\widehat{b}\overline{b}wt_1f_1 \in \textsc{good}_{}(0)  }} \left\vert P_0({f_0t_0f_1t_1w\wht{x}\widehat{b}\overline{b}}) - U(f_0t_0)P_0({f_1t_1w\wht{x}\widehat{b}\overline{b}}) \right\vert \\& \tab \tab \tab + \Pr_{\wht{x}\widehat{b}\overline{b}wt_1f_1}(\wht{x}\widehat{b}\overline{b}wt_1f_1\not\in\textsc{good}_{}(0))\\ \leq  &  
    \sum_{\substack{f_0t_0f_1t_1w\wht{x}\widehat{b}\overline{b},\\ \wht{x}\widehat{b}\overline{b}wt_1f_1 \in \textsc{good}_{}(0)  }} \left\vert P_0({f_0t_0f_1t_1w\wht{x}\widehat{b}\overline{b}}) - U(f_0t_0)P_0({f_1t_1w\wht{x}\widehat{b}\overline{b}}) \right\vert  +{\negl(\lambda)} &\mbox{(from Claim~\ref{claim:wtc_good})}\\  
{\leq} & \sum_{\substack{f_1t_1w\wht{x}\widehat{b}\overline{b},\\ \wht{x}\widehat{b}\overline{b}wt_1f_1 \in \textsc{good}_{}(0)}} P_0 (f_1t_1w\wht{x}\widehat{b}\overline{b}) \sum_{f_0t_0}\left\vert P_0 \vert_{f_1t_1w\wht{x}\widehat{b}\overline{b}} ({f_0t_0}) - U({f_0t_0}) \right\vert + {\negl(\lambda)}
\\
= &\sum_{\substack{f_1t_1w\wht{x}\widehat{b}\overline{b},\\ \wht{x}\widehat{b}\overline{b}wt_1f_1 \in \textsc{good}_{}(0)}} P_0 (f_1t_1w\wht{x}\widehat{b}\overline{b})  \left\Vert P_0^{F_0T_0} \vert_{\wht{x}\widehat{b}\overline{b}wt_1f_1} - U^{F_0T_0} \right\Vert_1  + {\negl(\lambda)}\\ 
     \leq &\sum_{f_1{t_1w\wht{x}\widehat{b}\overline{b}}} P_0 (f_1t_1w\wht{x}\widehat{b}\overline{b})\cdot{\negl(\lambda)}  +{\negl(\lambda)}&\mbox{(from Claim~\ref{claim:ot_extractor})}\\ 
      \leq & ~ {\negl(\lambda)}. %
\end{flalign*} }
\end{proof}

\begin{claim} \label{claim:wtc_good}
Conditioned on $\test=1$ and $\mathcal{L}_0(\wht{X}\widehat{B}\overline{B})$ in the state in state $\eta_\mathcal{Q}^{\wht{A}F_0F_1T_0T_1\wht{X}W\widehat{B}\overline{B}}$,  $$\Pr_{\wht{x}\widehat{b}\overline{b}wt_1f_1} \left( \wht{x}\widehat{b}\overline{b}wt_1f_1 \notin \textsc{good}_{}(0)\right) \leq {\negl(\lambda)}.$$
\end{claim}
\begin{proof}From Claim \ref{claim:ot_entropy}, with probability $1-\negl(\lambda)$ (over choice of $\wht{x}\widehat{b}\overline{b}$), conditioned on $\test=1$ and $\mathcal{L}_0(\wht{x}\widehat{b}\overline{b})$, we have
{\small\begin{flalign*}
    &\Pr_{wt_1f_1}\left( \hmin(A(0)_{S_{prot}}|\wht{x}\widehat{b}\overline{b}wt_1f_1)_{\eta_{\mathcal{Q}}}\geq {\Omega(\log^{10}(\lambda))}\right) \\
    &\geq \Pr_{wt_1f_1}\left( \hmin(A(0)_{S_{prot}}|\wht{x}\widehat{b}\overline{b}wt_1f_1)_{\eta_{\mathcal{Q}}}\geq\hmin(A(0)_{S_{prot}}|x\widehat{b}\overline{b}T_1WF_1)_{\eta_{\mathcal{Q}}}- {\log^2(\lambda)}\right)& (\mbox{from Claim~\ref{claim:ot_entropy}}) \\
    &\geq 1-2^{-{\log^2(\lambda)}}. & (\mbox{from Fact~\ref{fact:average_to_worst_min_entropy}})
\end{flalign*}}
So we get,
    \begin{align*}
        \Pr_{\wht{x}\widehat{b}\overline{b}wt_1f_1} \left( \wht{x}\widehat{b}\overline{b}wt_1f_1 \notin \textsc{good}_{}(0)\right) 
        & \leq   2^{-{\log^2(\lambda)}} +\negl(\lambda)
        \leq {\negl(\lambda)}. & 
    \end{align*}
\end{proof}

\begin{claim}
    
 \label{claim:ot_extractor}
   Conditioned on $\test=1$, $ \wht{x}\widehat{b}\overline{b}wt_1f_1 \in \textsc{good}_{}(0)$ and $\mathcal{L}_0(\wht{x}\widehat{b}\overline{b})$,
 $\forall   s_0,s_1$,
    \[ \onenorm{P_0^{F_0T_0} \vert_{\wht{x}\widehat{b}\overline{b}wt_1f_1} }{U^{F_0T_0}  }  \leq {\negl(\lambda)}.\]
\end{claim}
\begin{proof}
    Recall that by eq.~\eqref{eq:def_of_good_0}, we have 
\[\hmin(A(0)_{S_{prot}}|\wht{x}\widehat{b}\overline{b}wt_1f_1)_{\eta_{\mathcal{Q}}}\geq \Omega({\log^{10}(\lambda)}) .\]
 Fact \ref{fact:strong_extractor} and our setting of $\eps= 2^{-\log^3(\lambda)}$ and $m=O(\log^{10}(\lambda))$ imply that our source has enough min-entropy for the extractor to work. Since $T_0$ and $T_1$ are independently and uniformly generated by $\alice$, we get
    \begin{equation}
        \label{eq:entropy_final_old}
   \onenorm{\ext\left( R_0,T_0\right)\circ T_0 \vert_{\wht{x}\widehat{b}\overline{b}wt_1f_1}}{U_{1+{O(\log^6(\lambda))}}} =\negl(\lambda). \end{equation}
This gives
\begin{align*}
    \onenorm{P_0^{F_0T_0} \vert_{\wht{x}\widehat{b}\overline{b}wt_1f_1} }{U^{F_0T_0}  } 
&= \onenorm{(s_0 \oplus \ext\left( R_0,T_0\right))\circ T_0 \Big\vert_{\wht{x}\widehat{b}\overline{b}wt_1f_1} }{U_{1+{O(\log^6(\lambda))}}} \\
&= \onenorm{( \ext\left( R_0,T_0\right))\circ T_0 \Big\vert_{\wht{x}\widehat{b}\overline{b}wt_1f_1} }{U_{1+{O(\log^6(\lambda))}}} \\
       & = {\negl(\lambda)}.
\end{align*}
\end{proof}
\begin{claim} \label{claim:ot_entropy}
   With probability $1-\negl(\lambda)$ (over choice of $\wht{x}\widehat{b}\overline{b}$), conditioned on $\test=1$ and  $\mathcal{L}_0(\wht{x}\widehat{b}\overline{b})$, we have   
  $$\hmin(A(0)_{S_{prot}}|\wht{x}\widehat{b}\overline{b}T_1WF_1)_{\eta_{\mathcal{Q}}}\geq{\Omega(\log^{10}(\lambda))}.$$

\end{claim}

\begin{proof}
Recall that $H$ is a parity check matrix for a $\left(  |S_{prot}|, k, d\right)$-code.
Thus, from eq. \eqref{eq:ecc_rate}, we have \[|W|=2(|S_{prot}|-k) = 2|S_{prot}|(1- r ) \leq  \delta'n/40  .  \]
With probability $1-\negl(\lambda)$ (over choice of $\wht{x}\widehat{b}\overline{b}$), conditioned on $\test=1$ and  $\mathcal{L}_0(\wht{x}\widehat{b}\overline{b})$, we have
 \begin{align*}
     &\hmin(A(0)_{S_{prot}}|\wht{x}\widehat{b}\overline{b}T_1WF_1)_{\eta_{\mathcal{Q}}}\\
     & \geq \hmin(A(0)_{S_{prot}}|\wht{x} \widehat{b}\overline{b})_{\eta_{\mathcal{Q}}}-|W|-|T_1|-|F_1|
    \\ &\geq  \delta'n/20-\delta'n/40 -  O({\log^6(\lambda)})  - 1  &\mbox{(from Claim \ref{claim:case_analysis})}\\
    &\geq {\Omega(\log^{10}(\lambda))} .& 
 \end{align*}
 \end{proof}

\begin{claim} \label{claim:case_analysis} Let $\kappa  \leq \frac{1}{4}\eps_*-\delta_0$ where $\kappa$ is from Table \ref{table:OT}.
Let the registers with $\bob$ (after $\alice$'s communication) be $\eta_{\mcQ}^{F_0F_1T_0T_1W_0W_1\wht{X}\widehat{B}\overline{B}}$ (see Table \ref{algo:ot_cheating_bob}). Let $\mcL_0(\wht{x}\widehat{b}\overline{b}),\mcL_1(\wht{x}\widehat{b}\overline{b})$ be complementary events from eqs. \eqref{eq:def_L0} and \eqref{eq:def_L1}. Then one of the following holds:
\begin{itemize}
    \item probability of test phase passing is $\negl(\lambda)$ (in this case, $\alice$'s security holds trivially).
    \item conditioned on the test phase passing, with probability $1-\negl(\lambda)$ (over the choice of $\wht{x}\widehat{b}\overline{b}$)
    \begin{enumerate}
        \item if $\mcL_0(\wht{x}\widehat{b}\overline{b})$ happens,
    \[
\hmin(A(0)_{S_{prot}}|\wht{x}\wht{b}\bar{b})_{\eta_{\mathcal{Q}}}\geq\delta'n/20.
\] 
\item if $\mcL_1(\wht{x}\widehat{b}\overline{b})$ happens,
\[
\hmin(A(1)_{S_{prot}}|\wht{x}\wht{b}\bar{b})_{\eta_{\mathcal{Q}}}\geq\delta'n/20.       
\]
  \end{enumerate}
\end{itemize}
\end{claim}

\begin{proof}
Let $W_i$ be as defined in Claim \ref{claim:threshold}. Consider \[
    Lose \defeq\{i\in [n]: W_i=0\}.\]
  If $\test=1$, we have \[|S_{test}\cap Lose| \leq \kappa |S_{test}|.\] 
    
    \textbf{Case 1:}\[
    \Pr(\test=1) \leq 2^{-\delta' n/2}=\negl(\lambda).
    \]
     In this case, $\alice$ aborts with $1-\negl(\lambda)$ probability and sets $F_0=F_1=T_0=T_1=\wht{X}=W_0=W_1=\perp$. In this case, security follows trivially as all communication from $\alice$ to $\bob$ is independent of $\alice$'s input. 
  
    \textbf{Case 2:}\[
    \Pr(\test=1) > 2^{-\delta' n/2}.
    \] Let the probability of $\bob$ correctly guessing $A(0)_{S_{prot}}$ and $A(1)_{S_{prot}}$ in the state $\eta_{\mcQ}^{\wht{A}\wht{X}\widehat{B}\overline{B}}$ be as in eq. \eqref{eq:guessing_prof_def}. We have
    \begin{align}
& \Pr(\text{$\bob$ correctly guesses $A(0)_{S_{prot}}$ and $A(1)_{S_{prot}}$}~~|~~\test=1)_{\eta_{\mcQ}^{\wht{A}\wht{X}\widehat{B}\overline{B}}} \nonumber\\&\leq \Pr(\text{$\bob$ correctly guesses $A[0,1]$ for at least } (1-\kappa )|S_{prot}| \text{ devices}~~|~~\test=1)_{\eta_{\mcQ}^{\wht{A}\wht{X}\widehat{B}\overline{B}}}
        \nonumber\\&=\frac{\Pr(\text{$\bob$ correctly guesses $A[0,1]$ for at least } (1-\kappa )|S_{prot}| \text{ devices } \cap \test=1)_{\eta_{\mcQ}^{\wht{A}\wht{X}\widehat{B}\overline{B}}}}{\Pr(\test=1)}\nonumber
         \\&{\leq\frac{\Pr(\sum_{i\in [n]}W_i \geq (1-\frac{1}{4}\eps_*+\delta_0)n)_{G_n}}{\Pr(\test=1)}}\nonumber
         \\&\leq \frac{2^{-\delta' n}}{2^{-\delta' n/2}}  \hspace{4.6in} \mbox{(Claim \ref{claim:threshold})}\nonumber
         \\&\leq 2^{-\delta' n/2}. \label{eq:take_log}
    \end{align}
    For the 3rd last inequality to hold, we need \begin{align*}
        &  (1-\kappa )|S_{prot}|+|S_*|+(1-\kappa )|S_{test}| \geq (1-\frac{1}{4}\eps_*+\delta_0)n.
        \end{align*}
       This holds when \begin{align*}
           \kappa  \leq \frac{1}{4}\eps_*-\delta_0. 
    \end{align*}
    Using Fact \ref{fact:p_guess}, we get
    \[
\E_{\wht{x}\widehat{b}\overline{b}\leftarrow \wht{X}\widehat{B}\overline{B}|\test=1}\left[p_{guess}(A(0)_{S_{prot}}, A(1)_{S_{prot}}|\wht{X}\widehat{B}\overline{B}=\wht{x}\widehat{b}\overline{b},\test=1)\right] \leq 2^{-\delta' n/2}.  
    \]
 Applying Markov's inequality, we get that with probability atleast $1-2^{-\delta'n/4}=1-\negl(\lambda)$ over the choice of  $\wht{X}\widehat{B}\overline{B}=\wht{x}\widehat{b}\overline{b}$ and conditioned on $\test=1$,  the following holds: \begin{equation}
\label{eq:joint_prob_small}
    p_{guess}(A(0)_{S_{prot}}, A(1)_{S_{prot}}|\wht{X}\widehat{B}\overline{B}=\wht{x}\widehat{b}\overline{b},\test=1) \leq 2^{-\delta' n/4}.  
     \end{equation}
   Conditioning on a such a "good" $\wht{X}\widehat{B}\overline{B}=\wht{x}\widehat{b}\overline{b}$ (that satisfies the above equation) and $\test=1$, we have \begin{equation*}
    p_{guess}(A(0)_{S_{prot}}, A(1)_{S_{prot}}) \leq 2^{-\delta' n/4}. \end{equation*}  
    Define sets \[\mcP_0(\wht{x}\widehat{b}\overline{b})=\{a(0)_{S_{prot}}: \Pr(a(0)_{S_{prot}}|\wht{x}\widehat{b}\overline{b})\leq 2^{-\delta' n/10} \},\]
    \[\mcP_1(\wht{x}\widehat{b}\overline{b})=\{a(0)_{S_{prot}}: \Pr(a(0)_{S_{prot}}|\wht{x}\widehat{b}\overline{b})> 2^{-\delta' n/10} \}.\] 
    For all $a(0)_{S_{prot}}\in \mcP_1(\wht{x}\widehat{b}\overline{b}) $ and $a(1)_{S_{prot}}$, we have  \begin{align*}
       & \Pr(a(1)_{S_{prot}}| a(0)_{S_{prot}})
       \\&= \frac{\Pr(a(1)_{S_{prot}}a(0)_{S_{prot}})}{\Pr(  a(0)_{S_{prot}})}
       \\&\leq \frac{2^{-\delta' n/4}}{2^{-\delta'n/10}} \leq 2^{-\delta' n/10}. 
    \end{align*}
    Therefore, for all $a(1)_{S_{prot}}$, we have 
    \begin{equation}
\label{eq:prob_conditioned_on_L1}\Pr(a(1)_{S_{prot}}| A(0)_{S_{prot}}\in\mcP_1(\wht{x}\widehat{b}\overline{b}))\leq 2^{-\delta' n/10}.\end{equation} 
Therefore \[p_{guess}(A(1)_{S_{prot}}| A(0)_{S_{prot}}\in\mcP_1(\wht{x}\widehat{b}\overline{b}))\leq 2^{-\delta' n/10}.\]
Define \begin{align}
&\mcL_0(\wht{x}\wht{b}\bar{b})\defeq\{A(0)_{S_{prot}}\in \mcP_0(\wht{x}\widehat{b}\overline{b})\}, \label{eq:def_L0}
\\&\mcL_1(x\wht{b}\bar{b})\defeq\{A(0)_{S_{prot}}\in \mcP_1(\wht{x}\widehat{b}\overline{b})\}.\label{eq:def_L1}
\end{align}
Thus, conditioned on $\test=1$ and $\mcL_1(\wht{x}\wht{b}\bar{b})$, we have \[\hmin(A(1)_{S_{prot}}|\wht{x}\wht{b}\bar{b}) \geq \delta'n/10. \]
{If $\Pr(\mcL_0(\wht{x}\wht{b}\bar{b}))\geq 2^{-\delta'n/100}$, we get (conditioned on $\test=1$ and $\mcL_0(\wht{x}\wht{b}\bar{b})$)\[\hmin(A(0)_{S_{prot}}|\wht{x}\wht{b}\bar{b})\geq \delta'n/10-\delta'n/100 \geq \delta'n/20. \]}
\end{proof}

Consider single-copy game $G_1$ (see Table \ref{table:ot-single-copy}).  We will see that proving parallel repetition on this game helps us in proving security against cheating $\alice^*$. Let us represent the $n$-fold repetition of the game $G_1$ played in parallel by $G_n$. For the following claim, the parameters $\delta_0, t$, and $\delta_*$ are defined in eqs. \eqref{eq:parameter_t} and \eqref{eq:parameter_delta_0} and $\eps_*$ is from Claim \ref{claim:value_G_1}.
\begin{claim}[Threshold theorem]\label{claim:threshold}

  Define $T_i=V(X_i,Y_i,A_i,B_i)$, i.e., $T_i=1$ if the $n$-copy security game $G_n$  is won in the $i$-th coordinate.  
    Define  $W_i=\Id_{T_i=1}$.  Let $\delta_0>0$ be such that $\delta_0 n=t+100\sqrt{\delta_*}n$. Then \[
    \Pr\left(\sum_{i\in [n]}W_i\geq (1-\frac{1}{4}\eps_*+\delta_0)n\right)_{G_n} \leq  2^{-\delta'n},
    \]for some small constant $\delta'>0$.
\end{claim}

\vspace{-8mm}
\section{Proof of threshold theorem}\label{threshold}

Consider the following security game in the case of cheating $\bob$ and honest $\alice$. Here $\charlie$ (along with the $\ver$) and $\dave$ can be thought of as $\alice$ and $\bob$ respectively from Table \ref{table:OT}. Note that $\charlie$ and $\dave$ receive inputs $(S_0,S_1)$ and $D$ respectively.

\begin{table}[H]  \centering
          \resizebox{11cm}{7.3cm}{  \begin{tabular}{  l l c r r  }
            \hline \\
            \textbf{$\charlie$ } & &$\ver$ && \textbf{$\dave$} \\
            \hline \\ 
&&${X} \leftarrow_U \lbrace 0,1,2\rbrace^{n}$ 
            \\ && $\forall i\in[n], C_i \leftarrow \mathrm{Bern}(1-\delta)$ \\
           &&  $Y^{\alice}_i \leftarrow_U \lbrace 0,1,2\rbrace$ & 
            \\ &&$S_{test}=\lbrace i\in [n] : C_i=0\rbrace$ \\ &&$\forall i\in[n], \wtt{C}_i \leftarrow \mathrm{Bern}(1-\delta_*)$  
            \\ &&if $\wtt{C}_i=1, \wht{X}_i=X_i$  \\ &&else $\wht{X}_i=*$\\&& $S_{*}=\lbrace i\in [n]\setminus S_{test} : \wtt{C}_i=0\rbrace$
            \\&&$S_{prot}=[n]\setminus (S_*\cup S_{test})$\\&$\xleftarrow{X}$ &&$\xrightarrow{S_{test}, Y^{\alice}}$\\ 

            \\
            && if $C_i=0$ then $Y_i=Y^{\alice}_i$ \\
             && if $C_i=1$ then $Y_i=D$ \hspace{0.4cm} && Generate $B_{S_{test}}$ \\\\
            & &   {\textbf{DELAY} } 
         \\
              \\&&&$\xleftarrow{B_{S_{test}}}$\\&&&$\xrightarrow{\wht{X}}$\\
 & $\xrightarrow{A}$&&$\xleftarrow{\widetilde{A}[0,1]_{[n]\setminus\S_{test}}}$\\

             & &   \textbf{Winning Condition}   \\
    && \hspace{1.2cm} for $i\in S_{test}$:\tab \tab \quad \\
&&\tab$Z_i=\Id\left(\MS^i(X_i,Y^{}_i,A_i,B_i)=1 \lor \wht{X}_i=*\right)$\\
&&if $\sum\limits_{i\in S_{test}}Z_i<(1-{\kappa })|S_{test}|$,
\\
&&\hspace{1cm} $O^{test}_A=\perp$
\\&& \hspace{1.2cm}for $i\in S_{prot}$:\tab \tab \tab \\
&&\hspace{1cm}$\wtt{Z}_i=\Id(A_i[0,1]=\widetilde{A}_i[0,1])$\\
&&\hspace{2.5cm}if $\sum\limits_{i\in S_{prot}}\wtt{Z}_i \leq (1-\kappa )|S_{prot}|$ \tab \tab \\
 
&& $O^{test}_A=\perp$\\
&& Otherwise $O_A^{test}=\emptystring$ \tab
\\\\&& \textbf{Predicate}
\\&&  $V_{P_n}(X,Y^\alice,S_{test}, \wht{X},A,B, \wtt{A})=1$ iff $O_{A}^{test}=\emptystring$
\end{tabular} }
        \caption{Security game OT $P_n$}
        \label{security_game_ot} 
    \end{table}
    In Table \ref{table:ot-single-copy}, we define a single-copy security game $G_1$ for OT against cheating $\bob$.
    Note that for simplicity of notation, we will denote the inputs of the single-copy game $G_1$ with $X$ and $\wht{Y}$ and similarly for outputs (note that the inputs of the security game $P_n$ in Table \ref{table:OT} are also denoted by $X$ and $\wht{Y}$). Note that the size of $X$ and $\wht{Y}$ will automatically become clear from which game is being played. 
    
\begin{table}[H]      \centering
          \resizebox{12cm}{6cm}{  \begin{tabular}{  l l c r r  }
            \hline \\
            \textbf{$\charlie$} & &$\ver$ && \textbf{$\dave$} \\
            \hline \\ 
&&${X} \leftarrow_U \lbrace 0,1,2\rbrace$ 
            \\ && $C \leftarrow \mathrm{Bern}(1-\delta)$ \\
           &&  $\wht{Y} \leftarrow_U \lbrace 0,1,2\rbrace$ & 
            \\ &&  $\wtt{C} \leftarrow \mathrm{Bern}(1-\delta_*)$  \\ &$\xleftarrow{X}$ &&$\xrightarrow{C, \wht{Y}}$\\ 
            & &   \textbf{DELAY}   \\
              &&   if $\wtt{C}=1, \widehat{X}=X$
              \\ &&   \hspace{0mm}else if $\wtt{C}=0, \widehat{X}=*$ 
              \\
&&&$\xleftarrow{\text{if } C=0, \text{ send }B }$ \\
&&&$\xrightarrow{\wht{X} }$\\
 & $\xrightarrow{A}$&&$\xleftarrow{\text{if } C=1, \text{ send }\widetilde{A}[0,1]}$\\
  & &   \textbf{Winning Condition}   \\
    &&  if $C=0$: \tab \tab \tab \tab \\
&&if $\MS(X,\wht{Y},A,B)=0 \land \wht{X}\neq *$,
\\
&&\hspace{1cm} $O^{test}_A=\perp$\\
 &&   else if $C=1\land \wtt{C}=1$:  \tab \tab \quad \\
&&if $A[0,1]\neq\widetilde{A}[0,1]$ \\
&&\hspace{1cm} $O^{test}_A=\perp$\\
&& Otherwise $O_A^{test}=\emptystring$ \tab \tab \tab
\\\\&& \textbf{Predicate}
\\&&  $V(X,\wht{Y},C, \wht{X}, A,B,  \wtt{A})=1$ iff $O_{A}^{test}=\emptystring$
\\       \end{tabular} }
        \caption{1-copy security game $G_1$}
        \label{table:ot-single-copy} 
    \end{table}
\begin{claim} \label{claim:some_bit_entropy_OT} Consider the definition of $p_{guess}(\cdot)$ from eq. \eqref{eq:guessing_prof_def}. Let $\rho$ be any state and define $A^x(i)$ to be the outcome of measuring $\rho$ with $\lbrace \Pi_0^{xi}, \Pi_1^{xi} \rbrace$ (see Table \ref{tab:ms-meas}). There exists $x\in \lbrace 0,1,2\rbrace$, $i \in \lbrace 0,1 \rbrace$ and a small constant $d>0$  such that 
\[
p_{guess}\left({A^x(i)}\right)_\rho \leq 1-\frac{2d}{3}.
\] 
\end{claim}
\begin{proof}
We defer the proof of this claim to the Appendix.
\end{proof}
 Since $\alice$ is honest, she chooses $X\in \{0,1,2\}$ uniformly and independently of ${\kappa}^{\overline{A}}$ (Table \ref{algo:ot_cheating_bob}). In the previous claim, $\rho$ can be set to $\alice$'s state (conditioned on $\bob^*$'s outcome for any arbitrary measurement) to get the following claim, in which we show that the value of the game $G_1$ is bounded away from 1. 
\begin{claim} \label{claim:value_G_1}
   $\vall(G_1)\leq 1-\eps_*$ for some small constant $\eps_*>0$ (see Table \ref{table:ot-single-copy}). 
\end{claim}
\begin{proof}We defer the proof of this claim to the Appendix. 
\end{proof}

Given a single-copy game $G_1$ with $\vall(G_1)$ bounded away from 1, the parallel-repetition theorem says that for an $n$-fold repetition of the game $G_1$ played in parallel (denoted as $G_n$),  $\vall(G_n)$ is exponentially small in $n$. Since the input distribution to our game $G_1$ is not independent, we use anchoring ideas in our proof. Since $\vall(G_1)$ is bounded away from 1,  using parallel repetition, we prove that the winning probability of the $n$-copy anchored game $\vall(G_n)$ is exponentially small, which helps in proving min-entropy guarantees for proving security against cheating $\bob^*$.

 Consider the game $G_1$ (Table \ref{table:ot-single-copy}).  We convert the game with \textbf{DELAY}  to a game with a restricted strategy as follows: Suppose the initial shared quantum state is $\varphi^{MN}\otimes \ket{0}^{N'}$ where $M$ and $N$ are with $\charlie$ and $\dave$ respectively and $N'$ is also with $\dave$ (which will be used to do a $\CNOT$ later).
$\charlie$ and $\dave$ receive their inputs $X$ and $C\wht{Y}\wht{X}$ respectively at the beginning and their answer registers are denoted by $A$ and $B$ respectively. \begin{itemize}
    \item  If $C=0$ (test phase), $\dave$ acts on the shared quantum state $\varphi^N, C$ and $\wht{Y}$ (but does not have access to the register $\wht{X}$) and performs a $\CNOT$ on $\ket{0}^{N'}$ controlled on his answer register $B$ after which access to $N'$ is lost. The $\CNOT$ corresponds to a computational basis measurement that simulates the quantum state's decoherence after the \textbf{DELAY}. After this $\CNOT$, $\dave$ gets access to the register $\wht{X}$ but is forbidden from doing any post-processing on the answer register $B$. This reflects the fact that in the security game $G_1$, $\dave$ already sends the answer before he gets $\wht{X}$.  
    \item If $C=1$ (protocol-phase), there is a mandatory $\CNOT$ gate between $\dave$'s part of the quantum state $\varphi^N$ and $\ket{0}^{N'}$ only after which $\dave$ gets access to $\wht{X}$. Neither $\charlie$ nor $\dave$ have access to $N'$. Note that before the $\CNOT$, $\dave$ can do any arbitrary quantum operation on $\varphi^N, C$ and $\wht{Y}$.  After the $\CNOT$, $\dave$ can use $\wht{X}$ to do (classical) post-processing to generate his answer $B$. \end{itemize}

In the game $G_1$, let the distribution over the inputs conditioned on $\wtt{C}=1$ be $\wtt{\mu}(X\wht{Y}C\wht{X})$. Note that in this distribution, we have the guarantee that $X=\wht{X}$. 
Let the distribution over the inputs conditioned on $\wtt{C}=0$ be $\mu_*$. This distribution has the guarantee that $\wht{X}=*$, where $*$ is a special symbol. {In other words, $${\mu_*}^{X\wht{Y}C\wht{X}}=\wtt{\mu}^{X\wht{Y}C}\otimes \ketbra{*}{*}^{\wht{X}}.$$}
 We define the anchored distribution $\mu^{X\wht{Y}C\wht{X}}$ as follows: \begin{equation}
     \label{eq:anchored_def}
 \mu=\delta_*\mu_*+(1-\delta_*)\wtt{\mu},
 \end{equation}
 where $\delta_*>0$ is a small constant as defined in Claim \ref{claim:main}. Note that for the case when $\wht{X}=*$, the game $G_1$ is always won. 
We also note that $\mu$ (eq. \eqref{eq:anchored_def}) is the distribution of inputs in the protocol in Table \ref{table:ot-single-copy} and is the same as the distribution that is generated for proving security against cheating $\alice^*$ in Table \ref{table:OT}.  
We now apply parallel repetition over the $n$-copy anchored game (denoted as $G_n$) using ideas similar to \cite{JPY,anchoring_2,anchoring_1, anchoring3}.

Before the game begins, let $\charlie$ and $\dave$ share a pure state on registers $AE'_ABE'_B$. Here $AE'_A\cong M$ and $BE'_B\cong NN'$. After getting the inputs, 
$\charlie$ and $\dave$ apply unitaries on their part of the state and measure registers $A$ and $B$ to get the answers. Since $\charlie$ and $\dave$ follow the restricted strategy, the overall unitary corresponding to $\dave$'s restricted strategy (on receiving input $\wht{Y},C,\wht{X}$) consists of the following $3$ unitaries: \begin{itemize}
     \item The first unitary depends on $\widehat{Y},C$ corresponding to the quantum operations of $\dave$ before the \textbf{DELAY}.
     \item After this, a $\CNOT$ acts between all registers of $\dave$ (say $N$) and $\ket{0}^{N'}$ after which $N'$ cannot be accessed by $\dave$. This corresponds to the decoherence due to the \textbf{DELAY}. After this, $\dave$ gets access to the register $\wht{X}$. 
     \item If $C=0$, $\dave$ is forbidden from doing any post-processing on the answer register $B$.
     \item If $C=1$, after the $\CNOT$, $\dave$ applies an isometry (which can depend on $\wht{X}$) on all of his registers along with some ancilla registers. {This isometry maps computational basis states to computational basis states and corresponds to the classical post-processing that $\dave$ can do after the \textbf{DELAY} once he gets access to $\wht{X}$.}
 \end{itemize} We consider the $n$-copy game $G_n$. For simplicity of notation, we will denote inputs for $G_n$ by $X$ and $Y=\wht{Y}C\wht{X}$ (and treat them as a vector of $n$-inputs, i.e., $X=X_1, X_2 \dots X_n$) and similarly for outputs $A$ and $B$. We also define random variables $D_i$ and $U_i$ for all $i \in [n]$. \begin{itemize}
     \item $D_i \leftarrow_{u}\{1,2\}$.
 \item If $D_i=1$, $U_i=X_i$ ($\charlie$'s input).
  \item If $D_i=2$, $U_i=\wht{X}_i$ (part of $\dave$'s input).
  \end{itemize}This defines a joint distribution ${\mu}(X\wht{Y}C\wht{X}DU)$.

Let $\mcC\subseteq [n]$ and $\overline{\mcC} $ represent its complement. Let the state $\ket{\theta}$ be defined as (here {$y$ represents $(\wht{y},c,\wht{x})$, i.e., the entire input to $\dave $}):
\begin{equation}\label{eq:def_theta}
 \ket{\theta} \defeq \sum_{x,y}\sqrt{\mu(x,y)}\ket{xxyy}^{\wtt{X}X\wtt{Y}Y}\otimes\sum_{a_{\mcC}b_{\mcC}du}\ket{a_{\mcC}b_{\mcC}dduu}^{A_{\mcC}B_{\mcC}\wtt{D}D\wtt{U}U}\otimes \ket{\gamma_{xya_{\mcC}b_{\mcC}du}}^{E_AE_B},
 \end{equation}
 where $E_A=E'_AA_{\overline{{\mcC}}}$ and $E_B=E'_BB_{\overline{{\mcC}}}$. Note that $\ket{\gamma_{xya_{\mcC}b_{\mcC}du}}^{E_AE_B}$ is unnormalised and \[\sum_{a_{\mcC}b_{\mcC}du}\ket{a_{\mcC}b_{\mcC}dduu}^{A_{\mcC}B_{\mcC}\wtt{D}D\wtt{U}U}\otimes \ket{\gamma_{xya_{\mcC}b_{\mcC}du}}^{E_AE_B}\] is the shared state after $\charlie$ and $\dave$ do their unitary operations corresponding to the questions $x$ and $y$ ({according to the restricted strategy}). 
     Consider the (purification of) state conditioned on winning in ${\mcC}$ as 
{\small \begin{equation}\label{eq:def_varphi}
 \ket{\varphi} \defeq \frac{1}{\sqrt{q}}\sum_{x,y}\sqrt{\mu(x,y)}\ket{xxyy}^{\wtt{X}X\wtt{Y}Y}\otimes\sum_{a_{\mcC}b_{\mcC}du:V(x_{\mcC},y_{\mcC},a_{\mcC},b_{\mcC})=1}\ket{a_{\mcC}b_{\mcC}dduu}^{A_{\mcC}B_{\mcC}\wtt{D}D\wtt{U}U}\otimes \ket{\gamma_{xya_{\mcC}b_{\mcC}du}}^{E_AE_B},
 \end{equation}}
 where $q$ is the probability of success on ${\mcC}$. Note that copies $\wtt{X}$ and $\wtt{Y}$ (of $X$ and $Y$) are needed so that the marginals of $X$ and $Y$ on tracing out $\wtt{X}$ and $\wtt{Y}$ are classical (and similarly for $D$ and $U$).   
\subsection*{Parameters and notation}\label{subsection:parameter}

\noindent  We have $E_A=E'_AA_{\overline{{\mcC}}}$ and $E_B=E'_BB_{\overline{{\mcC}}}$. Let us define $\wtt{A}\defeq \wtt{X}_{\overline{\mcC}}XE_A$ and $\wtt{B} \defeq \wtt{Y}_{\overline{\mcC}}YE_B$.  
Define \begin{align} 
    R_j\defeq X_{\mcC}Y_{\mcC}A_{\mcC}B_{\mcC}X_{<j}Y_{<j}D_{-j}U_{-j}. \label{eq:define_R_j}\end{align} 
     We let $\ket{\varphi_{r_j}}$ denote the pure state that results when we measure the register $R_j$ in $\ket{\varphi}$ in the computational basis and get result $r_j$ (similarly for $X_{\mcC}, Y_{\mcC}$, etc). We let $\ket{\varphi_{*}}$ denote the pure state that results when we measure theregister $\wht{X}_j$ in $\ket{\varphi}$ in the computational basis and get result $*$.
    
    Let $\wtt{B}=\widehat{B}\wht{X}B_{\overline{\mcC}}N'$ where $\widehat{B}$ contains all registers of $\dave$ after the $\CNOT$ (other than $\wht{X}$ and the answer register $B_{\overline{\mcC}}$).     We have $0<\delta_1,\delta_2, \delta_3, \delta_*< 0.1$ and $k'=\lfloor\delta_1 n\rfloor$. We also have  \begin{equation}    
\delta_3=\delta_2+\delta_1\log(|\clA|\cdot|\clB|), \label{eq:deltas_relation} \end{equation}
  \begin{equation}
      \label{eq:parameter_t}
  \delta_3=\delta_*^2/10^5,~t=\Bigg\lfloor\frac{-\delta_2n}{\log\left(1-\frac{1}{4}(\eps_*-23\sqrt{\delta_*})\right)}\Bigg\rfloor. \end{equation} Let  $\delta_0>0$ be such that \begin{equation}
      \label{eq:parameter_delta_0}
  \delta_0 n=t+100\sqrt{\delta_*}n.\end{equation} 
  We also have $\kappa  \leq \frac{1}{4}\eps_*-\delta_0$.
\begin{lemma}\label{lemma:main}
    \[\E_{x_{\mcC}y_{\mcC}a_{\mcC}b_{\mcC}du\leftarrow \varphi^{X_{\mcC}Y_{\mcC}A_{\mcC}B_{\mcC}DU}}\D(\varphi_{x_{\mcC}y_{\mcC}a_{\mcC}b_{\mcC}du}^{\wtt{X}_{\overline{\mcC}}\wtt{Y}_{\overline{\mcC}}XYE_AE_B}\|\theta_{x_{\mcC}y_{\mcC}du}^{\wtt{X}_{\overline{\mcC}}\wtt{Y}_{\overline{\mcC}}XYE_AE_B}) \leq -\log q + |\mcC|\log(|\clA|\cdot|\clB|). \]
\end{lemma}
\begin{proof}
From eqs. \eqref{eq:def_theta} and \eqref{eq:def_varphi} and Fact \ref{fact:d_infty_calculation}, we have\begin{equation}
    \label{eq:Dmax1}\D_\infty(\varphi^{\wtt{X}_{\overline{\mcC}}\wtt{Y}_{\overline{\mcC}}XYE_AE_BDU}\|\theta^{\wtt{X}_{\overline{\mcC}}\wtt{Y}_{\overline{\mcC}}XYE_AE_BDU})\leq -\log q.
\end{equation}
Let $p(a_{\mcC},b_{\mcC})$ be the probability of getting $(a_{\mcC},b_{\mcC})$ as the output when measuring $(A_{\mcC},B_{\mcC})$ registers in $\ket{\varphi}$ in the computational basis. From Fact \ref{fact:d_infty_calculation}, we have \begin{equation}
    \label{eq:Dmax2}
\D_{\infty}(\varphi_{a_{\mcC}b_{\mcC}}^{\wtt{X}_{\overline{\mcC}}\wtt{Y}_{\overline{\mcC}}XYE_AE_BDU}\|\varphi_{}^{\wtt{X}_{\overline{\mcC}}\wtt{Y}_{\overline{\mcC}}XYE_AE_BDU})\leq-\log p(a_{\mcC},b_{\mcC})
\end{equation} We have \begin{align*}
&\E_{a_{\mcC}b_{\mcC}\leftarrow \varphi^{A_{\mcC}B_{\mcC}}}\D_{\infty}(\varphi_{a_{\mcC}b_{\mcC}}^{\wtt{X}_{\overline{\mcC}}\wtt{Y}_{\overline{\mcC}}XYE_AE_BDU}\|\theta_{}^{\wtt{X}_{\overline{\mcC}}\wtt{Y}_{\overline{\mcC}}XYE_AE_BDU}) \\& \leq \E_{a_{\mcC}b_{\mcC}\leftarrow \varphi^{A_{\mcC}B_{\mcC}}}\Big[\D_{\infty}(\varphi_{a_{\mcC}b_{\mcC}}^{\wtt{X}_{\overline{\mcC}}\wtt{Y}_{\overline{\mcC}}XYE_AE_BDU}\|\varphi_{}^{\wtt{X}_{\overline{\mcC}}\wtt{Y}_{\overline{\mcC}}XYE_AE_BDU}) \\&\tab \tab + \D_{\infty}(\varphi_{}^{\wtt{X}_{\overline{\mcC}}\wtt{Y}_{\overline{\mcC}}XYE_AE_BDU}\|\theta_{}^{\wtt{X}_{\overline{\mcC}}\wtt{Y}_{\overline{\mcC}}XYE_AE_BDU}) \Big] &\mbox{(Fact \ref{fact:chain_rule_dmax})}
\\&\leq \E_{a_{\mcC}b_{\mcC}\leftarrow \varphi^{A_{\mcC}B_{\mcC}}}\left[-\log p(a_{\mcC},b_{\mcC}) -\log q\right] &\mbox{(eqs. \eqref{eq:Dmax1} and \eqref{eq:Dmax2})}
\\&\leq -\log q + |\mcC|\log(|\clA|\cdot|\clB|). &\mbox{(Fact \ref{fact:entropic_inequalities} and Definition \ref{def:entropy})}
\end{align*}
Thus, we get \begin{align*}
     &-\log q + |\mcC|\log(|\clA|\cdot|\clB|) \\& \geq \E_{a_{\mcC}b_{\mcC}\leftarrow \varphi^{A_{\mcC}B_{\mcC}}}\D_{\infty}(\varphi_{a_{\mcC}b_{\mcC}}^{\wtt{X}_{\overline{\mcC}}\wtt{Y}_{\overline{\mcC}}XYE_AE_BDU}\|\theta_{}^{\wtt{X}_{\overline{\mcC}}\wtt{Y}_{\overline{\mcC}}XYE_AE_BDU}) \\
     &\geq \E_{a_{\mcC}b_{\mcC}\leftarrow \varphi^{A_{\mcC}B_{\mcC}}}\D(\varphi_{a_{\mcC}b_{\mcC}}^{\wtt{X}_{\overline{\mcC}}\wtt{Y}_{\overline{\mcC}}XYE_AE_BDU}\|\theta_{}^{\wtt{X}_{\overline{\mcC}}\wtt{Y}_{\overline{\mcC}}XYE_AE_BDU})  &\mbox{(Fact \ref{fact:monotonicity})}
    \\ & \geq \E_{x_{\mcC}y_{\mcC}a_{\mcC}b_{\mcC}du\leftarrow \varphi^{X_{\mcC}Y_{\mcC}A_{\mcC}B_{\mcC}DU}}\D(\varphi_{x_{\mcC}y_{\mcC}a_{\mcC}b_{\mcC}du}^{\wtt{X}_{\overline{\mcC}}\wtt{Y}_{\overline{\mcC}}XYE_AE_B}\|\theta_{x_{\mcC}y_{\mcC}du}^{\wtt{X}_{\overline{\mcC}}\wtt{Y}_{\overline{\mcC}}XYE_AE_B}). &\mbox{(Fact \ref{fact:chain_rule_D})}
\end{align*}
\end{proof}

For the $n$-copy game $G_n$, we condition on success on a subset $\mcC \subseteq [n]$ of size $r$. The following Claim says that either the probability of success in $\mcC$ is already (exponentially) small, otherwise there exists a coordinate  $j\notin \mcC$ for which the probability of winning in the coordinate $j$, conditioned on success in $\mcC$, is close to $\vall(G_1)$ and hence is bounded away from $1$. 
\begin{claim}\label{claim:main}
    Consider the $n$-copy game $G_n$. There exists a set $\{i_1, \dots i_{k'}\}$ such that $\forall~ 1\leq r< k'=\lfloor\delta_1 n\rfloor$, at least one of the following conditions hold:
    \begin{itemize}
        \item  $\Pr\left(T^{(r)}=1\right) \leq 2^{-\delta_2 n} .$
        \item $\Pr\left(T_{i_{r+1}}|T^{(r)}=1\right) \leq 1-\eps_*+23\sqrt{\delta_*}. $
    \end{itemize}
    where $T^{(r)}=\prod_{j=1}^rT_{i_j}$ and $T_i=V(X_i,Y_i,A_i,B_i)$.
\end{claim}
\noindent We first state Claim \ref{claim:all_markovs} and Claim \ref{claim:all_equations} (whose proofs are deferred to the Appendix) and use them to prove Claim \ref{claim:main}.  
\begin{claim}\label{claim:all_markovs}
 Suppose for some $1\leq r< k'$, we have\[\Pr\left(T^{(r)}=1\right) > 2^{-\delta_2 n} .
        \] Let $\mcC \defeq \{i_1, \dots i_r\}$. Then, for a randomly selected coordinate $j$ outside $\mcC$ with probability at least $\frac{1}{4}$ (over the choice of $j$), the following conditions hold:
    \begin{equation}
    \I(X_j:\wtt{B}|R_jD_jU_j)_{\varphi} \leq \frac{8\delta_3}{(1-\delta_1)}\leq 10\delta_3. \label{eq:condition1} \end{equation}
     \begin{equation}\I(Y_j:\wtt{A}|R_jD_jU_j)_{\varphi} \leq \frac{8\delta_3}{(1-\delta_1)}\leq 10\delta_3.\label{eq:condition2} \end{equation}
      \begin{equation}
\D(\varphi^{{X}_j{Y}_j}\|\theta^{X_jY_j})\leq \frac{8\delta_3}{(1-\delta_1)}\leq 10\delta_3.\label{eq:condition3} \end{equation}
     \begin{equation}\frac{1}{2}\left(\I(X_j:R_j|Y_j)_{\varphi} + \I(Y_j:R_j|X_j)_{\varphi} \right) \leq \frac{8\delta_3}{(1-\delta_1)}\leq 10\delta_3.\label{eq:condition4} 
     \end{equation}
     \begin{equation}
 \E_{r_jd_ju_j\leftarrow \varphi^{R_jD_jU_j}}\D(\varphi_{r_jd_ju_j}^{X_jY_j}\|\theta_{d_ju_j}^{X_jY_j})  \leq \frac{8\delta_3}{(1-\delta_1)}\leq 10\delta_3. \label{eq:intermediate_markov}
\end{equation}
\begin{equation}
    \label{eq:last_markov}
\I(\wht{X}_j:\wht{B}|R_jD_jU_j)_{\varphi} \leq \frac{8\delta_3}{(1-\delta_1)}\leq 10\delta_3.
\end{equation}
\end{claim}
\begin{proof}We defer the proof of this claim to the Appendix.
\end{proof}

\begin{claim}\label{claim:all_equations}
    Given that eqs. \eqref{eq:condition1} to \eqref{eq:last_markov} in Claim \ref{claim:all_markovs} hold (which happens with probability at least $\frac{1}{4}$ for a randomly selected coordinate $j$ outside $\mcC$), the following eqs. hold (recall the notation from eq. \eqref{eq:notation}):

 \noindent \textbf{Existence of unitaries (Fact \ref{fact:unitary_existence})} \begin{equation}
    \label{eq:existence_unitary_1}
\I(X_j:\wtt{B}|R_j)_{\varphi_*} \leq\delta_*.
\end{equation}
\begin{align}
\I(\wht{Y}_jC_j:\wtt{A}|R_j)_{\varphi_*}\leq \delta_*. \label{eq:existence_unitary_2}
\end{align}
\begin{equation}
     \E_{r_j\leftarrow \varphi_*^{R_j}}\onenorm{\varphi_{*,r_j}^{X_j\wht{Y}_jC_j}}{\varphi_{*,r_j}^{X_j}\otimes\varphi_{*,r_j}^{\wht{Y}_jC_j}}\leq \sqrt{\delta_*}. \label{eq:existence_unitary_3}\end{equation}
   \noindent \textbf{Correlated sampling of $R_j$ }

     \begin{equation}
\I(X_j:R_j|\wht{Y}_jC_j)_{\varphi_*} \leq\delta_*. \label{eq:joint_sampling1}
\end{equation}
\begin{equation}
    \I(R_j:\wht{Y}_jC_j|X_j)_{\varphi_*} \leq\delta_*.
\label{eq:joint_sampling2}
\end{equation}
\noindent \textbf{Correctly incorporating $\wht{X}_j$ after the \textbf{DELAY}}

We have following two {Markov Chains} (Definition \ref{def:markov}) in state $\ket{\varphi_{}}$ (here $\widehat{B},\wht{X}_j,X_j,R_j$ are classical):
\begin{equation}
\label{eq:markov_generate_B}
X_jE_A-\widehat{B}\wht{X}_jR_j-B_j.
\end{equation}
\begin{equation}
\label{eq:markov_in_phi}
E_A-X_j\widehat{B}R_j-
\wht{X}_jB_j.
\end{equation}
\begin{align}
 \E_{r_jx_j\leftarrow \varphi^{R_jX_j}} \D(\varphi_{r_jx_j}^{\wht{X}_j\wht{B}}\|\varphi_{r_jx_j}^{\wht{X}_j}\otimes \varphi_{r_jx_j}^{\wht{B}}) \leq 20\delta_3. \label{eq:tensor-lemma}
\end{align}
\begin{align}    
\onenorm{\varphi_{*}^{R_jX_j}}{\varphi^{R_jX_j}} \leq 2\sqrt{\delta_*}. \label{eq:R_jX_j_star}
\end{align}
  \begin{equation}   \label{eq:closeness_in_D_first}
\onenorm{\varphi_{*}^{R_jX_j\wht{B}}}{\varphi^{R_jX_j\wht{B}}} \leq3\sqrt{\delta_*}.  
         \end{equation}
\begin{align}
\onenorm{ \varphi^{X_jR_jE_A\wht{B}\wht{X}_jB_j}}{ \varphi_*^{X_jR_jE_A\wht{B}}\cdot \theta^{\wht{X}_j}_{X_j\wht{Y}_jC_j} \cdot\varphi^{B_j}_{R_j\wht{X}_j\wht{B}}}\leq 4\sqrt{\delta_*}. 
            \label{eq:closeness_hybrid_P1}
\end{align}
\end{claim}
\begin{proof}We defer the proof of this claim to the Appendix.
\end{proof}

\subsection*{Strategy for $G_1$}
The strategy for the 1-copy game will be as follows (see Figure \ref{fig:Strategy}):
\begin{figure}[H]
    \centering
\fbox{\begin{minipage}{40em} \vspace{2mm}
\begin{enumerate}
    \item $\charlie$ and $\dave$ get their inputs $x_j$ and $y_j=\wht{y}_jc_j\wht{x}_j$ according to the anchored distribution $\mu$ (eq. \eqref{eq:anchored_def}).
    \item Using Fact \ref{fact:hollenstein_sampling} and eqs. \eqref{eq:joint_sampling1} and \eqref{eq:joint_sampling2}, both of them jointly sample $R_j=r_j$ (here $\dave$ does not have access to $\wht{x}_j$).
    \item They start with the shared state $\ket{\varphi_{*,r_j}}$. Using eqs. \eqref{eq:existence_unitary_1} and \eqref{eq:existence_unitary_2}, they use the unitary operators $\{U_{x_j,r_j}\}_{x_j}$ and $\{V_{\wht{y}_j,c_j,r_j} \}_{\wht{y}_j,c_j}$ from 
    Fact \ref{fact:unitary_existence} to create the state \begin{equation}
        \label{eq:state_after_unitaries}
    (U_{x_j,r_j} \otimes V_{\wht{y}_j,c_j,r_j} ) \ketbra{\varphi_{*,r_j}}{\varphi_{*,r_j}} (U_{x_j,r_j}^\dagger \otimes V^\dagger_{\wht{y}_j,c_j,r_j}) .\end{equation}  
  \item If $c_j=0$, before the \textbf{DELAY} in the protocol, $\dave$ measures his register $B_j$ to generate the answer. 
\item If $c_j=1$, a $\CNOT$ acts between all registers of $\dave$ and $\ket{0}^{N''}$ after which $N''$ cannot be accessed by $\dave$ (corresponding to the \textbf{DELAY} in the restricted strategy). After this, $\dave$ gets access to $\wht{x}_j$.
\item  $\dave$ now creates $B_j$ (locally) according to the distribution $\varphi^{B_j}_{\wht{x},r_j\wht{b}}$. 
\end{enumerate}
\end{minipage}}
 \caption{Protocol $P$ for $G_1$.}
 \label{fig:Strategy}
\end{figure}

\section*{Proof of Claim \ref{claim:main} }
\begin{proof}
Let $\mcC \defeq \{i_1, \dots i_r\}$.
    Suppose there exists some protocol (with the restricted strategy) that wins in $\mcC$ with a probability greater than $2^{-\delta_2 n}$, and conditioning on success in $\mcC$, let the probability of winning the game in coordinate $j$ be $\omega$. In other words,  consider the state $\varphi^{X_jY_jR_jE_AE_B}$ where $X_jE_A$ is given to $\charlie$, $Y_jE_B$ is given to $\dave$ and $R_j$ is shared between them. On measuring their answer registers $A_j$ and $B_j$ in the computational basis (and outputting them), the probability that they win the game in the $j$-th coordinate (i.e., $V(X_j,Y_j,A_j,B_j)=1$) with the restricted strategy is $\omega$. Note that $\wtt{B}=\wtt{Y}_{\widebar{\mcC}}YE_B=\widehat{B}\wht{X}B_{\widebar{\mcC}}N'$. Since $\wht{Y}_jC_j$ is contained in $\wht{B}$, we can consider the state $\varphi^{X_j\wht{X}_jR_jE_A\wht{B}B_j}$ and trace out the other registers as they do not affect $V(X_j,Y_j,A_j,B_j)$.

    We now consider protocols $P_1, P_2$ and $P_3$. We show that the states created in these protocols are close to $\varphi^{X_j\wht{X}_jR_jE_A\wht{B}B_j}$, which will show that the probability of winning the game in the $j$-th coordinate for these protocols are close to $\omega$. Note that the final protocol is the same as the protocol $P$ in Fig. \ref{fig:Strategy} (which is the strategy we use for winning $G_1$). This will help us in upper bounding the value of $\omega$ since $\val(G_1)\leq 1-\eps_*$ (Claim \ref{claim:value_G_1}).

\begin{itemize}
    \item Consider a protocol $P_1$ where $\charlie$ and $\dave$ get questions from $(x_j,\wht{y}_jc_j\wht{x}_j)\leftarrow \varphi_*^{X_j\wht{Y}_jC_j}\cdot \theta^{\wht{X}_j}_{X_j\wht{Y}_jC_j}$ and they share $r_j\leftarrow\varphi_{*,x_j\wht{y}_jc_j}^{R_j}$. They also share the state $\ket{\varphi_{*,r_j}}$ for all $r_j$.  Consider the starting state $\ket{\varphi_{*,r_j}}$. From 
    Fact \ref{fact:unitary_existence},   $\charlie$ and $\dave$ use the unitary operators $\{U_{x_j,r_j}\}_{x_j}$ and $\{V_{\wht{y}_j,c_j,r_j} \}_{\wht{y}_j,c_j}$  to create the state 
    \begin{align}
\sigma=\E_{(x_j\wht{y}_jc_jr_j)\leftarrow\varphi_*^{X_j\wht{Y}_jC_jR_j}}  (U_{x_j,r_j} \otimes V_{\wht{y}_j,c_j,r_j} ) \ketbra{\varphi_{*,r_j}}{\varphi_{*,r_j}} (U_{x_j,r_j}^\dagger \otimes V^\dagger_{\wht{y}_j,c_j,r_j}). \label{eq:define_sigma} \end{align}
$\charlie$ and $\dave$ now consider the state $\sigma^{X_jR_jE_A\wht{B}}$ (and trace out the other registers of $\dave$). $\charlie$ then measures his answer register $A_j$ in the computational basis. If $c_j=0$, before the $\CNOT$ in the protocol, $\dave$ measures his register $B_j$ to generate the answer. If $c_j=1$, a $\CNOT$ acts between all registers of $\dave$ and $\ket{0}^{N''}$ after which $N''$ cannot be accessed by $\dave$. Since the relevant registers of $\dave$ for winning the game in the $j$-th coordinate are $\wht{X}_jR_j\wht{B}B_j$, which are classical, this $\CNOT$ does not change the state on the relevant registers.   After this, $\dave$ gets access to $\wht{x}_j$ (which has distribution $\theta^{\wht{X}_j}_{x_j\wht{y}_jc_j}$). $\dave$ now creates $B_j$ (locally) according to the distribution $\varphi^{B_j}_{\wht{x},r_j\wht{b}}$. Thus, the resultant state created in the protocol is \[\sigma^{X_jR_jE_A\wht{B}}\cdot \theta^{\wht{X}_j}_{X_j\wht{Y}_jC_j} \cdot\varphi^{B_j}_{R_j\wht{X}_j\wht{B}}.\]

    Claim \ref{claim:closeness_of_sigma} along with eq.~\eqref{eq:closeness_hybrid_P1} gives \[\onenorm{ \varphi^{X_jR_jE_A\wht{B}\wht{X}_jB_j}}{\sigma^{X_jR_jE_A\wht{B}}\cdot \theta^{\wht{X}_j}_{X_j\wht{Y}_jC_j} \cdot\varphi^{B_j}_{R_j\wht{X}_j\wht{B}}}\leq 14\sqrt{\delta_*}, \]
    and hence, the winning probability for the protocol $P_1$ is at least  $\omega-14\sqrt{\delta_*}$.

            \item Consider a protocol $P_2$ where $\charlie$ and $\dave$ get questions $(x_j,\wht{y}_jc_j\wht{x}_j)\leftarrow \varphi_*^{X_j\wht{Y}_jC_j}\cdot\theta^{\wht{X}_j}_{X_j\wht{Y}_jC_j}$. {Fact $\ref{fact:approx_markov_sampling}$  and Fact \ref{fact:hollenstein_sampling} along with eqs. \eqref{eq:joint_sampling1} and \eqref{eq:joint_sampling2} allows joint-sampling of $R_j\leftarrow\varphi_{*,x_j\wht{y}_jc_j}^{R_j}$ by $\charlie$ and $\dave$ using public randomness $S$ and local functions $f_A$ and $f_B$ (here ${R}_j$ is classical and $\wtt{R}_j$ is a copy of ${R}_j$):} \begin{equation*}
\onenorm{\varphi_*^{X_j\wht{Y}_jC_jR_j\wtt{R}_j}}{\varphi_*^{X_j\wht{Y}_jC_j}f_{A}(S,X_j)f_B(S,\wht{Y}_jC_j)} \leq 5\sqrt{\delta_*}.
            \end{equation*} They then execute the same strategy as in $P_1$. Thus, we get that the winning probability for the protocol $P_2$ is at least $\omega-14\sqrt{\delta_*}-5\sqrt{\delta_*}=\omega-19\sqrt{\delta_*}$.
           
    \item Consider a protocol $P_3$ where $\charlie$ and $\dave$ get questions  $(x_j,y_j)\leftarrow \theta^{X_j\wht{Y}_jC_j\wht{X}_j}$ and they execute same strategy as in $P_2$. From eq. \eqref{eq:closeness_in_D_first} and Fact \ref{fact:data}, we get
 \[ \onenorm{\varphi_{*}^{X_j\wht{Y}_jC_j}}{\varphi^{X_j\wht{Y}_jC_j}} \leq 3\sqrt{\delta_*}. 
   \]From eq. \eqref{eq:condition3} and Fact \ref{fact:pinsker}, we have \[ \onenorm{\theta^{X_j\wht{Y}_jC_j}}{\varphi^{X_j\wht{Y}_jC_j}} \leq \sqrt{10\delta_3}. 
   \] Thus, we get that \begin{align*}
\onenorm{\theta^{X_j\wht{Y}_jC_j\wht{X}_j}}{\varphi_*^{X_j\wht{Y}_jC_j}\cdot \theta^{\wht{X}_j}_{X_j\wht{Y}_jC_j}} &=\onenorm{\theta^{X_j\wht{Y}_jC_j}\cdot \theta^{\wht{X}_j}_{X_j\wht{Y}_jC_j}}{\varphi_*^{X_j\wht{Y}_jC_j}\cdot \theta^{\wht{X}_j}_{X_j\wht{Y}_jC_j}}
\\&=\onenorm{\theta^{X_j\wht{Y}_jC_j}}{\varphi_*^{X_j\wht{Y}_jC_j}}
\\&\leq  3\sqrt{\delta_*}+\sqrt{10\delta_3},
   \end{align*} and hence, the winning probability for the protocol $P_3$ is at least $\omega-22\sqrt{\delta_*}-\sqrt{10\delta_3}$. Note that $P_3$ is the same as $P$ in Fig. \ref{fig:Strategy}, i.e., our strategy for game $G_1$. From Claim \ref{claim:value_G_1},  we get that $\omega-22\sqrt{\delta_*}-\sqrt{10\delta_3}\leq 1-\eps_*.$
   
\end{itemize}
Using $\delta_3=\delta_*^2/10^5$, we get
$\omega\leq 1-\eps_*+23\sqrt{\delta_*}.$
\end{proof}
\begin{claim}\label{claim:closeness_of_sigma}
    \begin{align*} &\onenorm{\varphi_*^{X_jR_jE_A\wht{B}}\cdot \theta^{\wht{X}_j}_{X_j\wht{Y}_jC_j} \cdot\varphi^{B_j}_{R_j\wht{X}_j\wht{B}}}{\sigma^{X_jR_jE_A\wht{B}}\cdot \theta^{\wht{X}_j}_{X_j\wht{Y}_jC_j} \cdot\varphi^{B_j}_{R_j\wht{X}_j\wht{B}}}\leq 10\sqrt{\delta_*},
\end{align*}where $\sigma$ is from eq. \eqref{eq:define_sigma}.
\end{claim}
\begin{proof}
    
    Define $\eps^a_{r_j}\defeq \I(X_j:\wtt{B})_{\varphi_{*,r_j}} $ and $\eps^b_{r_j}\defeq\I(\wht{Y}_jC_j:\wtt{A})_{\varphi_{*,r_j}}$. From eqs. \eqref{eq:existence_unitary_1} and \eqref{eq:existence_unitary_2}, we have
    \begin{align} \label{eq:expectation_epsilon}
\E_{r_j\leftarrow\varphi_*^{R_j}} ~\eps^a_{r_j} \leq \delta_* \text{ and }
\E_{r_j\leftarrow\varphi_*^{R_j}} ~\eps^b_{r_j} \leq \delta_*. 
    \end{align}
 Using Fact \ref{fact:unitary_existence} (with the starting state $\ket{\varphi_{*,r_j}}$), we have \[
\E_{(x_j\wht{y}_jc_j)\leftarrow\varphi_{*,r_j}^{X_j\wht{Y}_jC_j}}
[\big\|(U_{x_j,r_j} \otimes V_{\wht{y}_j,c_j,r_j} ) \ketbra{\varphi_{*,r_j}}{\varphi_{*,r_j}} (U_{x_j,r_j}^\dagger \otimes V^\dagger_{\wht{y}_j,c_j,r_j}) -\ketbra{\varphi_{*,r_jx_j\wht{y}_jc_j}}{\varphi_{*,r_jx_j\wht{y}_jc_j}} \big\|_1
]\]\begin{equation*}
\leq 4\sqrt{\eps^a_{r_j}}+4\sqrt{\eps^b_{r_j}} + 2 \onenorm{\varphi_{*,r_j}^{X_j\wht{Y}_jC_j}}{\varphi_{*,r_j}^{X_j}\otimes\varphi_{*,r_j}^{\wht{Y}_jC_j}}.\end{equation*}
    Taking $\E_{r_j\leftarrow \varphi_*^{R_j}}$, we get
\[
\E_{(x_j\wht{y}_jc_jr_j)\leftarrow\varphi_*^{X_j\wht{Y}_jC_jR_j}}
[\big\|(U_{x_j,r_j} \otimes V_{\wht{y}_j,c_j,r_j} ) \ketbra{\varphi_{*,r_j}}{\varphi_{*,r_j}} (U_{x_j,r_j}^\dagger \otimes V^\dagger_{\wht{y}_j,c_j,r_j})-\ketbra{\varphi_{*,r_jx_j\wht{y}_jc_j}}{\varphi_{*,r_jx_j\wht{y}_jc_j}} \big\|_1]\]\vspace{-8mm}\begin{align} & \leq \E_{r_j\leftarrow \varphi_*^{R_j}} \left(4\sqrt{\eps^a_{r_j}}+4\sqrt{\eps^b_{r_j}} + 2 \onenorm{\varphi_{*,r_j}^{X_j\wht{Y}_jC_j}}{\varphi_{*,r_j}^{X_j}\otimes\varphi_{*,r_j}^{\wht{Y}_jC_j}}\right)
 \nonumber
\\&\leq 4\sqrt{\delta_*}+ 4\sqrt{\delta_*} + 2 \sqrt{\delta_*} =10\sqrt{\delta_*}.\label{eq:closeness_P2}\end{align}
Here, the last inequality follows from eqs. \eqref{eq:expectation_epsilon} and \eqref{eq:existence_unitary_3} and concavity of square-root. 
Consider\begin{align} &\onenorm{\varphi_*^{X_jR_jE_A\wht{B}}\cdot \theta^{\wht{X}_j}_{X_j\wht{Y}_jC_j} \cdot\varphi^{B_j}_{R_j\wht{X}_j\wht{B}}}{\sigma^{X_jR_jE_A\wht{B}}\cdot \theta^{\wht{X}_j}_{X_j\wht{Y}_jC_j} \cdot\varphi^{B_j}_{R_j\wht{X}_j\wht{B}}}\\
&=\onenorm{\sigma^{X_jR_jE_A\wht{B}}}{\varphi_*^{X_jR_jE_A\wht{B}}} \nonumber\\&\leq \Big\|\E_{(x_j\wht{y}_jc_jr_j)\leftarrow\varphi_*^{X_j\wht{Y}_jC_jR_j}}  (U_{x_j,r_j} \otimes V_{\wht{y}_j,c_j,r_j} ) \ketbra{\varphi_{*,r_j}}{\varphi_{*,r_j}} (U_{x_j,r_j}^\dagger \otimes V^\dagger_{\wht{y}_j,c_j,r_j}) \nonumber \\&\hspace{1.5cm}-{\E_{(x_j\wht{y}_jc_jr_j)\leftarrow\varphi_*^{X_j\wht{Y}_jC_jR_j}} \ketbra{\varphi_{*,r_jx_j\wht{y}_jc_j}}{\varphi_{*,r_jx_j\wht{y}_jc_j}}  }\Big\|_1 \nonumber \\&\leq 10\sqrt{\delta_*}.  &\mbox{(eq. \eqref{eq:closeness_P2})} \nonumber
\end{align}
\end{proof}

\begin{claim}\label{claim:main_avg}
    Let $k'=\lfloor\delta_1 n\rfloor$. From Claim \ref{claim:value_G_1}, we have $\val(G_1)\leq1-\eps_*$. Consider the $n$-copy game $G_n$. For  $0<\delta_1,\delta_2, \delta_3, \delta_*< 0.1$ where $\delta_3=\delta_*^2/10^5$ and \[\delta_3=\delta_2+\delta_1\log(|\clA|\cdot|\clB|),\] suppose there  exists a set $\mcC=\{i_1, \dots i_{r}\}$ for some $1\leq r< k'$ such that
   \[\Pr\left(T^{(r)}=1\right) > 2^{-\delta_2 n} \]
   where $T^{(r)}=\prod_{j=1}^rT_{i_j}$ and $T_i=V(X_i,Y_i,A_i,B_i)$. Then
    
    \[\E_{i_{r+1}\leftarrow \widebar{\mcC}}\Pr\left(T_{i_{r+1}}|T^{(r)}=1\right) \leq1-\frac{1}{4}(\eps_*-23\sqrt{\delta_*}) . \]
\end{claim}

\begin{proof}Consider the following two cases:
   \begin{itemize}
       \item Case 1- Sampling a random coordinate $j$ outside $\mcC$ satisfies all the six Markov's inequalities from Claim \ref{claim:all_markovs} (eqs. \eqref{eq:condition1}, \eqref{eq:condition2}, \eqref{eq:condition3}, \eqref{eq:condition4}, \eqref{eq:intermediate_markov} and \eqref{eq:last_markov}). This happens with probability at least $\frac{1}{4}$. Setting $i_{r+1}=j$, from Claim \ref{claim:main}, we get \[\Pr\left(T_{i_{r+1}}|T^{(r)}=1\right) \leq 1-\eps_*+23\sqrt{\delta_*}. \]
        \item Case 2- Sampling a random $j$ outside $\mcC$ does not satisfy at least one of the 6 Markov's inequalities. In this case,  we get \[\Pr\left(T_{i_{r+1}}|T^{(r)}=1\right) \leq 1. \]
   \end{itemize}
   We thus have
   \begin{align*}
   \E_{i_{r+1}\leftarrow \widebar{\mcC}}\Pr\left(T_{i_{r+1}}|T^{(r)}=1\right) &\leq \frac{1}{4}(1-\eps_*+23\sqrt{\delta_*}) + \frac{3}{4}\cdot 1
   \\&\leq 1-\frac{1}{4}(\eps_*-23\sqrt{\delta_*}).\end{align*}
\end{proof}

\subsection*{Threshold theorem}
\noindent  This proof follows ideas similar to \cite{concentration,anup_rao}. 
 Let $t$ be the largest  integer such that \[2^{-\delta_2n} \leq \left(1-\frac{1}{4}(\eps_*-23\sqrt{\delta_*})\right)^t.\] This gives \begin{equation}
\label{eq:def_t}
 t=\Bigg\lfloor\frac{-\delta_2n}{\log\left(1-\frac{1}{4}(\eps_*-23\sqrt{\delta_*})\right)}\Bigg\rfloor.\end{equation}

\begin{claim} \label{claim:before_threshold} Let $t$ be from eq. \eqref{eq:def_t}.  Let $i_1,i_2 \dots i_k$ denote a sequence of distinct random elements of $[n]$. For each $j \in [t]$, let $S_j=\{i_1,i_2\dots i_j\}$. Let $W_S$ denote the event that the players win in the coordinates
included in the set $S$. We have
\[\Pr(W_{S_t})\leq 2\left(1-\frac{1}{4}(\eps_*-23\sqrt{\delta_*})\right)^t.\]    
\end{claim}
\begin{proof}
    For each $j \in [t]$, let $L_j$ denote the event determined by $S_j$ that \[\Pr[W_{S_j}] \leq \left(1-\frac{1}{4}(\eps_*-23\sqrt{\delta_*})\right)^t,\] and $\widebar{L}_j$ denote the complement event. Then \[\Pr[W_{S_t}] = \Pr[W_{S_t} \land L_t] + \Pr[W_{S_t} \land \widebar{L}_t].\]
    The first term is bounded by $\Pr[W_{S_t}|L_t] \leq \left(1-\frac{1}{4}(\eps_*-23\sqrt{\delta_*})\right)^t$. 
\[\Pr[W_{S_t} \land \widebar{L}_t] = \prod^t_{j=1}
\Pr[W_{S_j} \land \widebar{L}_j |W_{S_{j-1}} \land \widebar{L}_{j-1}].\]
By Claim \ref{claim:main_avg}, the $j^{th}$ term in this product is bounded by $\Pr[W_{S_j}
|W_{S_{j-1}} \land \widebar{L}_{j-1}] \leq 1-\frac{1}{4}(\eps_*-23\sqrt{\delta_*}).$
Thus,\[ \Pr[W_{S_t}
] \leq 2\left(1-\frac{1}{4}(\eps_*-23\sqrt{\delta_*})\right)^t.\]
\end{proof}

\subsection*{Proof of Claim \ref{claim:threshold}}
\begin{proof}
Recall that $$\delta_0 n=t+100\sqrt{\delta_*}n.$$    Whenever
$\left(\sum_{i\in [n]}W_i\geq (1-\frac{1}{4}\eps_*+\delta_0)n\right)$, let us pick a random subset $S$ of size $t$ from the set of coordinates where the
players won and we blame this set for this bad event. Then, the probability that we blame any fixed
set $S$ is at most ${\Pr(W_S)}{\binom{(1-\frac{1}{4}\eps_*+\delta_0)n}{t}}^{-1}$. By the union bound, we have \begin{align*}
&\Pr\left(\sum_{i\in [n]}W_i\geq (1-\frac{1}{4}\eps_*+\delta_0)n\right)_{G_n} \\&\leq \sum_{S}{\Pr(W_S)}{\binom{(1-\frac{1}{4}\eps_*+\delta_0)n}
 {t}}^{-1}
 \\&\leq 2{\binom{n}
 {t}}\left(1-\frac{1}{4}(\eps_*-23\sqrt{\delta_*})\right)^t{\binom{(1-\frac{1}{4}\eps_*+\delta_0)n}
 {t}}^{-1} &\mbox{(Claim \ref{claim:before_threshold})}
 \\&\leq 2\left(1-\frac{1}{4}(\eps_*-23\sqrt{\delta_*})\right)^t \left(\frac{n}{(1-\frac{1}{4}\eps_*+\delta_0)n-t+1}\right)^t
\\& \leq 2 \left(\frac{n\left(1-\frac{1}{4}(\eps_*-23\sqrt{\delta_*})\right)}{(1-\frac{1}{4}\eps_*+\delta_0)n-t+1}\right)^t.
\end{align*}
On setting $\delta_0 =\frac{t}{n}+100\sqrt{\delta_*}$, we get
\begin{align*}
 \Pr\left(\sum_{i\in [n]}W_i\geq (1-\frac{1}{4}\eps_*+\delta_0)n\right)_{G_n}    &\leq 2 \left(\frac{\left(1-\frac{1}{4}(\eps_*-23\sqrt{\delta_*})\right)}{(1-\frac{1}{4}\eps_*+100\sqrt{\delta_*})}\right)^t\\&\leq 2^{-\delta' n}, &\mbox{(eq. \eqref{eq:def_t})}
\end{align*}
for some small constant $\delta'>0$.
\end{proof}

\noindent {\bf Acknowledgements:} We thank Pranab Sen for helpful discussions. 

This work is supported by the National Research Foundation, Singapore, through the National Quantum Office, hosted in A*STAR, under its Centre for Quantum Technologies Funding Initiative (S24Q2d0009). This work was done in part while R.J. was visiting the
Technion-Israel Institute of Technology, Haifa, Israel, and the Simons
Institute for the Theory of Computing, Berkeley, CA, USA, and when U.K. and S.C. were at the Centre for Quantum Technologies, National University of Singapore.

\newpage

\appendix

\section{Appendix}\vspace{-2mm}
\subsection*{Composition theorem}\label{sec:composition}
\vspace{-2mm}

Consider two cryptographic tasks $f$ and $g$, where $\textsc{ideal}(f)$ and $\textsc{ideal}(g)$ are the ideal functionalities for $f$ and $g$ respectively.
Let $\Pf$ be an (honest) protocol for $f$  using devices $\device_f$.
Let $\Pgf$ (see Table~\ref{prot:pgf}) be an (honest) oracle-aided protocol for $g$ using devices $\device_g$ and oracle calls to $\textsc{ideal}(f)$.
{Recall that we follow the convention that $\widehat{A}$ and $\widehat{B}$ denote all residual registers with $\alice$ and $\bob$ (other than the ones mentioned explicitly), hence they can denote different sets of registers in different states.} 

Consider an (honest) protocol $\Pg$ for task $g$ defined as follows: Take protocol $\Pgf$ and replace all calls to $\textsc{ideal}(f)$ by $\Pf$.
For example, if $\Pgf$ does one query to $\textsc{ideal}(f)$, then $\Pg$ can be described by  Table~\ref{prot:pgpf}. Note that multiple queries to $\textsc{ideal}(f)$ can also be handled similarly (assuming the remaining queries to be part of the communication protocol before or after $\textsc{ideal}(f)$ query).
\begin{table}[!h]
        \centering
        \resizebox{12cm}{6cm}{
            \begin{tabular}{  l c r }
            \hline \\
            $\alice$ (input $A_g$) & & $\bob$ (input $B_g$) \\
            \hline \\

 $A_g\widehat{A}$ & $\eta$ & $B_g\widehat{B}$

\\ \hdashline\\  

& \large{COMMUNICATION PROTOCOL} & \\ \\ 
& \large{PRIOR TO $\textsc{{ideal}}(f)$ QUERY} & \\ \\ 
$A_gA_f\widehat{A}$ & $\theta$ & $B_gB_f\widehat{B}$
\\
 \hdashline \\

 \\
& 
\centering
\resizebox{0.3\textwidth}{!}{%
\begin{circuitikz}
\tikzstyle{every node}=[font=\normalsize]
\draw (5.25,14.75) rectangle (8.5,11.75);
\node [font=\large] at (7,13.25) {IDEAL($f$)};
\draw [->, >=Stealth, dashed] (3.5,14.25) -- (5.25,14.25);
\draw [->, >=Stealth, dashed] (9.75,14.25) -- (8.5,14.25);
\node [font=\normalsize] at (4.5,14.5) {$A_f$};
\node [font=\normalsize] at (9,14.75) {$B_f$};
\draw [->, >=Stealth, dashed] (5.25,13.5) -- (3.5,13.5);
\node [font=\normalsize] at (4.5,13.75) {$O_{A_f}^\ideal$};
\draw [->, >=Stealth, dashed] (3.5,12.75) -- (5.25,12.75);
\node [font=\normalsize] at (4.5,13) {${A_f}_\perp$};
\draw [->, >=Stealth, dashed] (8.5,12.25) -- (10,12.25);
\node [font=\normalsize] at (9.25,12.5) {$O_{B_f}^\ideal$};
\end{circuitikz}
}%

 & \\

  \\ $O_{A_f}A_gA_f\widehat{A}$ & $\tau$ & $O_{B_f}B_gB_f\widehat{B}$\\
 \hdashline \\  
  & \large{COMMUNICATION PROTOCOL} & \\ \\ 
& \large{AFTER $\textsc{{ideal}}(f)$ QUERY} & \\  Output $O_{A_g}$  && Output $O_{B_g}$\\ \\$O_{A_g}O_{A_f}A_gA_f\widehat{A}$ & $\zeta$ & $O_{B_g}O_{B_f}B_gB_f\widehat{B}$\\
 \hdashline \\  
    \end{tabular} }
        \caption{Protocol $\Pgf$($\alice, \bob$)}
        \label{prot:pgf} 
    \end{table}

    \begin{table}[!h]
        \centering \resizebox{12cm}{5cm}{
            \begin{tabular}{  l c r }
            \hline \\
            $\alice$ (input $A_g$) & \hspace{1cm} & $\bob$ (input $B_g$) \\
            \hline \\
            $A_g\widehat{A}$ & $\eta$ & $B_g\widehat{B}$\\
 \hdashline \\

& \large{COMMUNICATION PROTOCOL} & \\ \\ 
& \large{PRIOR TO $\textsc{{ideal}}(f)$ QUERY} & \\ \\ $A_gA_f\widehat{A}\overline{A}$ & $\theta$ & $B_gB_f\widehat{B}\overline{B}$\\
 \hdashline \\  
 & \large{Run $\Pf$($\alice, \bob$) with input $A_f, B_f$} & \\  \\ 
  Obtain $O_{A_{f}}$ & &   Obtain $O_{B_{f}}$  \\  \\  $O_{A_{f}}A_gA_f\widehat{A}\overline{A}$ & $\tau$ & $O_{B_{f}}B_gB_f\widehat{B}\overline{B}$\\
 \hdashline \\  
  & \large{COMMUNICATION PROTOCOL} & \\ \\ 
& \large{AFTER $\textsc{{ideal}}(f)$ QUERY} & \\  Output $O_{A_g}$  && Output $O_{B_g}$\\\\ $O_{A_g}O_{A_{f}}A_gA_f\widehat{A}\overline{A}$ & $\zeta$ & $O_{B_g}O_{B_{f}}B_gB_f\widehat{B}\overline{B}$\\
 \hdashline \\  

    \end{tabular} }
        \caption{Protocol  $\Pg$($\alice, \bob$)}
        \label{prot:pgpf} 
    \end{table}

Any of the protocols below can use quantum resources and devices, but the communication is limited to classical communication. We adopt the convention that if any honest party aborts at any stage of the protocol, that party does not invoke any further subroutine or take any further part in the protocol execution. Hence, security and correctness guarantees are not required after a party aborts. This also ensures that we do not need to consider $\bot$ inputs to any of the parties in the descriptions of the honest protocols.

{\noindent {\bf Remark:} Essentially the same proof would also work when we consider the composability of general (device-dependent) protocols.}

\begin{theorem}[Composition theorem]\label{thm:composition}
If $\Pgf$ is augmented-secure and $\Pf$ is augmented-secure, then $\Pg$ is augmented-secure. 
\end{theorem}
\begin{proof}
We will first prove the theorem assuming that $\Pgf$ makes only one query to the $\textsc{ideal}({f})$.
The proof for the general case when $\Pgf$ makes polynomially many queries (sequentially) to the $\textsc{ideal}({f})$ follows similarly (since we deal with negligible closeness of states for augmented-security).
We will also show security only in the case $\alice^*$ is cheating, and $\bob$ is honest.
The other case follows by merely switching parties.

Consider protocol $\mathcal{P}_f$($\alice, \bob$). 
In an honest implementation of this, let us say that $\alice$ and $\bob$ act on $\overline{A}$ and $\overline{B}$, respectively (in addition to inputs $A_f$ and $B_f$).
However, a cheating $\alice^*$ may use all the registers that are available to her.
Thus, when we consider $\Pg$($\alice^*, \bob$), the overall registers that can be accessed during the execution of subroutine $\Pf$ are all those in $\alice^*$ possession and $B_f\overline{B}$ (since $\bob$ is honest he won't access other registers).

\noindent Let the initial starting state (Table \ref{prot:pgpf}) be $\eta^{A_g\widehat{A}B_g\widehat{B}}$. To show that the honest protocol $\Pg$ (Table \ref{prot:pgpf}) is augmented-secure, we need to exhibit a simulator for $\Pg$($\alice^*, \bob$) (Table \ref{prot:pgpf_chaeting_alice}) as per Definition~\ref{def:augmented-secure}.
We will do this as follows:
\begin{enumerate}
    \item Given a protocol $\Pg$($\alice^*, \bob$) (Table \ref{prot:pgpf_chaeting_alice}), we will do the following steps to construct a (cheating) protocol $\Pgf$($\alice^*, \bob$) (Table \ref{prot:pgf_cheating_alice}):   
    \begin{enumerate}

        \item Note that honest $\Pg$ uses honest $\Pf$  as a subroutine.
        In the (cheating) $\alice^*$ case, we will replace this part of the protocol with a simulator $\textsc{Sim}^f$ that gives augmented security for $\Pf$($\alice^*, \bob$).
  \item The constructed protocol can now be viewed as an oracle-aided $\alice^*$ protocol for ${g\vert f}$.   \end{enumerate}   
With this constructed protocol, we will show that end states of $\Pg$($\alice^*, \bob$) and  $\Pgf$($\alice^*, \bob$) are close on relevant registers (see eq.~\eqref{eq:composability_closenss_two}).
    \item Since $\Pgf$($\alice^*, \bob$) is augmented-secure, there exists a simulator $\textsc{Sim}^{g \vert f}$ for  $\Pgf$($\alice^*, \bob$).
    Now, by the previous step, we know that the end states for $\Pg$($\alice^*, \bob$) and  $\Pgf$($\alice^*, \bob$) are close. 
    Hence we get a simulator for $\Pg$($\alice^*, \bob$) as well (which is a composition of $\textsc{Sim}^{g \vert f}$ and $\textsc{Sim}^{f}$).
   
\end{enumerate}
    \begin{table}[H]
        \centering \resizebox{12cm}{6cm}{
            \begin{tabular}{  l c r }
            \hline \\
            $\alice^*$ (input $A_g$) &  & $\bob$ (input $B_g$) \\
            \hline \\ 
             $A_g\widehat{A}$ & $\eta$ & $B_g\widehat{B}$\\
 \hdashline \\

& \large{\textsc{Protocol Phase 1}} & \\ \\

& \large{\textsc{(Denoted $\priorprt$)}} &  \\  $A_gA_f\widehat{A}\overline{A
}$ & $\theta$ &  $B_gB_f\widehat{B}\overline{B}$\\
\hdashline  \\
 & \large{Run $\Pf$ with input $A_f, B_f$} & \\  \\ 
  Obtain $O_{A_{f}}$ & &   Obtain $O_{B_{f}}$  \\  \\ 
   $A_gA_fO_{A_{f}}\widehat{A}\overline{A}$ & $\tau_{\Pf}$ &  $B_gB_fO_{B_{f}}\widehat{B}\overline{B}$\\
 \hdashline \\  
& \large{\textsc{Protocol Phase 2}} & \\ \\ 
& \large{\textsc{(Denoted $\afterprt$)}} & \\
 Output $O_{A_g}$  && Output $O_{B_g}$
  \\ \\ 
  $A_gA_fO_{A_{f}}O_{A_g}\widehat{A}\overline{A}$ & $\zeta_{\Pf}$ &  $B_gB_fO_{B_{f}}O_{B_g}\widehat{B}\widebar{B}\overline{B}$\\
 \hdashline \\  
 
    \end{tabular} }
        \caption{Protocol  $\Pg$($\alice^*, \bob$)}
        \label{prot:pgpf_chaeting_alice} 
    \end{table}
   
First, consider the protocol $\Pf$($\alice, \bob$). 
Protocol $\Pg$ uses $\Pf$($\alice, \bob$) as a subroutine. 
Note that $\alice^*$, being a cheating party, may use any information generated before this subroutine is called. 
Moreover, cheating $\alice^*$ may act on the whole joint state present at that time that it has access to. 
Considering this, a protocol   $\Pg$($\alice^*, \bob$) will be as described in  Table~\ref{prot:pgpf_chaeting_alice}. Note that when both $\alice$ and $\bob$ are honest (Table \ref{prot:pgpf}), in the state $\theta$, $\widebar{A}\widebar{B}$ are independent of the other registers, but this is not true when $\alice^*$ is arbitrarily deviating from the honest behaviour (Table \ref{prot:pgpf_chaeting_alice}).

Since $P_f$($\alice^*, \bob$) is augmented-secure, there exists a corresponding simulator $\textsc{Sim}^f$.  
Let $\mathcal{P}_f \left( \textsc{Sim}^f, \bob \right)$ be corresponding to the protocol defined as per Table~\ref{prot:Sim_sec_alice}. We construct a protocol  $P_{g \vert f}$($\alice^*, \bob$) in which $\alice^*$ uses $\Sim^f$ as a subroutine (as given in Table~\ref{prot:pgf_cheating_alice}). 
    \begin{table}[!h]
        \centering
            \begin{tabular}{  l c r }
            \hline \\
            $\alice^*$ (input $A_g$) &  & $\bob$ (input $B_g$) \\
            \hline \\   $A_g\widehat{A}$ & $\eta$ & $B_g\widehat{B}$\\
 \hdashline \\

& \large{Run \textsc{$\priorprt$}} &  \\ \\  $A_gA_f\widehat{A}\overline{A}$ & $\theta$ &  $B_gB_f\widehat{B}\widebar{B}$\\
\hdashline \\
 & $\mathcal{P}_f \left( \textsc{Sim}^f, \bob \right)$  with input $A_f,B_f$& \\
Obtain $O_{A_{f}}$ &  & Obtain $O_{B_{f}}$ \\\\
 $A_gA_fO_{A_{f}}\widehat{A}\overline{A}$ & $\tau_{\textsc{Sim}^f }$ &  $B_gB_fO_{B_{f}}\widehat{B}\widebar{B}$\\
 \hdashline \\

%& \large{\textsc{Protocol Phase 2}} & \\ \\ 
& \large{Run \textsc{$\afterprt$}} & \\
 Output $O_{A_g}$  && Output $O_{B_g}$
 \\ \\
 $A_gA_fO_{A_{f}}O_{A_g}\widehat{A}\overline{A} $ & $\zeta_{\textsc{Sim}^f }$ &  $B_gB_fO_{B_{f}}O_{B_g}\widehat{B}\widebar{B}$\\
 
 \hdashline \\  
    \end{tabular} 
        \caption{Protocol $\Pgf$($\alice^*, \bob$)}
        \label{prot:pgf_cheating_alice} 
    \end{table}
\noindent      Since $\textsc{Sim}^f$ is a simulator (for augmented security) for $\Pf$,
we have
\begin{equation} \label{eq:p_f_closeness}
   {\onenorm{{\tau}_{\Pf}^{A_gA_f\widehat{A}\overline{A}B_gB_f\widehat{B}O_{{A_{f}}}O_{{B_{f}}}}}{{\tau}_{\textsc{Sim}^f}^{A_gA_f\widehat{A}{\overline{A}}B_gB_f\widehat{B}O_{{A_{f}}}O_{{B_{f}}}}} \leq {\negl(\lambda)}.}
\end{equation}
Here, $B_g$ and $\widehat{B}$ are part of the additional register $R$ in the definition of augmented-security (Definition~\ref{def:augmented-secure}).
{Note that  these registers are not touched during $\mathcal{P}_f$ (since $\bob$ is honest), and hence, they can be considered to be a  part of $R$.}

Also, note that eq.~\eqref{eq:p_f_closeness} contains all the registers that are used in $\afterprt$ since $\bob$ is honest. 
Hence, by Fact \ref{fact:data},
\begin{equation} \label{eq:composability_closness_one}
\onenorm{\zeta _{\Pf}^{A_gA_f\widehat{A}\overline{A}B_gB_f\widehat{B}O_{{A_{f}}}O_{{B_{f}}}O_{A_g}O_{B_g}}}{\zeta^{A_gA_f\widehat{A}\overline{A}B_gB_f\widehat{B}O_{{A_{f}}}O_{{B_{f}}}O_{A_g}O_{B_g}}_{\textsc{Sim}^f}}  \leq {\negl(\lambda)}.  
    \end{equation}
    Now viewing the protocol in Table~\ref{prot:pgf_cheating_alice} as a cheating $\alice^*$ protocol for $g \vert f$ and using the fact that $\Pgf$ is augmented-secure, there exists a simulator $\textsc{Sim}^{g|f}$  
    such that 

\begin{equation} \label{eq:composability_closenss_two}
    \onenorm{\zeta_{\textsc{Sim}^{g|f}}^{A_gA_f\widehat{A}\overline{A}B_gO_{{A_{f}}}O_{A_g}O_{B_g}}}{\zeta_{\textsc{Sim}^f}^{A_gA_f\widehat{A}\overline{A}B_gO_{{A_{f}}}O_{A_g}O_{B_g}}} \leq {\negl(\lambda)}.
\end{equation} 
Now from eqs.~\eqref{eq:composability_closness_one} and \eqref{eq:composability_closenss_two}, using the triangle inequality, we get,
\[  \onenorm{\zeta_{\textsc{Sim}^{g|f}}^{A_gA_f\widehat{A}\overline{A}B_gO_{{A_{f}}}O_{A_g}O_{B_g}}}{\zeta_{\mathcal{P}_f}^{A_gA_f\widehat{A}\overline{A}B_gO_{{A_{f}}}O_{A_g}O_{B_g}}} \leq {\negl(\lambda)}. \]
Thus, $\textsc{Sim}^{g|f}$, when combined with $\textsc{Sim}^{f}$, acts as a simulator for protocol for $g$ in Table~\ref{prot:pgpf_chaeting_alice}. This completes the proof that $\Pg$ is augmented-secure when it makes a single call to $\Pf$. 
 \end{proof}

\section*{Proof of Claim \ref{claim:L_2_to_L_1}}

\begin{proof}

  First we show that if $\Vert \ket{u} \Vert_2 = \Vert \ket{u} \Vert_2 =1$, then $\onenorm{\ketbra{u}{u}}{\ketbra{v}{v}} \leq 2\eps $.

Suppose $\Vert \ket{u} \Vert_2 = \Vert \ket{u} \Vert_2 =1$. Using Fact \ref{fact:fuchs}, we have \begin{align}
    \frac{\onenorm{\ketbra{u}{u}}{\ketbra{v}{v}}^2}{8} \leq 1-F(\ketbra{u}{u},\ketbra{v}{v}) \label{eq:fvd}
\end{align}
Then,
\begin{align*}
    \eps^2 & \geq \Vert \ket{u} - \ket{v} \Vert_2^2 & (\mbox{given}) \\
    & = \Vert \ket{u} \Vert_2^2 + \Vert \ket{v} \Vert_2^2 - 2 \mathrm{Real}( \langle u,v \rangle) \\
    & \geq 2 \left( 1- \vert  \langle u,v \rangle \vert \right) & (\mbox{since $\Vert \ket{u} \Vert_2 = \Vert \ket{v} \Vert_2 =1$}) \\
    & = 2\left(1- F\left( \ketbra{u}{u}, \ketbra{v}{v}\right)\right) \\
    & \geq \frac{1}{4} \onenorm{\ketbra{u}{u}}{\ketbra{v}{v}}^2 &(\mbox{eq. \eqref{eq:fvd}}).
\end{align*}
This completes the proof that $\onenorm{\ketbra{u}{u}}{\ketbra{v}{v}}  \leq 2\eps $.

\noindent Now we move on to $\Vert \ket{u} \Vert_2 \neq 1$ or $\Vert \ket{v} \Vert_2 \neq 1$.

\noindent Define $\ket {u^\prime} = \frac {\ket{u}}{\Vert \ket{u} \Vert_2 }$ and $\ket {v^\prime} = \frac {\ket{v}}{\Vert \ket{v} \Vert_2 }$.
First, lets analyse $\Vert \ket{u^\prime} - \ket{v^\prime} \Vert_2$.

\begin{align*}
    \Vert \ket{u^\prime} - \ket{v^\prime} \Vert_2 & = \left\Vert  \frac {\ket{u}}{\Vert \ket{{u}} \Vert_2} -  \frac {\ket{v}}{\Vert \ket{v} \Vert_2}  \right\Vert_2 \\
    & \leq \left\Vert  \frac {\ket{u}}{\Vert \ket{u} \Vert_2} -  \frac {\ket{v}}{\Vert \ket{u} \Vert_2}\right\Vert_2 + \left\Vert \frac {\ket{v}}{\Vert \ket{v} \Vert_2} -\frac {\ket{v}}{\Vert \ket{u} \Vert_2}  \right\Vert_2 & (\mbox{By triangle inequality}) \\
    & \leq \frac{\Vert \ket{u} - \ket{v} \Vert_2}{\Vert \ket{u} \Vert_2} + \Vert \ket{v} \Vert_2 \left\Vert \frac{1}{\Vert \ket{v} \Vert_2} - \frac{1}{\Vert \ket{u} \Vert_2}\right\Vert \\
    & \leq \frac{\Vert \ket{u} - \ket{v} \Vert_2}{\Vert \ket{u} \Vert_2} + \frac{\Big\vert \Vert \ket{u} \Vert_2 - \Vert \ket{v} \Vert_2 \Big\vert }{\Vert \ket{u} \Vert_2} \\
    & \leq \frac{\eps}{c} + \frac{\eps}{c} & (\mbox{since $\Vert \ket{u} - \ket{v} \Vert_2 \leq \eps$ and $\Vert \ket{u} \Vert_2 \geq c$ })    \\
    & \leq \frac{2 \eps}{c}.
\end{align*}
Thus, since $\Vert \ket{u^\prime}\Vert_2 = \Vert \ket{v^\prime}\Vert_2 =1$, we get,
\[ \Vert \ket{u^\prime}\bra{u^\prime} - \ket{v^\prime}\bra{v^\prime} \Vert_1 \leq \frac{4 \eps}{c}.\]
This gives \begin{align*}
    & \frac{4 \eps}{c}\\ &\geq \Vert \ket{u^\prime}\bra{u^\prime} - \ket{v^\prime}\bra{v^\prime} \Vert_1
     \\&= \onenorm{\frac{\ketbra{u}{u}}{\| u\|_2^2}}{\frac{\ketbra{v}{v}}{\| v\|_2^2}}
      \\&\geq \onenorm{\frac{\ketbra{u}{u}}{\| u\|_2^2}}{\frac{\ketbra{v}{v}}{\| u\|_2^2}}-\onenorm{\frac{\ketbra{v}{v}}{\| u\|_2^2}}{\frac{\ketbra{v}{v}}{\| v\|_2^2}}& (\mbox{By triangle inequality}) 
       \\&\geq  \onenorm{{\ketbra{u}{u}}}{{\ketbra{v}{v}}} -\frac{\Big\vert \Vert \ket{u} \Vert_2 - \Vert \ket{v} \Vert_2 \Big\vert \cdot(\Vert  \ket{u} \Vert_2+\Vert  \ket{v} \Vert_2)}{\Vert  \ket{u} \Vert_2^2 \Vert  \ket{v} \Vert_2^2}& (\mbox{since $\Vert  \ket{u} \Vert_2 \leq 1$ and  $\Vert \ket{v} \bra{v} \Vert_1 \leq 1$})
       \\&\geq  \onenorm{{\ketbra{u}{u}}}{{\ketbra{v}{v}}} -2\cdot\frac{\Big\vert \Vert \ket{u} \Vert_2 - \Vert \ket{v} \Vert_2 \Big\vert }{\Vert  \ket{u} \Vert_2^2 \Vert  \ket{v} \Vert_2^2}
       \\&\geq  \onenorm{{\ketbra{u}{u}}}{{\ketbra{v}{v}}} -2\cdot\frac{\Vert  \ket{u}-  \ket{v} \Vert_2}{\Vert  \ket{u} \Vert_2^2 \Vert \ket{v} \Vert_2^2}. & (\mbox{By triangle inequality}) 
\end{align*}
This gives \begin{align*}
    \onenorm{{\ketbra{u}{u}}}{{\ketbra{v}{v}}} \leq  \frac{4 \eps}{c} + \frac{2 \eps }{c^2\cdot(c-\eps)^2}.
\end{align*}
Since $\eps<c/10$, we have \begin{align*}
   \onenorm{{\ketbra{u}{u}}}{{\ketbra{v}{v}}} \leq  O\left(\frac{\eps}{c^4}\right).
\end{align*}
\end{proof}

\section*{Proof of Claim \ref{claim:some_bit_entropy_OT}}
\begin{proof}
    Consider any state $\rho$.
    For $k,l \in \lbrace 0,1,+,- \rbrace$ let $\alpha_{kl} = \Tr\left( \rho \ketbra{kl}{kl}\right)$.
    By Fact \ref{fact:gentle_measurement}, at least one of the following two holds (for some constant $d>0$) :
    \begin{enumerate}
        \item for all $k,l \in \lbrace 0,1 \rbrace$
    we have $\alpha_{kl} \leq 1-d$. In particular, two of them are at least $\frac{d}{3}$.
    \item for all $k,l \in \lbrace +,- \rbrace$ we have
    $\alpha_{kl} \leq 1-d$.   In particular, two of them are at least $\frac{d}{3}$.  
    \end{enumerate}
   % Note that it suffices to show, for some $i \in \lbrace 0,1\rbrace$, we have,
%\[ \frac{d}{3} \leq \Pr\left( A^{0}({i})\right) \leq 1-\frac{2d}{3}.\] 
    Suppose the first condition holds. Then choose $x=0$.
    Now, $\Pr\left( A^{0}({0}) =0 \right) = \alpha_{00}+\alpha_{01}$ and  $\Pr\left( A^{0}({1})=0 \right) = \alpha_{00}+\alpha_{10}$. 

    \textsf{Case 1:} $\alpha_{00} \geq \frac{d}{3}$. 
    Then, $\Pr\left( A^{0}({0}) =0 \right) \geq \frac{d}{3}$ and $\Pr\left( A^{0}({1}) =0 \right) \geq \frac{d}{3}$. 
    Now at least one of $\alpha_{01}, \alpha_{10}$ or $\alpha_{11}$ is at least $\frac{d}{3}$.  If $\alpha_{01} \geq \frac{d}{3}$, then $\Pr\left( A^{0}(1)=0 \right) \leq 1 - \frac{d}{3}$.
    If $\alpha_{10} \geq \frac{d}{3}$, then $\Pr\left( A^{0}(0)=0 \right) \leq 1 - \frac{d}{3}$. 
    If   $\alpha_{11} \geq \frac{d}{3}$, then $\Pr\left( A^{0}(0)=0 \right) \leq 1 - \frac{d}{3}$ and $\Pr\left( A^{0}(1)=0 \right) \leq 1 - \frac{d}{3}$. 
    Thus in all the cases, for some $i \in \lbrace 0,1\rbrace$, we have,
\[ \frac{d}{3} \leq \Pr\left( A^{0}(i)=0\right) \leq 1-\frac{d}{3}.\]

\textsf{Case 2:} $\alpha_{00} \leq \frac{d}{3}$.
Recall that $\alpha_{kl} \leq 1-d$ for all $k,l \in \lbrace 0,1\rbrace$. 
Hence, for both $i$, we have 
$\Pr\left( A^{0}(i)=0\right) \leq 1-\frac{2d}{3}.$
   Moreover, $\alpha_{01} \geq \frac{d}{3}$ or $\alpha_{10} \geq \frac{d}{3}$.
   Hence,
     in all the cases, for some $i \in \lbrace 0,1\rbrace$, we have
    
\[ \frac{d}{3} \leq \Pr\left( A^{0}(i)=0\right) \leq 1-\frac{2d}{3}.\] 
This completes the proof when Condition 1 holds. 
The proof for Condition 2, follows similarly after choosing $x=1$.     
\end{proof}
\section*{Proof of Claim \ref{claim:value_G_1}}
\begin{proof} Consider the following two cases:

   \textbf{Case 1:} Suppose in $G_1$, $\Pr(\MS(X,\wht{Y},A,B)=1|C=0)< 1-\eps_r$.  Then \begin{align*}
       \vall(G_1)&\leq\delta[\delta_*\cdot1+(1-\delta_*)(1-\eps_r)]+(1-\delta)\cdot 1 
       \\&\leq\delta[\delta_*+1-\delta_*-\eps_r(1-\delta_*)]+(1-\delta)
        \\&\leq\delta[1-\eps_r(1-\delta_*)]+(1-\delta)
   \\&\leq 1-\delta\eps_r(1-\delta_*).\end{align*}

\textbf{Case 2:} Suppose $\Pr(\MS(X,\wht{Y},A,B)=1|C=0)\geq  1-\eps_r$.  
From Fact~\ref{fact:rigidity}, for all $a,x$, we have 
\begin{align}
        &\norm{\left(V_A\otimes V_B\right)\ket{\rho}-\ket{\Psi^{\MS}} \otimes \junkstate}_2 \leq O(\eps_r^{1/4}), \label{eq:rigidity_state2}\\
                &\norm{\left(V_A\otimes V_B\right)\left(M_{a|x}\otimes\Id\right)\ket{\rho}-\left(M^{\MS}_{a|x}\otimes \Id \right)\ket{\Psi^{\MS}}\otimes {\junkstate}}_2 \leq O(\eps_r^{1/4}). \label{eq:rigidity_Alice2}
    \end{align}
    Let $\left(V_A\otimes V_B\right)\left(M_{a|x}\otimes\Id\right)\left(V_A^\dagger\otimes V_B^\dagger\right)=\left(M^\prime_{a|x}\otimes\Id\right)$ and $\ket{\rho^\prime}=\left(V_A\otimes V_B\right)\ket{\rho}$.
We have
\begin{align}
\norm{\left(M^\prime_{a|x}\otimes\Id\right)\ket{\rho^\prime}}_2&=\norm{\left(V_A\otimes V_B\right)\left(M_{a|x}\otimes\Id\right)\ket{\rho}}_2 \nonumber\\&\geq \norm{\left(M^{\MS}_{a|x}\otimes \Id \right)\ket{\Psi^{\MS}}\otimes {\junkstate}}_2 -  O(\eps_r^{1/4}) &\mbox{(eq. \eqref{eq:rigidity_Alice2})}
\nonumber\\&\geq \frac{1}{2} -O(\eps_r^{1/4}) \geq \frac{1}{3}.  &\mbox{{(Claim \ref{claim:length})}} \label{eq:norm_lower}
\end{align}
 From eq. \eqref{eq:rigidity_state2}, we get
 \begin{align}    
 \norm{\left(M^{\MS}_{a|x}\otimes \Id \right)\left(V_A\otimes V_B\right)\ket{\rho}-\left(M^{\MS}_{a|x}\otimes \Id \right)\ket{\Psi^{\MS}} \otimes \junkstate}_2 \leq O(\eps_r^{1/4}). \label{eq:before_triangle}\end{align}
 From eqs. \eqref{eq:rigidity_Alice2} and  \eqref{eq:before_triangle}, using triangle inequality and $\ket{\rho^\prime}=\left(V_A\otimes V_B\right)\ket{\rho}$, we get
    \[\norm{\left(M^\prime_{a|x}\otimes\Id\right)\ket{\rho^\prime}-\left(M^{\MS}_{a|x}\otimes \Id \right)\ket{\rho^\prime}}_2 \leq O(\eps_r^{1/4}).
    \]
  Using Claim \ref{claim:L_2_to_L_1} and eq. \eqref{eq:norm_lower}, we get 
 \[\norm{{\left(M^\prime_{a|x}\otimes\Id\right)\ketbra{\rho^\prime}{\rho^\prime}\left(M^\prime_{a|x}\otimes\Id\right)}-{\left(M^\MS_{a|x}\otimes\Id\right)\ketbra{\rho^\prime}{\rho^\prime}\left(M^\MS_{a|x}\otimes\Id\right)}}_1 \leq O(\eps_r^{1/4}).
    \]
Consider (cheating) Bob's measurement on his part of the state $\rho$ to be {$\{\wtt{E}_b\}_b$ }. Let $E_b=V_B\wtt{E}_bV_B^{\dagger}$. Since \textbf{DELAY} has already happened, using Fact \ref{fact:data-PSD}, we get
{\small\[\norm{\sum_b  \ketbra{b}{b} \otimes \left({\left(M^\prime_{a|x}\otimes E_b\right)\ketbra{\rho^\prime}{\rho^\prime}\left(M^\prime_{a|x}\otimes E_b^{\dagger}\right)}-{\left(M^\MS_{a|x}\otimes E_b\right)\ketbra{\rho^\prime}{\rho^\prime}\left(M^\MS_{a|x}\otimes E_b^{\dagger}\right)}\right)}_1 \leq O(\eps_r^{1/4}).
    \]}
    Define \[\Pr(b)=\Tr\left(\Id\otimes E_b\right)\ketbra{\rho^\prime}{\rho^\prime}\left(\Id\otimes E_b^{\dagger}\right).\]
    We rewrite the above as
      {\small\[\norm{ \sum_b \Pr(b)\otimes \ketbra{b}{b}\otimes \left(\frac{\left(M^\prime_{a|x}\otimes E_b\right)\ketbra{\rho^\prime}{\rho^\prime}\left(M^\prime_{a|x}\otimes E_b^{\dagger}\right)}{\Tr(\left(\Id\otimes E_b\right)\ketbra{\rho^\prime}{\rho^\prime}\left(\Id\otimes E_b^{\dagger}\right))}-\frac{\left(M^\MS_{a|x}\otimes E_b\right)\ketbra{\rho^\prime}{\rho^\prime}\left(M^\MS_{a|x}\otimes E_b^{\dagger}\right)}{\Tr(\left(\Id\otimes E_b\right)\ketbra{\rho^\prime}{\rho^\prime}\left(\Id\otimes E_b^{\dagger}\right))}\right)}_1 \leq O(\eps_r^{1/4}).
    \]}This gives 
    {\small\[\sum_b \Pr(b)\norm{ \frac{\left(M^\prime_{a|x}\otimes E_b\right)\ketbra{\rho^\prime}{\rho^\prime}\left(M^\prime_{a|x}\otimes E_b^{\dagger}\right)}{\Tr(\left(\Id\otimes E_b\right)\ketbra{\rho^\prime}{\rho^\prime}\left(\Id\otimes E_b^{\dagger}\right))}-\frac{\left(M^\MS_{a|x}\otimes E_b\right)\ketbra{\rho^\prime}{\rho^\prime}\left(M^\MS_{a|x}\otimes E_b^{\dagger}\right)}{\Tr(\left(\Id\otimes E_b\right)\ketbra{\rho^\prime}{\rho^\prime}\left(\Id\otimes E_b^{\dagger}\right))}}_1 \leq O(\eps_r^{1/4})=\eps^\prime.
    \]}
    Define {\small\[\good(a,x)=\Bigg\{b~:~\norm{ \frac{\left(M^\prime_{a|x}\otimes E_b\right)\ketbra{\rho^\prime}{\rho^\prime}\left(M^\prime_{a|x}\otimes E_b^{\dagger}\right)}{\Tr(\left(\Id\otimes E_b\right)\ketbra{\rho^\prime}{\rho^\prime}\left(\Id\otimes E_b^{\dagger}\right))}-\frac{\left(M^\MS_{a|x}\otimes E_b\right)\ketbra{\rho^\prime}{\rho^\prime}\left(M^\MS_{a|x}\otimes E_b^{\dagger}\right)}{\Tr(\left(\Id\otimes E_b\right)\ketbra{\rho^\prime}{\rho^\prime}\left(\Id\otimes E_b^{\dagger}\right))}}_1\leq \sqrt{\eps^\prime}\Bigg\}.\]}
By Markov-inequality, $\Pr(b\in \good(a,x))\geq 1-\sqrt{\eps^\prime}$. Now, by union bound over all $a,x$, we get that there exists a set $\good=\bigcap_{a,x}\good(a,x)$ such that $\Pr(b \in \good) \geq 1-O(\sqrt{\eps^\prime})$.

Using $\left(V_A\otimes V_B\right)\left(M_{a|x}\otimes\Id\right)\left(V_A^\dagger\otimes V_B^\dagger\right)=\left(M^\prime_{a|x}\otimes\Id\right)$ and $\ket{\rho^\prime}=\left(V_A\otimes V_B\right)\ket{\rho}$, we get\begin{align*}
    \Tr \left(\frac{\left(M^\prime_{a|x}\otimes E_b\right)\ketbra{\rho^\prime}{\rho^\prime}\left(M^\prime_{a|x}\otimes E_b^{\dagger}\right)}{\Tr(\left(\Id\otimes E_b\right)\ketbra{\rho^\prime}{\rho^\prime}\left(\Id\otimes E_b^{\dagger}\right))}\right)=\Tr \left(\frac{\left(M_{a|x}\otimes \wtt{E}_b\right)\ketbra{\rho}{\rho}\left(M_{a|x}\otimes \wtt{E}_b^{\dagger}\right)}{\Tr(\left(\Id\otimes \wtt{E}_b\right)\ketbra{\rho}{\rho}\left(\Id\otimes \wtt{E}_b^{\dagger}\right))}\right).
\end{align*}
Combining this with Claim \ref{claim:some_bit_entropy_OT} with $\rho=\frac{\left(\Id\otimes \wtt{E}_b\right)\ketbra{\rho}{\rho}\left(\Id\otimes \wtt{E}_b^{\dagger}\right)}{\Tr\left(\Id\otimes \wtt{E}_b\right)\ketbra{\rho}{\rho}\left(\Id\otimes \wtt{E}_b^{\dagger}\right)}$ and choosing a small enough $\eps_r>0$ (and noting that since $\alice$ is honest, $X\leftarrow_{U}\{0,1,2\}$), we get  \[
p_{guess}(A[0,1]) \leq 
\Pr(b\in \good)\cdot\left(1-\frac{1}{3}\cdot \frac{2d}{3}+O(\sqrt{\eps'})\right)+\Pr(b\not\in \good) \leq 1-d/10.
\]
We thus have
\begin{align*}
\vall(G_1) & \leq \delta\cdot 1 + (1-\delta) \cdot \left[\delta_*\cdot1+(1-\delta_*)(1-d/10) \right]
\\& \leq \delta + (1-\delta) \cdot \left[\delta_*+1-\delta_*-(1-\delta_*)d/10 \right]
\\& \leq \delta + (1-\delta) \cdot \left[1-(1-\delta_*)d/10 \right]
\\& \leq \delta + 1-\delta - (1-\delta)(1-\delta_*)d/10 
\\& \leq 1 - (1-\delta)(1-\delta_*)d/10 .
\end{align*}
Let $\eps_*=\min(\delta\eps_r(1-\delta_*), (1-\delta)(1-\delta_*)d/10)$. Then \[
\vall(G_1)\leq 1-\eps_*.
\]
 
\end{proof}

\section*{Proof of Claim \ref{claim:all_markovs}}
\begin{proof}
 We note that conditioning on $\{D_jU_j\}_{j\in [n]}$, the input distribution for each coordinate becomes a product distribution in the state $\theta$. 
Recall that $\mcC= \{i_i, \dots i_r\}$. We note that \begin{equation} \label{eq:relate_q}
    q=\Pr\left(T^{(r)}=1\right) > 2^{-\delta_2 n} .
\end{equation}
We want that for a typical coordinate outside $\mcC$, the distribution of questions in state $\varphi$ is close to $\mu$. {\small\begin{align}
       &\delta_3n \nonumber\\ &\geq \delta_2n &\mbox{(eq. \eqref{eq:deltas_relation})} \nonumber \\&\geq -\log q &\mbox{(eq. \eqref{eq:relate_q})} \nonumber
       \\&\geq \D_\infty(\varphi^{XY}\|\theta^{XY}) &\mbox{(eqs. \eqref{eq:def_theta} and \eqref{eq:def_varphi})}\nonumber\\
       & \geq \D(\varphi^{XY}\|\theta^{XY})  &\mbox{(Fact \ref{fact:monotonicity})}\nonumber\\
       &= \sum_{i\in[n]}\E_{{x_{<i}y_{<i}}\leftarrow \varphi^{X_{<i}Y_{<i}}}\D(\varphi_{x_{<i}y_{<i}}^{{X}_i{Y}_i}\|\theta^{X_iY_i}) &\mbox{(Fact \ref{fact:chain_rule_D}, since $\{\theta^{X_iY_i}\}_i$ are independent)}\nonumber\\
       &\geq \sum_{i\in[n]}\D(\E_{{x_{<i}y_{<i}}\leftarrow \varphi^{X_{<i}Y_{<i}}}\varphi_{x_{<i}y_{<i}}^{{X}_i{Y}_i}\|\E_{{x_{<i}y_{<i}}\leftarrow \varphi^{X_{<i}Y_{<i}}} \theta^{X_iY_i}) &\mbox{(Fact \ref{fact:joint_convexity})}\nonumber\\ 
       &=\sum_{i\in[n]}\D(\varphi^{{X}_i{Y}_i}\|\theta^{X_iY_i}) 
       \nonumber\\&\geq \sum_{i\notin \mcC}\D(\varphi^{{X}_i{Y}_i}\|\theta^{X_iY_i}). &\mbox{(Fact \ref{fact:non_negative})} \label{eq:markov_1}
\end{align}}
Consider
\begin{align}
   & \delta_3n \nonumber\\ &\geq \delta_2n +|\mcC|\log(|\clA|\cdot|\clB|)  &\mbox{(eq. \eqref{eq:deltas_relation})} \nonumber
    \\&\geq \E_{x_{\mcC}y_{\mcC}a_{\mcC}b_{\mcC}du\leftarrow \varphi^{X_{\mcC}Y_{\mcC}A_{\mcC}B_{\mcC}DU}}[\D(\varphi_{x_{\mcC}y_{\mcC}a_{\mcC}b_{\mcC}du}^{\wtt{X}_{\widebar{\mcC}}\wtt{Y}_{\widebar{\mcC}}XYE_AE_B}\|\theta_{x_{\mcC}y_{\mcC}du}^{\wtt{X}_{\widebar{\mcC}}\wtt{Y}_{\widebar{\mcC}}XYE_AE_B})]  &\mbox{(Lemma \ref{lemma:main})} \nonumber
    \\&\geq  \E_{x_{\mcC}y_{\mcC}a_{\mcC}b_{\mcC}du\leftarrow \varphi^{X_{\mcC}Y_{\mcC}A_{\mcC}B_{\mcC}DU}}\left[\D(\varphi_{x_{\mcC}y_{\mcC}a_{\mcC}b_{\mcC}du}^{XY}\|\theta_{x_{\mcC}y_{\mcC}du}^{XY}) \right] &\mbox{(Fact \ref{fact:data})} \nonumber
   \\& = \sum_{i\notin \mcC} \E_{x_{\mcC}y_{\mcC}a_{\mcC}b_{\mcC}dux_{<i}y_{<i}\leftarrow \varphi^{X_{\mcC}Y_{\mcC}A_{\mcC}B_{\mcC}DUX_{<i}Y_{<i}}}\D(\varphi_{x_{\mcC}y_{\mcC}a_{\mcC}b_{\mcC}dux_{<i}y_{<i}}^{X_iY_i}\|\theta_{d_iu_i}^{X_iY_i}) &\mbox{(Fact \ref{fact:chain_rule_D})} \nonumber
   \\& = \sum_{i\notin \mcC} \E_{r_id_iu_i\leftarrow \varphi^{R_iD_iU_i}}\D(\varphi_{r_id_iu_i}^{X_iY_i}\|\theta_{d_iu_i}^{X_iY_i}) &\mbox{(eq. \eqref{eq:define_R_j})} \label{eq:markov_5}
   \\& = \frac{1}{2}\sum_{i\notin \mcC} \left(\E_{r_i\hat{x}_i\leftarrow \varphi^{R_i\hat{X}_i}}\D(\varphi_{r_i\hat{x}_i}^{X_i\wht{Y}_iC_i}\|\theta_{\hat{x}_i}^{X_i\wht{Y}_iC_i}) + \E_{r_ix_i\leftarrow \varphi^{R_iX_i}}\D(\varphi_{r_ix_i}^{Y_i}\|\theta_{x_i}^{Y_i}) \right)  \nonumber
    \\& \geq \frac{1}{2}\sum_{i\notin \mcC} \left(\E_{r_iy_i\leftarrow \varphi^{R_iY_i}}\D(\varphi_{r_iy_i}^{X_i}\|\theta_{y_i}^{X_i}) + \E_{r_ix_i\leftarrow \varphi^{R_iX_i}}\D(\varphi_{r_ix_i}^{Y_i}\|\theta_{x_i}^{Y_i}) \right)  &\mbox{(Fact \ref{fact:chain_rule_D})}\nonumber
     \\& = \frac{1}{2}\sum_{i\notin \mcC}\left( \E_{y_i\leftarrow \varphi^{Y_i}}\D(\varphi_{y_i}^{R_iX_i}\|\varphi_{y_i}^{R_i}\otimes\theta_{y_i}^{X_i}) + \E_{x_i\leftarrow \varphi^{X_i}}\D(\varphi_{x_i}^{R_iY_i}\|\varphi_{x_i}^{R_i}\otimes\theta_{x_i}^{Y_i}) \right) &\mbox{(Fact \ref{fact:chain_rule_D})} \nonumber
     \\& \geq \frac{1}{2}\sum_{i\notin \mcC} \left(\E_{y_i\leftarrow \varphi^{Y_i}}\D(\varphi_{y_i}^{R_iX_i}\|\varphi_{y_i}^{R_i}\otimes\varphi_{y_i}^{X_i}) + \E_{x_i\leftarrow \varphi^{X_i}}\D(\varphi_{x_i}^{R_iY_i}\|\varphi_{x_i}^{R_i}\otimes\varphi_{x_i}^{Y_i}) \right) &\mbox{(Fact \ref{fact:minimisation_for_mutual_info})} \nonumber
     \\& = \frac{1}{2}\sum_{i\notin \mcC} \left(\I(X_i:R_i|Y_i)_{\varphi} + \I(Y_i:R_i|X_i)_{\varphi} \right). \label{eq:markov_2}
\end{align}
We also want that for a typical coordinate $j$ outside $\mcC$, the mutual information between $\charlie$'s question $X_j$ and $\dave$'s registers $\wtt{B}$ is small in $\ket{\varphi}$ conditioned on $R_jD_jU_j$. From Lemma \ref{lemma:main}, Fact \ref{fact:data} and eq. \eqref{eq:deltas_relation}, we get 
\begin{align}
    \delta_3 n &\geq  \E_{x_{\mcC}y_{\mcC}a_{\mcC}b_{\mcC}du\leftarrow \varphi^{X_{\mcC}Y_{\mcC}A_{\mcC}B_{\mcC}DU}}\left[\D(\varphi_{x_{\mcC}y_{\mcC}a_{\mcC}b_{\mcC}du}^{X\wtt{B}}\|\theta_{x_{\mcC}y_{\mcC}du}^{X\wtt{B}}) \right] \nonumber
    \\& \geq \I(X:\wtt{B}|X_{\mcC}Y_{\mcC}A_{\mcC}B_{\mcC}DU)_{\varphi}
      \nonumber\\& \geq \sum_{i \notin \mcC}\I(X_i:\wtt{B}|X_{\mcC}Y_{\mcC}A_{\mcC}B_{\mcC}DUX_{<i})_{\varphi} &\mbox{(Fact \ref{fact:chain_rule_mutual_info})}  \nonumber
     \\& \geq \sum_{i \notin \mcC}\I(X_i:\wtt{B}|R_iD_iU_i)_{\varphi}. \label{eq:markov_3}
\end{align}
Here, for the second inequality, we use that {$\theta^{X\wtt{B}}_{x_{\mcC}y_{\mcC}du}=\theta^{X}_{x_{\mcC}y_{\mcC}du}\otimes\theta^{\wtt{B}}_{x_{\mcC}y_{\mcC}du}$} along with Fact \ref{fact:minimisation_for_mutual_info}. For the last inequality, we use Fact \ref{fact:chain_rule_mutual_info} and the observation that $\wtt{B}$ contains $Y_{<i}$.

Similarly, we have \begin{equation}
\delta_3n\geq \sum_{i \notin \mcC}\I(Y_i:\wtt{A}|R_iD_iU_i)_{\varphi}.
\label{eq:markov_4}\end{equation}
From Lemma \ref{lemma:main} and Fact \ref{fact:data}, we have 
\begin{align}
    \delta_3 n &\geq  \E_{x_{\mcC}y_{\mcC}a_{\mcC}b_{\mcC}du\leftarrow \varphi^{X_{\mcC}Y_{\mcC}A_{\mcC}B_{\mcC}DU}}\left[\D(\varphi_{x_{\mcC}y_{\mcC}a_{\mcC}b_{\mcC}du}^{X\wht{X}\wht{B}}\|\theta_{x_{\mcC}y_{\mcC}du}^{X\wht{X}\wht{B}}) \right] \nonumber
    \\& \geq \I(X\wht{X}:\wht{B}|X_{\mcC}Y_{\mcC}A_{\mcC}B_{\mcC}DU)_{\varphi}
      \nonumber\\& \geq \sum_{i \notin \mcC}\I(X_i\wht{X}_i:\wht{B}|X_{\mcC}Y_{\mcC}A_{\mcC}B_{\mcC}DUX_{<i}\wht{X}_{<i})_{\varphi} &\mbox{(Fact \ref{fact:chain_rule_mutual_info})}  \nonumber
     \\& \geq \sum_{i \notin \mcC}\I(X_i\wht{X}_i:\wht{B}|R_iD_iU_i)_{\varphi} \nonumber
     \\& \geq \sum_{i \notin \mcC}\I(\wht{X}_i:\wht{B}|R_iD_iU_i)_{\varphi}. \label{eq:markov_for_classical_postprocessing}
\end{align}
Here, for the second inequality, we use that {$\theta^{X\wht{X}\wht{B}}_{x_{\mcC}y_{\mcC}du}=\theta^{X\wht{X}}_{x_{\mcC}y_{\mcC}du}\otimes\theta^{\wht{B}}_{x_{\mcC}y_{\mcC}du}$} along with Fact \ref{fact:minimisation_for_mutual_info}. For the second last inequality, we use Fact \ref{fact:chain_rule_mutual_info} and the observation that $\wht{B}$ contains $\wht{Y}_{<i}C_{<i}$.

    On applying Markov's inequality on eqs. \eqref{eq:markov_1}, \eqref{eq:markov_5}, \eqref{eq:markov_2}, \eqref{eq:markov_3}, \eqref{eq:markov_4} and \eqref{eq:markov_for_classical_postprocessing} and using union bound, we get that the following conditions  hold for a randomly selected coordinate $j$ outside $\mcC$ with probability at least $\frac{1}{4}$:
    \begin{equation*}
    \I(X_j:\wtt{B}|R_jD_jU_j)_{\varphi} \leq \frac{8\delta_3}{(1-\delta_1)}\leq 10\delta_3. \end{equation*}
     \begin{equation*}\I(Y_j:\wtt{A}|R_jD_jU_j)_{\varphi} \leq \frac{8\delta_3}{(1-\delta_1)}\leq 10\delta_3.\end{equation*}
      \begin{equation*}
\D(\varphi^{{X}_j{Y}_j}\|\theta^{X_jY_j})\leq \frac{8\delta_3}{(1-\delta_1)}\leq 10\delta_3.\end{equation*}
     \begin{equation*}\frac{1}{2}\left(\I(X_j:R_j|Y_j)_{\varphi} + \I(Y_j:R_j|X_j)_{\varphi} \right) \leq \frac{8\delta_3}{(1-\delta_1)}\leq 10\delta_3.
     \end{equation*}
     \begin{equation*}
 \E_{r_jd_ju_j\leftarrow \varphi^{R_jD_jU_j}}\D(\varphi_{r_jd_ju_j}^{X_jY_j}\|\theta_{d_ju_j}^{X_jY_j})  \leq \frac{8\delta_3}{(1-\delta_1)}\leq 10\delta_3.
\end{equation*}
\begin{equation*}
\I(\wht{X}_j:\wht{B}|R_jD_jU_j)_{\varphi} \leq \frac{8\delta_3}{(1-\delta_1)}\leq 10\delta_3.
\end{equation*}
\end{proof}
\section*{Proof of Claim \ref{claim:all_equations}}

\begin{proof}  Note that for a classical state, we will use the notation from eq. \eqref{eq:notation}.

\textbf{Existence of unitaries (using Fact \ref{fact:unitary_existence})} 

\noindent Conditioning on $D_j=2$ in eq. \eqref{eq:condition1}, we get \[
\I(X_j:\wtt{B}|R_j\wht{X}_j)_{\varphi} \leq 20\delta_3.
\]
Using eq. \eqref{eq:condition3} and Fact \ref{fact:pinsker}, we get \begin{align*}
    \Pr(\wht{X}_j=*)_{\varphi}\geq \delta_*-\sqrt{10\delta_3} \geq  0.99\delta_*. &\tab\mbox{(since $\delta_3=\delta_*^2/10^5$)}\end{align*} This gives \begin{equation*}
\I(X_j:\wtt{B}|R_j)_{\varphi_*} \leq 20\delta_3/0.99\delta_*\leq\delta_*.
\end{equation*} This proves eq. \eqref{eq:existence_unitary_1}.

Conditioning on $D_j=2$ in eq. \eqref{eq:condition2}, we get \begin{align*}
   20\delta_3 &\geq \I(\wht{Y}_jC_j:\wtt{A}|R_j\wht{X}_j)_{\varphi}.
\end{align*}
As before, conditioning on $\wht{X}_j=*$, we get 
\begin{equation*}
\I(\wht{Y}_jC_j:\wtt{A}|R_j)_{\varphi_*} \leq 20\delta_3/0.99\delta_*\leq\delta_*.
\end{equation*}
This proves eq. \eqref{eq:existence_unitary_2}.

In eq. \eqref{eq:intermediate_markov}, conditioning on $D_j=2$, we get\[
 \E_{r_j\wht{x}_j\leftarrow \varphi^{R_j\wht{X}_j}}\D(\varphi_{r_j\wht{x}_j}^{X_j\wht{Y}_jC_j}\|\theta_{\wht{x}_j}^{X_j\wht{Y}_jC_j}) \leq 20\delta_3.
\]
Conditioning on $\wht{X}_j=*$, we get\[
 \E_{r_j\leftarrow \varphi_*^{R_j}}\D(\varphi_{*,r_j}^{X_j\wht{Y}_jC_j}\|\theta_{*}^{X_j\wht{Y}_jC_j}) \leq 20\delta_3/0.99\delta_*\leq\delta_*.
\]
This gives \begin{align*}
    \delta_* & \geq \E_{r_j\leftarrow \varphi_*^{R_j}}\D(\varphi_{*,r_j}^{X_j\wht{Y}_jC_j}\|\theta_{*}^{X_j\wht{Y}_jC_j})\nonumber\\& =\E_{r_j\leftarrow \varphi_*^{R_j}}\D(\varphi_{*,r_j}^{X_j\wht{Y}_jC_j}\|\theta_{*}^{X_j}\otimes \theta_{*}^{\wht{Y}_jC_j}) &\mbox{(Since $\theta_{*}^{X_j\wht{Y}_jC_j}=\theta_{*}^{X_j}\otimes \theta_{*}^{\wht{Y}_jC_j}$)}\nonumber
    \\& \geq \E_{r_j\leftarrow \varphi_{*}^{R_j}}\D(\varphi_{*,r_j}^{X_j\wht{Y}_jC_j}\|\varphi_{*,r_j}^{X_j}\otimes \varphi_{*,r_j}^{\wht{Y}_jC_j}) &\mbox{(Fact  \ref{fact:minimisation_for_mutual_info})}\nonumber
    \\&\geq \E_{r_j\leftarrow \varphi_*^{R_j}}\onenorm{\varphi_{*,r_j}^{X_j\wht{Y}_jC_j}}{\varphi_{*,r_j}^{X_j}\otimes\varphi_{*,r_j}^{\wht{Y}_jC_j}}^2 &\mbox{(Fact \ref{fact:pinsker})}
     \\&\geq \left(\E_{r_j\leftarrow \varphi_*^{R_j}}\onenorm{\varphi_{*,r_j}^{X_j\wht{Y}_jC_j}}{\varphi_{*,r_j}^{X_j}\otimes\varphi_{*,r_j}^{\wht{Y}_jC_j}}\right)^2. &\mbox{(convexity)}
\end{align*}
 This gives \begin{equation*}
     \E_{r_j\leftarrow \varphi_*^{R_j}}\onenorm{\varphi_{*,r_j}^{X_j\wht{Y}_jC_j}}{\varphi_{*,r_j}^{X_j}\otimes\varphi_{*,r_j}^{\wht{Y}_jC_j}}\leq \sqrt{\delta_*}. \end{equation*}
     This proves eq. \eqref{eq:existence_unitary_3}.
     \vspace{4mm}\\
\noindent \textbf{Correlated sampling of $R_j$ (using Fact \ref{fact:hollenstein_sampling})}

From eq. \eqref{eq:condition4}, we get \[
\I(X_j:R_j|Y_j)_{\varphi}=\I(X_j:R_j|\wht{Y}_jC_j\wht{X}_j)_{\varphi} \leq 20\delta_3. 
\]
As before, conditioning on $\wht{X}_j=*$, we get \begin{equation*}
\I(X_j:R_j|\wht{Y}_jC_j)_{\varphi_*} \leq 20\delta_3/0.99\delta_*\leq\delta_*. 
\end{equation*}
Similarly, we have \begin{align*}
 20\delta_3 &\geq  \I(R_j:Y_j|X_j)_{\varphi} 
 \\&= \I(R_j:\wht{Y}_jC_j\wht{X}_j|X_j)_{\varphi} 
  \\&= \I(R_j:\wht{X}_j|X_j)_{\varphi} + \I(R_j:\wht{Y}_jC_j|X_j\wht{X}_j)_{\varphi}  &\mbox{(Fact \ref{fact:chain_rule_mutual_info})}
  \\&\geq  \I(R_j:\wht{Y}_jC_j|X_j\wht{X}_j)_{\varphi}. 
\end{align*}
As before, conditioning on $\wht{X}_j=*$, we get \begin{equation*}
    \I(R_j:\wht{Y}_jC_j|X_j)_{\varphi_*} \leq 20\delta_3/0.99\delta_*\leq\delta_*.
\end{equation*}
{This proves eqs. \eqref{eq:joint_sampling1} and  \eqref{eq:joint_sampling2}.}

\noindent\textbf{Correctly incorporating $\wht{X}_j$ after the \textbf{DELAY}}

{ We have $\wtt{B}=\widehat{B}\wht{X}B_{\widebar{\mcC}}N'$ where $\widehat{B}$ contains all registers of $\dave$ after the $\CNOT$ (other than $\wht{X}$ and the answer register $B_{\widebar{\mcC}}$). In the state $\ket{\theta}$, on conditioning on $R_j\wht{X}_j$, all the inputs between $\charlie$ and $\dave$ are independent. {After the $\CNOT$, we have the following Markov Chain in $\ket{\theta}$}:  
\[
X_jE_A-\widehat{B}\wht{X}_jR_j-B_{j}.
\]
To get $\ket{\varphi}$, we condition $\ket{\theta}$ on $V(X_\mcC,Y_\mcC,A_\mcC,B_\mcC)=1$ which could generate correlations (and hence the Markov Chain above would not hold in general) in $\ket{\varphi}$. Since $R_j$ contains  $X_\mcC,Y_\mcC,A_\mcC,B_\mcC$, conditioning on $V(X_\mcC,Y_\mcC,A_\mcC,B_\mcC)=1$ does not disturb the above Markov Chain. Thus, we get the same Markov Chain in state $\ket{\varphi_{}}$:
\begin{equation*}
X_jE_A-\widehat{B}\wht{X}_jR_j-B_j.
\end{equation*}
 Similarly, we also get the following Markov Chain in state $\ket{\varphi_{}}$ (and $\ket{\theta_{}}$):
\begin{equation*}
E_A-X_j\widehat{B}R_j-
\wht{X}_jB_j.
\end{equation*}
This proves eqs. \eqref{eq:markov_generate_B} and \eqref{eq:markov_in_phi}. }
Conditioning on $D_j=1$ in eq. \eqref{eq:last_markov}, we get
\begin{align}
 20\delta_3 &\geq 
\I(\wht{X}_j:\wht{B}|R_jX_j)_{\varphi} \nonumber
\\ &=\E_{r_jx_j\leftarrow \varphi^{R_jX_j}} \D(\varphi_{r_jx_j}^{\wht{X}_j\wht{B}}\|\varphi_{r_jx_j}^{\wht{X}_j}\otimes \varphi_{r_jx_j}^{\wht{B}}). \nonumber
\end{align}
This proves eq. \eqref{eq:tensor-lemma}.

\noindent From eq. \eqref{eq:condition3}, we have
\begin{align*}
 10\delta_3 &\geq \D(\varphi^{{X}_j{\wht{X}}_j}\|\theta^{X_j\wht{X}_j})\geq \E_{\wht{x}_j\leftarrow \varphi^{\wht{X}_j}} \D(\varphi_{\wht{x}_j}^{X_j}\|\theta_{\wht{x}_j}^{{X}_j}).
 \end{align*}
Conditioning on $\wht{X}_j=*$, we have \[
\D(\varphi_{*}^{X_j}\|\theta_{*}^{{X}_j}) \leq 10\delta_3/0.99\delta_*\leq\delta_*/4.
\]
On applying Fact \ref{fact:pinsker} and using $\theta_{*}^{{X}_j}=\theta^{{X}_j}$, we have\[
\onenorm{\varphi_*^{X_j}}{\theta_{}^{{X}_j}} \leq \sqrt{\delta_*}/2.
\]
From eq. \eqref{eq:condition3}, using Fact \ref{fact:pinsker} and Fact \ref{fact:data}, we have\[
\onenorm{\varphi^{X_j}}{\theta_{}^{{X}_j}} \leq \sqrt{10\delta_3}.
\]
This gives \begin{equation}
 \label{eq:distance_on_switch}    
\onenorm{\varphi_*^{X_j}}{\varphi^{{X}_j}} \leq \sqrt{10\delta_3} + \sqrt{\delta_*}/2 \leq \sqrt{\delta_*}.
\end{equation}
From eq. \eqref{eq:condition4}, we have \begin{align*} 
20\delta_3 &\geq \I(Y_j:R_j|X_j)_\varphi
 \geq \I(\wht{X}_j:R_j|X_j)_\varphi.
\end{align*}
This gives 
\begin{align} 
20\delta_3 &\geq \E_{x_j\leftarrow \varphi^{X_j}} \D(\varphi_{x_j}^{\wht{X}_jR_j}\|\varphi_{x_j}^{\wht{X}_j}\otimes \varphi_{x_j}^{R_j}) \label{eq:l1_corrections}\\
&\geq \E_{x_j\wht{x}_j\leftarrow \varphi^{X_j\wht{X}_j}} \D(\varphi_{x_j\wht{x}_j}^{R_j}\|\varphi_{x_j}^{R_j}). &\mbox{(Fact \ref{fact:chain_rule_D})} \nonumber
\end{align}
Conditioning on $\wht{X}_j=*$, we get \begin{align*}
\delta_*
\geq  20\delta_3/0.99\delta_* & \geq 
\E_{x_j\leftarrow \varphi_*^{X_j}} \D(\varphi_{*,x_j}^{R_j}\|\varphi_{x_j}^{R_j}) \\ & \geq 
\E_{x_j\leftarrow \varphi_*^{X_j}} \onenorm{\varphi_{*,x_j}^{R_j}}{\varphi_{x_j}^{R_j}}^2  &\mbox{(Fact \ref{fact:pinsker})}
\\ &\geq 
\left(\E_{x_j\leftarrow \varphi_*^{X_j}} \onenorm{\varphi_{*,x_j}^{R_j}}{\varphi_{x_j}^{R_j}}\right)^2. &\mbox{(convexity)} 
\end{align*}
Thus, we get
\begin{align}
\E_{x_j\leftarrow \varphi_*^{X_j}} \onenorm{\varphi_{*,x_j}^{R_j}}{\varphi_{x_j}^{R_j}} \leq\sqrt{\delta_*}. \label{eq:small_expectation}
\end{align}
Thus, we have \begin{align}    
&\onenorm{\varphi_{*}^{R_jX_j}}{\varphi^{R_jX_j}} 
\nonumber\\& =  \onenorm{\varphi_{*}^{X_j}\cdot \varphi_{*,X_j}^{R_j}}{\varphi^{X_j}\cdot\varphi^{R_j}_{X_j}}  \nonumber
\\& \leq   \onenorm{\varphi_{*}^{X_j}\cdot \varphi_{*,X_j}^{R_j}}{\varphi_*^{X_j}\cdot\varphi^{R_j}_{X_j}}  + \onenorm{\varphi_*^{X_j}\cdot\varphi^{R_j}_{X_j}}{\varphi^{X_j}\cdot\varphi^{R_j}_{X_j}} &\mbox{({triangle inequality})} \nonumber
\\ &\leq \E_{x_j\leftarrow \varphi_*^{X_j}} \onenorm{\varphi_{*,x_j}^{R_j}}{\varphi_{x_j}^{R_j}} +  \onenorm{\varphi_*^{X_j}}{\varphi^{{X}_j}} \nonumber
\\&\leq 2\sqrt{\delta_*}. &\mbox{(eqs. \eqref{eq:distance_on_switch} and \eqref{eq:small_expectation})}  \nonumber
\end{align}
This proves eqs. \eqref{eq:R_jX_j_star}.

\noindent Using eq. \eqref{eq:markov_in_phi}, we have \begin{equation}
\onenorm{ \varphi^{X_jR_jE_A\wht{B}\wht{X}_j}}{ \varphi^{X_jR_jE_A\wht{B}}\cdot \varphi^{\wht{X}_j}_{X_jR_j\wht{B}}}=0.
            \label{eq:triangle1}
\end{equation}
            From eq. \eqref{eq:tensor-lemma}, we have
\begin{align*}
     20\delta_3 &\geq \E_{r_jx_j\leftarrow \varphi^{R_jX_j}} \D(\varphi_{r_jx_j}^{\wht{X}_j\wht{B}}\|\varphi_{r_jx_j}^{\wht{X}_j}\otimes \varphi_{r_jx_j}^{\wht{B}})
     \\ & \geq \E_{r_jx_j\wht{b}\leftarrow \varphi^{R_jX_j\wht{B}}} \D(\varphi_{r_jx_j\wht{b}}^{\wht{X}_j}\|\varphi_{r_jx_j}^{\wht{X}_j}). &\mbox{(Fact \ref{fact:chain_rule_D})}
\end{align*}
Fact \ref{fact:pinsker} and convexity gives \begin{align}
 \sqrt{20\delta_3} & \geq \E_{r_jx_j\wht{b}\leftarrow \varphi^{R_jX_j\wht{B}}} \onenorm{\varphi_{r_jx_j\wht{b}}^{\wht{X}_j}}{\varphi_{r_jx_j}^{\wht{X}_j}} \nonumber\\&\geq\onenorm{ \varphi^{X_jR_j\wht{B}}\cdot \varphi^{\wht{X}_j}_{X_jR_j\wht{B}}}{ \varphi^{X_jR_j\wht{B}}\cdot \varphi^{\wht{X}_j}_{X_jR_j}} \nonumber\\
 &{\geq\onenorm{ \varphi^{X_jR_jE_A\wht{B}}\cdot \varphi^{\wht{X}_j}_{X_jR_j\wht{B}}}{ \varphi^{X_jR_jE_A\wht{B}}\cdot \varphi^{\wht{X}_j}_{X_jR_j}} }&\mbox{(eq. \eqref{eq:markov_in_phi} and Fact \ref{fact:recovery})}.
\label{eq:triangle2}
\end{align}
From eq. \eqref{eq:l1_corrections}, we have\begin{align*} 
20\delta_3 &\geq \E_{x_j\leftarrow \varphi^{X_j}} \D(\varphi_{x_j}^{\wht{X}_jR_j}\|\varphi_{x_j}^{\wht{X}_j}\otimes \varphi_{x_j}^{R_j}) \\
&\geq \E_{x_jr_j\leftarrow \varphi^{X_jR_j}} \D(\varphi_{x_jr_j}^{\wht{X}_j}\|\varphi_{x_j}^{\wht{X}}). &\mbox{(Fact \ref{fact:chain_rule_D})} 
\end{align*}
Again using Fact \ref{fact:pinsker},  convexity and eq. \eqref{eq:markov_in_phi} gives \begin{equation}
\onenorm{ \varphi^{X_jR_jE_A\wht{B}}\cdot \varphi^{\wht{X}_j}_{X_jR_j}}{ \varphi^{X_jR_jE_A\wht{B}}\cdot \varphi^{\wht{X}_j}_{X_j}} \leq \sqrt{20\delta_3}.
\label{eq:triangle3}
\end{equation}
From eq. \eqref{eq:condition3}, we have
\begin{align*}
 10\delta_3 &\geq \D(\varphi^{{X}_j{\wht{X}}_j}\|\theta^{X_j\wht{X}_j})\\
 &\geq \E_{x_j\leftarrow \varphi^{{X}_j}} \D(\varphi_{{x}_j}^{\wht{X}_j}\|\theta_{{x}_j}^{\wht{X}_j}).
 \end{align*}
Fact \ref{fact:pinsker},  convexity and eq. \eqref{eq:markov_in_phi} \begin{equation}
\onenorm{ \varphi^{X_jR_jE_A\wht{B}}\cdot \varphi^{\wht{X}_j}_{X_j}}{ \varphi^{X_jR_jE_A\wht{B}}\cdot \theta^{\wht{X}_j}_{X_j}} \leq \sqrt{10\delta_3}. \label{eq:triangle4}
\end{equation}
Since $\theta^{\wht{X}_j}_{X_j}=\theta^{\wht{X}_j}_{X_j\wht{Y}_jC_j}$, using eqs. \eqref{eq:triangle1}, \eqref{eq:triangle2}, \eqref{eq:triangle3}, \eqref{eq:triangle4} and triangle inequality, we get \begin{align}
\onenorm{ \varphi^{X_jR_jE_A\wht{B}\wht{X}_j}}{ \varphi^{X_jR_jE_A\wht{B}}\cdot \theta^{\wht{X}_j}_{X_j\wht{Y}_jC_j}} &\leq 2\sqrt{20\delta_3}+\sqrt{10\delta_3} \nonumber\\&\leq \sqrt{\delta_*}. &\mbox{(since $\delta_3=\delta_*^2/10^5$)} \nonumber
\end{align}
 Using eq. \eqref{eq:markov_generate_B}, we get
\begin{align}
\onenorm{ \varphi^{X_jR_jE_A\wht{B}\wht{X}_jB_j}}{ \varphi^{X_jR_jE_A\wht{B}}\cdot \theta^{\wht{X}_j}_{X_j\wht{Y}_jC_j} \cdot\varphi^{B_j}_{R_j\wht{X}_j\wht{B}}} \leq \sqrt{\delta_*}. \label{eq:using_markov_for_B}
\end{align}
From eq. \eqref{eq:tensor-lemma}, we have
      \begin{align}
        20\delta_3 &\geq \E_{r_jx_j\leftarrow \varphi^{R_jX_j}} \D(\varphi_{r_jx_j}^{\wht{X}_j\wht{B}}\|\varphi_{r_jx_j}^{\wht{X}_j}\otimes \varphi_{r_jx_j}^{\wht{B}}) \nonumber 
        \\ & \geq \E_{r_jx_j\wht{x}_j\leftarrow \varphi^{R_jX_j\wht{X}_j}} \D(\varphi_{r_jx_j\wht{x}_j}^{\wht{B}}\|\varphi_{r_jx_j}^{\wht{B}}).  &\mbox{(Fact \ref{fact:chain_rule_D})} \label{eq:hybrid_2_expansion}
      \end{align}
      Conditioning on $\wht{X}_j=*$, as before, we have \begin{align}
\E_{r_jx_j\leftarrow \varphi_*^{R_jX_j}} \D(\varphi_{*,r_jx_j}^{\wht{B}}\|\varphi_{r_jx_j}^{\wht{B}})  \leq 20\delta_3/0.99\delta_*\leq\delta_*.    
      \label{eq:hybrid_P_2}\end{align}
        We have \begin{align*}
             \delta_* &\geq \E_{r_jx_j\leftarrow \varphi_*^{R_jX_j}} \D(\varphi_{*,r_jx_j}^{\wht{B}}\|\varphi_{r_jx_j}^{\wht{B}}) 
             \\ & \geq \E_{r_jx_j\leftarrow \varphi_*^{R_jX_j}} \onenorm{\varphi_{*,r_jx_j}^{\wht{B}}}{\varphi_{r_jx_j}^{\wht{B}}}^2 &\mbox{(Fact \ref{fact:pinsker})}
             \\ & \geq \left(\E_{r_jx_j\leftarrow \varphi_*^{R_jX_j}} \onenorm{\varphi_{*,r_jx_j}^{\wht{B}}}{\varphi_{r_jx_j}^{\wht{B}}}\right)^2. &\mbox{(convexity)}
         \end{align*}
     This gives\begin{align}
   &\onenorm{\varphi_{*}^{R_jX_j\wht{B}}}{\varphi^{R_jX_j\wht{B}}} \nonumber\\& = \onenorm{\varphi_{*}^{R_jX_j}\cdot \varphi_{*,R_jX_j}^{\wht{B}}}{\varphi^{R_jX_j}\cdot \varphi_{R_jX_j}^{\wht{B}}} \nonumber
   \\& \leq \E_{r_jx_j\leftarrow \varphi_*^{R_jX_j}} \onenorm{\varphi_{*,r_jx_j}^{\wht{B}}}{\varphi_{r_jx_j}^{\wht{B}}} + \onenorm{\varphi^{R_jX_j}}{\varphi_*^{R_jX_j}} &\mbox{(triangle inequality)}\nonumber
   \\& \leq \sqrt{\delta_*}+2\sqrt{\delta_*}=3\sqrt{\delta_*}. &\mbox{(eq. \eqref{eq:R_jX_j_star})}\label{eq:closeness_in_D}
\end{align} 
 This proves eq. \eqref{eq:closeness_in_D_first}.
        Using eqs. \eqref{eq:closeness_in_D} and \eqref{eq:markov_in_phi}, we get 
        \begin{equation}   \label{eq:star_vs_nonstar}
\onenorm{\varphi_{*}^{E_AR_jX_j\wht{B}}}{\varphi^{E_AR_jX_j\wht{B}}} \leq3\sqrt{\delta_*}.  
         \end{equation}
         Combining eqs. \eqref{eq:using_markov_for_B} and \eqref{eq:star_vs_nonstar}, we get \begin{align*}
\onenorm{ \varphi^{X_jR_jE_A\wht{B}\wht{X}_jB_j}}{ \varphi_*^{X_jR_jE_A\wht{B}}\cdot \theta^{\wht{X}_j}_{X_j\wht{Y}_jC_j} \cdot\varphi^{B_j}_{R_j\wht{X}_j\wht{B}}}\leq 4\sqrt{\delta_*}. 
\end{align*}     
         This proves eq. \eqref{eq:closeness_hybrid_P1}.
\end{proof}
\section*{Proof of Claim \ref{claim:non-aug-to-aug}}
\begin{proof} 
Suppose we want to use our OT protocol as a subroutine in a larger protocol. {Since the quantum state of the outer protocol decoheres before the protocol for OT starts, we know that the residual state of the outer protocol is classical.} Let $R$ be the additional relevant register corresponding to the outer protocol (see Definition \ref{def:augmented-secure}).
Let the starting state of our OT protocol (including the register $R$) in case of cheating $\alice^*$ be ${\eta^*_{\mathcal{Q}}}^{S_0S_1\widehat{A}D\widehat{B}{R}}$ (see Table \ref{algo:ot_cheating_alice}). From Lemma \ref{lem:sim_alice_security}, we have \[ \onenorm {\tau^{\widehat{A} O^\Sim_AO^\Sim_BS_0S_1D}_{\textsc{Sim}}} { \tau_{\mathcal{Q}}^{\widehat{A}O_AO_{B}S_0S_1D}} \leq {\negl(\lambda)}.\]
Because of the decoherence in the protocol due to \textbf{DELAY}, for state $\kappa_{\mathcal{Q}}$  (Table \ref{algo:ot_cheating_alice}), we can write
 \[\kappa_\mathcal{Q}^{\widehat{\widehat{A}}D\widehat{B}R}=\sum_{\widehat{\widehat{a}}d\widehat{b}} \Pr(\widehat{\widehat{a}}d\widehat{b})_{\kappa_\mathcal{Q}}\ketbra{\widehat{\widehat{a}}d\widehat{b}}{\widehat{\widehat{a}}d\widehat{b}}\otimes
{\kappa_{\mathcal{Q}_{\widehat{\widehat{a}}d\widehat{b}}}^{R}}.\]
 In the simulated protocol, the corresponding state is
   \[\kappa_\mathcal{\Sim}^{\widehat{\widehat{A}}D\widehat{B}R}=\sum_{\widehat{\widehat{a}}d\widehat{b}} \Pr(\widehat{\widehat{a}}d\widehat{b})_{\kappa_\Sim}\ketbra{\widehat{\widehat{a}}d\widehat{b}}{\widehat{\widehat{a}}d\widehat{b}}\otimes
{\kappa_{\mathcal{\Sim}_{\widehat{\widehat{a}}d\widehat{b}}}^{R}}.\]
Note, since up to the  \textbf{DELAY}, no messages were passed between $\alice^*$ and $\bob$ ({other than in the test phase which the simulator can simulate locally} since $\bob$ is honest), and $\textsc{Sim}$ was following $\alice^*$ identically, we have, 
 \begin{equation}     \label{eq:kappa_equivalence}
\kappa_\mathcal{Q}^{\widehat{\widehat{A}}D\widehat{B}R}= \kappa_\mathcal{\Sim}^{\widehat{\widehat{A}}D\widehat{B}R}.
 \end{equation}
   Since $\bob$'s actions are honest, and in particular, since neither $\alice^*$ nor $\bob$ access $\widehat{B}R$,  we have (with a slight abuse of notation since $R$ is a quantum register),
$$(\widehat{A} O^\Sim_AO^\Sim_B S_0S_1D- \widehat{\widehat{A}}D - \widehat{B}R )^{\tau_\textsc{Sim}},$$
and 
$$(\widehat{A} O_AO_B S_0S_1D- \widehat{\widehat{A}}D - \widehat{B}R )^{\tau_\textsc{Q}}.$$
Since $\widehat{A}$ contains $\widehat{\widehat{A}}$, the above Markov chains combined with eq.~\eqref{eq:kappa_equivalence} and Fact \ref{fact:recovery} implies, {
$$ (\widehat{A} O_AO_BS_0S_1DR )^{\tau_\mathcal{Q}} \approx_{\negl(\lambda)}  (\widehat{A} O^\Sim_AO^\Sim_BS_0S_1DR)^{\tau_\textsc{Sim}}.$$}

\end{proof}

\section*{Proof of Fact \ref{fact:approx_markov_sampling}}
\begin{proof}
  We have  \begin{align*}
      \eps  &\geq\I(R:Y|X)_\rho 
      \\ & =\E_{x\leftarrow\rho^X}\D(\rho^{YR}_x\|\rho^{Y}_x\otimes\rho^{R}_x) &\mbox{(Definition \ref{def:mutual_info})}\\
      &\geq  \E_{x\leftarrow\rho^X} \onenorm{\rho^{YR}_x}{\rho^{Y}_x\otimes\rho^{R}_x}^2 &\mbox{(Fact \ref{fact:pinsker})}
     \\ &=  \E_{x\leftarrow\rho^X} \onenorm{\rho^{Y}_x \cdot \rho^{R}_{x,Y}}{\rho^{Y}_x\otimes\rho^{R}_x} ^2
     \\&\geq \left(\E_{x\leftarrow\rho^{X}} \onenorm{\rho^{Y}_x \cdot \rho^{R}_{x,Y}}{\rho^{Y}_x\otimes\rho^{R}_x} \right) ^2 &\mbox{(convexity)}
      \\&= \left(\E_{x,y\leftarrow\rho^{X,Y}} \onenorm{ \rho^{R}_{x,y}}{\rho^{R}_x} \right)^2.
      \\&= \left( \onenorm{ P_{XY}P_{R|XY}}{P_{XY}P_{R|X}} \right)^2
       \\&= \left( \onenorm{ P_{XYR}}{P_{XY}P_{R|X}} \right)^2.
  \end{align*}
This gives \begin{equation*}
     \onenorm{ P_{XYR}}{P_{XY}P_{R|X}} \leq \sqrt{\eps}.
  \end{equation*}
  Similarly, we get \begin{equation*}
      \onenorm{ P_{XYR}}{P_{XY}P_{R|Y}}  \leq \sqrt{\eps}.
  \end{equation*}
\end{proof}

\bibliography{name}
\bibliographystyle{alpha}
\end{document}